\documentclass[11pt]{article}

\usepackage{palatino}
\usepackage{pgfplots}
\usepackage{tikz}
\usetikzlibrary{shapes.multipart}
\usetikzlibrary{fit}
\usetikzlibrary{shapes.geometric,positioning}
\usetikzlibrary{calc}
\usetikzlibrary{patterns.meta}
\usepackage{caption}
\usepackage{subcaption}
\usepackage{url}
\usepackage{stfloats}
\usepackage{booktabs}
\usepackage{siunitx}
\usepackage{xr}
\usepackage[utf8]{inputenc}
\usepackage{mathtools}
\usepackage{graphicx}
\usepackage{algorithm}
\usepackage{cite}
\usepackage{amsmath,amssymb,amsfonts, amsthm}
\usepackage{algorithmic}
\usepackage{textcomp}
\usepackage{xcolor}
\usepackage{enumitem}

\newtheorem{lemma}{Lemma}
\newtheorem{theorem}{Theorem}

\newtheorem{remark}{Remark}

\newtheorem{corollary}{Corollary}

\newcommand\rk{\normalfont{\mbox{rk}}}

\newcommand\aaseq{\stackrel{\mbox{\tiny a.a.s.}}{=}}
\newcommand\aasgeq{\stackrel{\mbox{\tiny a.a.s.}}{\geq}}
\newcommand\aasg{\stackrel{\mbox{\tiny a.a.s.}}{>}}
\newcommand\aasleq{\stackrel{\mbox{\tiny a.a.s.}}{\leq}}
\newcommand\aasl{\stackrel{\mbox{\tiny a.a.s.}}{<}}

\newcommand{\FqV}{{\bf v}}
\newcommand{\FpV}{{\bf V}}
\newcommand{\Fqv}{v}
\newcommand{\Fpv}{V}
\newcommand{\Fqr}{r}
\newcommand{\Fpr}{R}
\newcommand{\Fqs}{\theta}
\newcommand{\Fqt}{t} 
\newcommand{\Fpt}{T}
\newcommand{\FpT}{{\bf T}}
\newcommand{\Fqh}{h} 
\newcommand{\Fph}{H} 
\newcommand{\FqH}{{\bf h}} 
\newcommand{\FpH}{{\bf H}}
\newcommand{\Fqf}{f} 
\newcommand{\Fpf}{F} 
\newcommand{\FqX}{{\bf x}} 
\newcommand{\FpX}{{\bf X}} 
\newcommand{\Fqx}{x} 
\newcommand{\Fpx}{X} 
\newcommand{\FqXL}{{\bf x}}

\newcommand{\FqWL}{{\bf w}}
\newcommand{\FqU}{{\bf u}}
\newcommand{\Insc}{\Upsilon}
\newcommand{\Un}{\Pi}
\newcommand{\Sp}{\Gamma}

\allowdisplaybreaks

\title{On the Generic Capacity of $K$-User Symmetric Linear Computation Broadcast}
\author{Yuhang Yao, Syed A. Jafar\\
{\small Center for Pervasive Communications and Computing (CPCC)}\\
{\small University of California Irvine, Irvine, CA 92697}\\
{\small \it Email: \{yuhangy5, syed\}@uci.edu}
}

\date{}     

\begin{document}
\maketitle

\begin{abstract}
Linear computation broadcast (LCBC) refers to a setting with $d$ dimensional data stored at a central server, where $K$ users, each with some prior  linear side-information, wish to compute various linear combinations of the data. For each computation instance, the data is represented as a $d$-dimensional vector with elements  in a finite field $\mathbb{F}_{p^n}$ where $p^n$ is a power of a prime. The computation is to be performed many times, and the goal is to determine the minimum amount of information per computation instance that  must be broadcast  to satisfy all the users. The reciprocal of the optimal broadcast cost per computation instance is the capacity of LCBC. The capacity is known for up to $K=3$ users. Since LCBC includes index coding as a special case, large $K$ settings of LCBC are at least as hard as the index coding problem. As such the general LCBC problem is beyond our reach and we do not pursue it. Instead of the general setting (\emph{all} cases), by focusing on the \emph{generic} setting  (\emph{almost all} cases) this work shows that the generic capacity of the symmetric LCBC (where every user has $m'$ dimensions of side-information and $m$ dimensions of demand) for large number of users ($K \geq d$ suffices) is $C_g=1/\Delta_g$, where  $\Delta_g=\min\left\{\max\{0,d-m'\}, \frac{dm}{m+m'}\right\}$, is the broadcast cost that is both  achievable and unbeatable asymptotically almost surely for large $n$, among all LCBC instances with the given parameters $p,K,d,m,m'$. Relative to baseline schemes of random coding or separate transmissions, $C_g$ shows an extremal gain by a factor of $K$ as a function of number of users, and by a factor of $\approx d/4$ as a function of data dimensions, when optimized over remaining parameters.  For arbitrary number of users, the generic capacity of the symmetric LCBC is characterized within a factor of $2$.\end{abstract}

\allowdisplaybreaks
\section{Introduction}
Recent observations of `megatrends' in the communication industry indicate that the number of devices connected to the internet is expected to cross 500 billion, approaching 60 times the estimated human population over the next decade \cite{Samsung_6G}. With machines set to become the dominant users of future communication networks, along with accompanying developments in artificial intelligence and virtual/augmented/mixed reality applications, a major paradigm shift is on the horizon where communication networks increasingly take on a new role, as \emph{computation networks}. The changing paradigm brings with it numerous challenges and opportunities. 

One of the distinguishing features of computation networks is their  algorithmic nature, which creates \emph{predictable} dependencies and side-information structures. To what extent can such structures be exploited for gains in communication efficiency? Answering this question requires an understanding of the capacity of computation networks.

The study of the capacity of computation networks has a rich history in information theory, spanning a  variety of ideas and directions that include zero error capacity and confusability graphs \cite{Shannon56}, graph entropy \cite{Korner_entropy, Witsenhausen}, conditional graph entropy \cite{Orlitsky_Roche}, multiterminal source coding \cite{Slepian_Wolf}, encoding of correlated sources \cite{Doshi_Shah_Medard_Effros, Feizi_Medard, Choi_Pradhan, Han_Costa}, sum-networks \cite{Korner_Marton_sum,Rai_Dey,Tripathy_Ramamoorthy},  computation over acyclic directed networks\cite{Huang_Tan_Yang_Guang,Guang_Yeung_Yang_Li}, compute-and-forward \cite{Nazer_Gastpar_Compute}, federated learning \cite{Federated_Learning}, private computation \cite{Sun_Jafar_PC,Heidarzadeh_Sprintson_Computation}, coded computing \cite{Cadambe_Grover_Tutorial,Yu_Lagrange,CodedComputing_Survey,CompCommTradeoff}, and distributed matrix multiplication \cite{Yu_Maddah-Ali_Avestimehr,Dutta_Fahim_Haddadpour, GPolyDot,  Jia_Jafar_SDMM}, to name a few. However, due to the enormous scope, hardness, and inherent combinatorial complexity of such problems, a \emph{cohesive} foundation is yet to emerge. 

Following the ground-up approach of classical network information theory which focuses on elemental scenarios, and taking cues from systems theory that builds on an elegant foundation of \emph{linear} systems, it is conceivable that a cohesive foundation could emerge from the study of the building blocks of \emph{linear computation networks}. Linear computation networks are characterized by the presence of side-information and demands that are linear functions of the data. Linear side information and dependencies are quite valuable as theoretical abstractions because in principle they allow the study of a complex linear computation network by breaking it down into tractable components, while retaining some of the critical relationships between the components in the form of side-information. For example, multi-round/multi-hop linear computation networks may be optimized one-round/hop at a time, by accounting for the information from other rounds/hops as side-information. 

\begin{figure}[!h]
\center
\def\colh{white}
\def\colw{white}
\begin{tikzpicture}[yscale=0.9, xscale=0.9]
\def\h{1.1}
\node at (-1.6,-1.5)  [draw=black, rounded corners, fill=\colw] (S)  { ${\bf S}$};
\node[draw=black, fill=\colh,  rounded corners,minimum width=3cm, rotate=90, above left=1.5cm and 0.3cm of S] (T)  [align=center]  {\footnotesize Server has data   $\FqXL\in\mathbb{F}_q^{d\times L}$};
\node[draw=black, fill=\colh, rounded corners] (U1) at (1,0*\h) [align=center] {\footnotesize User $1$ has $\FqWL_1'$\\[-0.1ex] \footnotesize $=\FqXL^T\FqV_1'\in\mathbb{F}_q^{L\times m_1'}$};
\node[draw=black, fill=\colh, rounded corners] (U2) at (1,-1.2*\h)  [align=center] {\footnotesize User $2$ has $\FqWL_2'$\\[-0.1ex] \footnotesize $=\FqXL^T\FqV_2'\in\mathbb{F}_q^{L\times m_2'}$};
\node[draw=black, fill=\colh, rounded corners] (U3) at (1,-3*\h)  [align=center] {\footnotesize User $K$ has $\FqWL_K'$\\[-0.1ex] \footnotesize $=\FqXL^T\FqV_K'\in\mathbb{F}_q^{L\times m_K'}$};
\node[draw=black, right=0.4cm of U1, align=center, minimum width=2.8cm,  rounded corners, fill=\colw](W1)  {\footnotesize Wants $\FqWL_1$\\[-0.1ex] \footnotesize $=\FqXL^T\FqV_1\in\mathbb{F}_q^{L\times m_1}$};
\node[draw=black, right=0.4cm of U2, align=center,  minimum width=2.8cm, rounded corners, fill=\colw](W2) {\footnotesize Wants $\FqWL_2$\\[-0.1ex] \footnotesize $=\FqXL^T\FqV_2\in\mathbb{F}_q^{L\times m_2}$};
\node[draw=black, right=0.4cm of U3, align=center, minimum width=2.8cm, rounded corners, fill=\colw](W3)  {\footnotesize Wants $\FqWL_K$\\[-0.1ex] \footnotesize $=\FqXL^T\FqV_K\in\mathbb{F}_q^{L\times m_K}$};
\node[below=-0.2cm of W2] {\large $\vdots$};
\node[below=-0.2cm of U2] {\large $\vdots$};
\draw[thick, ->] (S.west-|T.south) -- (S.west);
\draw[thick, ->] (U1) -- (W1);
\draw[thick, ->] (U2) -- (W2);
\draw[thick, ->] (U3) -- (W3);
\draw[thick, ->] (S.east) -- (U1.west);
\draw[thick, ->] (S.east) -- (U2.west);
\draw[thick, ->] (S.east) -- (U3.west);
\end{tikzpicture}

\caption{\small LCBC$\big({\Bbb F}_q, \FqV_{[K]}, \FqV'_{[K]}  \big)$ with batch-size $L$. $q=p^n$ is a power of a prime. The coefficient matrices $\FqV_k\in\mathbb{F}_q^{d\times m_k}$, $\FqV_k'\in\mathbb{F}_q^{d\times m_k'}$ for all $k\in[K]$ specify the desired computations and side-informations, respectively. }\label{fig:LCBC}
\end{figure}
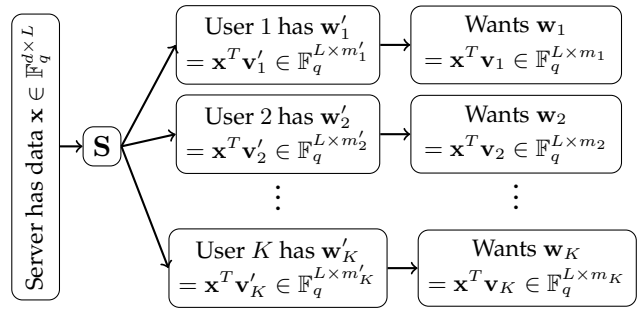

As a fundamental building block, the linear computation broadcast (LCBC) problem is introduced in \cite{Sun_Jafar_CBC}. LCBC refers to the setting illustrated in Figure \ref{fig:LCBC}, where $K$ users, each with some prior side-information $(\FqWL_k'=\FqX^T\FqV_k')$ comprised of various linear combinations of  $d$-dimensional data ($\FqX$) over a finite field ($\mathbb{F}_q = \mathbb{F}_{p^n}$), wish to retrieve other individually specialized linear combinations $(\FqWL_k=\FqX^T\FqV_k)$ of the data, with the help of a central server that has all the data. The goal is to determine the minimum amount of information that the central server must broadcast in order to satisfy all the users' computational demands. In addition to its significance as an elemental building block of computational networks, the LCBC setting is remarkably powerful by itself, e.g., it includes index coding \cite{Birk_Kol_SI,Maleki_Cadambe_Jafar,Young_Han_FnT} as a special case, and generalizes linear coded caching  \cite{Yu_Maddah_Avestimehr_exact, Maddah_Ali_Niesen, Wan_2020} to allow arbitrary cached information and demands. The one-to-many topology represented by LCBC arises naturally in any context where distributed nodes coordinate with each other \cite{Cuff_coordination,Fei_Chen_Wang_Jafar} with the help of a master node. Such scenarios may be pervasive in the future as interactive networked VR environments \cite{Networked_VR} become commonplace.

The capacity of LCBC is characterized for $K=2$ users in \cite{Sun_Jafar_CBC}. More recently, in \cite{Yao_Jafar_3LCBC} the capacity is fully characterized for the $K=3$ user LCBC. In addition to such efforts that are aimed at small number of users, it is also important to develop  insights into the fundamental limits of larger LCBC networks. However, any such attempt runs into immediate obstacles. In addition to the combinatorial complexity of large networks, the LCBC --- because it includes index coding as a special case --- is at least as hard as the index coding problem in general. The difficulty of the index coding problem is well recognized \cite{Effros_Rouayheb_Langberg, Lubetzky_Nonlinear, Young_Han_FnT, Rouayheb_Sprintson_Georghiades, Sun_Jafar_nonshannon}. How to overcome this obstacle, is the central question that motivates our work in this paper. 

A key idea that makes this work possible is the distinction between the \emph{general} LCBC problem --- which includes \emph{all} instances, and the \emph{generic} LCBC problem --- which includes \emph{almost all} instances. We focus on the latter. While the general LCBC problem is necessarily at least as hard as the general index coding problem, the generic LCBC problem may still be tractable. Such observations are common in many fields, for example computational complexity theory posits that for many computation problems, the difficult cases are rare and most (generic) instances are much easier, thereby motivating the sub-field of \emph{generic-case complexity} \cite{Generic1,Generic2}. Drawing parallels to the degrees of freedom (DoF) studies of wireless interference networks, there also the general problem remains open --- for arbitrary channel realizations the DoF are not known for even the $3$-user interference channel. However, the generic problem is settled for the $K$-user (any $K$) interference channel; we know the DoF for \emph{almost all} realizations \cite{Cadambe_Jafar_int,  Motahari_Gharan_Khandani, Gul_Stotz_Jafar_Bolcskei_Shamai}. For general MIMO interference channels, even maximizing linearly achievable DoF is shown to be NP-hard \cite{Razaviyayn_Sanjabi_Luo}, yet  it is tractable in the generic sense \cite{Wang_Gou_Jafar_subspace,Genie_chains}. Similarly, while the index coding problem is hard, index coding instances represent a negligible fraction of all possible instances of LCBC. Thus, there remains hope that a foundation for a cohesive theory of linear computation networks may yet be built by studying the \emph{generic capacity} of its building blocks.

With some oversimplification for the sake of intuition, consider the following toy example. We have a $K=4$ user setting, say over $\mathbb{F}_7$, with $d=4$ dimensional data  represented by ${\bf x}=({\bf A},{\bf B},{\bf C}, {\bf D})^T$. The users each have $1$-dimensional side-information  and demands,
{\small
\begin{align}
{\bf w}_1'&={\bf A}+{\bf B}+{\bf C}+{\bf D}, && {\bf w}_1={\bf A}+2{\bf B}+3{\bf C}+4{\bf D}, \\
{\bf w}_2'&={\bf A}+3{\bf B}+2{\bf C}+5{\bf D},&& {\bf w}_2=2{\bf A}+{\bf B}+4{\bf C}+6{\bf D}, \\
{\bf w}_3'&=5{\bf A}+4{\bf B}+{\bf C}+3{\bf D},&& {\bf w}_3=6{\bf A}+3{\bf B}+4{\bf C}+{\bf D},\\
{\bf w}_4'&=4{\bf A}+{\bf B}+5{\bf C}+6{\bf D},&& {\bf w}_4=5{\bf A}+2{\bf B}+6{\bf C}+3{\bf D}.
\end{align}
}If we had only the first $2$ users to consider, the broadcast cost of $2$ would be trivially achieved, e.g., by broadcasting ${\bf w}_1, {\bf w}_2$ which satisfies both users. If we had only the first $3$ users, the solution is less trivial, but we still find (see Section \ref{sec:onedim}) that broadcasting $({\bf S}_1, {\bf S}_2)=(2{\bf A}+6{\bf B}+3{\bf C}, ~4{\bf A}+4{\bf B}+{\bf C}+{\bf D})$ incurs a cost of $2$, while satisfying all $3$ users' demands -- it is easy to verify that User $1$ recovers ${\bf w}_1=2{\bf S}_1-3{\bf w}_1'$, User $2$ recovers ${\bf w}_2=5{\bf S}_2-4{\bf w}_2'$, and User $3$ recovers ${\bf w}_3={\bf S}_1+{\bf S}_2$, all operations in $\mathbb{F}_7$, represented as integers modulo $7$. However, as the number of users increases, the problem becomes much more challenging. It is far from obvious that a broadcast cost of $2$ could still suffice to satisfy all $4$ users listed above, and highly counter-intuitive that the optimal broadcast cost may still be only $2$ for large number of users, e.g., $K=100$ users. This surprising conclusion follows from the results found in this work, with the important caveat that the results are shown to be true only  asymptotically almost surely for large $n$. In other words, for this example, suppose we have the $4$ dimensional data $({\bf A}, {\bf B}, {\bf C}, {\bf D})$ over $\mathbb{F}_{q}$, $q=p^n$ for any arbitrary prime $p$, and a large number (say $K=100$) of users, and the coefficients of the users' $1$-dimensional demands and side-informations are chosen uniformly randomly from $\mathbb{F}_{p^n}$, each choice representing a particular instance of this LCBC. Then we prove that as $n\rightarrow\infty$, almost all instances have optimal broadcast cost $2$ (in $q$-ary units). The larger the number of users $K$, the larger $n$ may need to be for the convergence to take effect, but the optimal  broadcast cost must ultimately converge in probability to $2$ $q$-ary symbols. 

The main result of this work is the characterization of the \emph{generic} capacity, $C_g(p, K, d, m, m')=1/\Delta_g$, where  $\Delta_g=\min\left\{\max\{0,d-m'\}, \frac{dm}{m+m'}\right\}$, for a $K$ user  LCBC with $d$ dimensional data over $\mathbb{F}_{p^n}$, in the symmetric setting where every user has $m$ dimensional demands and $m'$ dimensional side-information, for large enough number of users ($K \geq d$ suffices) and large $n$. Informally, $\Delta_g$ represents a broadcast cost that is both achievable, and  unbeatable, asymptotically almost surely for large $n$, among the class of all  LCBC problems with  fixed parameters $K, p, d, m, m'$. Setting aside the trivial regimes $d\leq m+m'$ where random coding is optimal $(\Delta_g=\max\{0,d-m'\})$, and $d\geq K(m+m')$ where separate transmissions for each user are optimal $(\Delta_g=Km)$, in the remaining non-trivial regime where $m+m'<d<K(m+m')$, we have $\Delta_g=dm/(m+m')$. Note that this depends only on $d, m,m'$, i.e., the dimensions of the data, demands, and side-information. In particular, our generic capacity results do not depend on the characteristic $p$ of the finite field, and in the non-trivial regime with sufficiently large number of users, the generic capacity also does not depend on the number of users $K$. The converse proofs are information theoretic and make use of functional submodularity \cite{Tao_FS, Kontoyiannis_Madiman, Yao_Jafar_3LCBC}. The achievability arguments build upon the idea of asymptotic interference alignment \cite{Cadambe_Jafar_int}, both by adapting it from the $K$ user wireless interference channel to the $K$ user LCBC context, and by a non-trivial strengthening of the original scheme involving an additional symbol-extension overhead that is needed to harness sufficient \emph{diversity} in the finite field setting.
 The generic capacity characterization reveals that the capacity can be significantly higher than what is achievable with the  baseline schemes of random coding and/or separate transmissions. For example, the extremal gain \cite{Chan_extremal} of generic capacity over baseline schemes, as a function of the number of users (maximized over the remaining parameters) is $K$, and the extremal gain as a function of the data dimension $d$ is $\approx d/4$ (Observation \ref{obs:exgain} in Section \ref{sec:obs}). As an immediate corollary of the main result, the generic capacity of the symmetric LCBC is characterized  within a factor of $2$ for any number of users $K$ (Observation \ref{obs:obs1}). Notably, an exact characterization is found (Theorem \ref{thm:one_dim})  for any number of users if the side-information and demands are one-dimensional. Some extensions to asymmetric settings are obtained as well.

\section{Notation}\label{sec:notation}
\subsection{Miscellaneous}
 The notation $[a:b]$ represents the set of integers $\{a,a+1,\hdots, b\}$ if $b\geq a$ and $\emptyset$ otherwise. The compact notation $[K]$ is equivalent to $[1:K]$. For a set of indices $S$, the  notation $A_{S}$ represents $\{A_s, s\in S\}$, e.g., $A_{[K]}=\{A_1,A_2\cdots, A_K\}$. $|\mathcal{S}|$ denotes the cardinality of a set $\mathcal{S}$. $\mathbb{F}_q=\mathbb{F}_{p^n}$ is a finite field with $q=p^n$ a power of a prime. The elements of the prime field $\mathbb{F}_p$ are represented as $\mathbb{Z}/p\mathbb{Z}$, i.e.,  integers modulo $p$. The notation $\mathbb{F}_q^{n_1\times n_2}$ represents the set of $n_1\times n_2$ matrices with elements in $\mathbb{F}_q$. ${\Bbb F}_{q}$ is a sub-field of ${\Bbb F}_{q^z}$, and ${\Bbb F}_{q^z}$ is an extension field of ${\Bbb F}_q$ for $z>1$. $\mathbb{N}=\{1,2,\cdots\}$ is the set of natural numbers.  The greatest common divisor of $a,b$ is denoted $\mbox{gcd}(a,b)$. $(x)^+\triangleq \max\{0,x\}$. $\mathsf{Pr}(E)$ stands for the probability of the event $E$. Given an event $E_n$ that depends on an integer parameter $n$, we say that Event $E_n$ holds asymptotically almost surely (a.a.s.) if $\lim_{n\rightarrow\infty}\mathsf{Pr}(E_n)=1$. Throughout this work when we  use a.a.s., the quantity approaching infinity will be denoted by $n$.  For variables $a,b$ that depend on an integer $n$, we use the notation $a\aaseq b$ to represent the statement, $\lim_{n\rightarrow\infty}\mathsf{Pr}(a=b)=1$. Similarly, $a\aasgeq b$ represents $\lim_{n\rightarrow\infty}\mathsf{Pr}(a\geq b)=1$; $a\aasleq b$ represents $\lim_{n\rightarrow\infty}\mathsf{Pr}(a\leq b)=1$; $a\aasl b$ represents $\lim_{n\rightarrow\infty}\mathsf{Pr}(a< b)=1$, and $a\aasg b$ represents $\lim_{n\rightarrow\infty}\mathsf{Pr}(a> b)=1$.

\subsection{Matrix operations}
By default we will consider matrices in a finite field $\mathbb{F}_q$. 
For two matrices $M_1,M_2$ with the same number of rows,  $[M_1, M_2]$ represents a concatenated matrix which can be partitioned column-wise into $M_1$ and $M_2$. $M_{[i]}$ denotes the $i$-th column of $M$. The notation $M_{[a:b]}$ stands for the sub-matrix $[M_{[a]},M_{[a+1]},...,M_{[b]}]$ if $b\geq a$, and $[~]$ otherwise. The rank of $M\in\mathbb{F}_q^{m\times n}$ is denoted by $\rk(M)$, and we say that $M$ has \textit{full rank} if and only if $\rk(M) = \min\{m,n\}$. $\langle M\rangle_q$ denotes the $\mathbb{F}_q$-linear vector space spanned by the columns of $M$. The subscript $q$ will typically be suppressed as it is clear from the context. If  $M$ has full column rank, then we say that $M$ forms a basis for $\langle M\rangle$. The notation $M_1\cap M_2$ represents a matrix whose columns form a basis of $\langle M_1\rangle\cap\langle M_2\rangle$. In addition, ${\bf 0}^{a\times b}$ represents the $a\times b$ zero matrix. ${\bf I}^{a\times a}$ represents the $a\times a$ identity matrix.

\subsection{Conditional matrix notation: $(M_1|M_2)$}\label{sec:condmat}
Say $M_1 \in {\Bbb F}_q^{d\times \mu_1}$ and $M_2 \in {\Bbb F}_q^{d\times \mu_2}$. By Steinitz Exchange Lemma, there exists a sub-matrix of $M_1$ with full column rank, denoted by $(M_1|M_2)$, that is comprised of $\rk(M_1)-\rk(M_1\cap M_2)$ columns of $M_1$ such that $[M_1\cap M_2, (M_1 | M_2)]$ forms a basis for $\langle M_1\rangle$. We have,
\begin{align}
	\rk(M_1|M_2) &= \rk(M_1) - \rk(M_1\cap M_2) = \rk([M_1,M_2]) - \rk(M_2)
\end{align}
where we made the use of the fact that $\rk([M_1,M_2]) = \rk(M_1)+\rk(M_2)-\rk(M_1\cap M_2)$. 

\section{Problem Formulation: Linear Computation Broadcast} \label{sec: Problem Formulation}\label{sec:probstat}
\subsection{LCBC$\big({\Bbb F}_q, \FqV_{[K]}, \FqV_{[K]}' \big)$}
An LCBC problem is specified by its parameters as LCBC$\big({\Bbb F}_q,  \FqV_{[K]}, \FqV_{[K]}' \big)$, where $\mathbb{F}_q$ is a finite field with $q=p^n$ a power of a prime,  and $\FqV_k\in\mathbb{F}_q^{d\times m_k}$, $\FqV_k'\in\mathbb{F}_q^{d\times m_k'}$, for all $k\in[K]$, are matrices with the same number of rows, $d$. The value $K$ represents the number of users, $d$ represents the data dimension, and $m_k, m_k'$ quantify the amounts of desired computations and side-information corresponding to User $k$. The context is as follows. A central server stores multiple instances of $d$ dimensional data over a finite field $\mathbb{F}_q$. The $\ell^{th}$ instance of the data vector is denoted as $\FqX(\ell)=[\Fqx_1(\ell),\hdots,\Fqx_d(\ell)]^T\in \mathbb{F}_q^{d\times 1}$, and $\FqXL = [\FqX(1),...,\FqX(L)] \in {\Bbb F}_q^{d\times L}$ collects $L\in {\Bbb N}$ data instances.\footnote{The parameter $L$ is referred to as the batch size and may be chosen freely by a coding scheme.} A broadcast link connects $K$ distributed users to the central server.  The \emph{coefficient} matrices $\FqV_k'$ and $\FqV_k$ specify the side-information and desired computations for the $k^{th}$ user. Specifically, for all $k\in[K]$, User $k$ has side information $\FqWL_k'= \FqXL^T\FqV_k'\in\mathbb{F}_q^{L\times m_k'}$, and wishes to compute  $\FqWL_k=\FqXL^T\FqV_k\in\mathbb{F}_q^{L\times m_k}$. For compact notation in the sequel it is useful to define, 
\begin{align} \label{eq:def_uk}
\FqU_k \triangleq [\FqV_k', \FqV_k].
\end{align}

A coding scheme for an LCBC problem is denoted by a tuple $(L, \Phi, \Psi_{[K]},\mathcal{S})$, which specifies a batch size $L$, an encoding function $\Phi: {\Bbb F}_q^{L\times d}\rightarrow\mathcal{S}$ that maps the data to the broadcast information ${\bf S}$ over some alphabet $\mathcal{S}$, i.e., 
\begin{align}
	\Phi(\FqXL)={\bf S}
\end{align}
and decoders, $\Psi_{k}: \mathcal{S}\times \mathbb{F}_q^{L\times m_k'}\rightarrow\mathbb{F}_q^{L\times m_k}$,  
that allow the $k^{th}$ user to retrieve $\FqWL_k$ from the broadcast information ${\bf S}$ and the side-information $\FqWL_k'$ for all $k\in[K]$, i.e.,
\begin{align} \label{eq:decoding_constraint}
	\FqWL_k = \Psi_k({\bf S}, \FqWL_k') = \Psi_k(\Phi(\FqXL), \FqWL_k') , ~~~~~~\forall k\in[K].
\end{align}
A coding scheme that allows successful decoding for all data realizations, i.e., satisfies \eqref{eq:decoding_constraint} for all $\FqXL\in\mathbb{F}_q^{d\times L}$, is called an \emph{achievable}  scheme. Let us define $\mathcal{A}\Big({\Bbb F}_q, \FqV_{[K]}, \FqV_{[K]}' \Big)$ as the set of all achievable  schemes for LCBC$\Big({\Bbb F}_q, \FqV_{[K]}, \FqV_{[K]}' \Big)$.

The broadcast cost (normalized by $L$ and measured in $q$-ary units) for an achievable scheme  is defined as $\Delta= \log_q|\mathcal{S}|/L$. The optimal broadcast cost $\Delta^*\big({\Bbb F}_q, \FqV_{[K]}, \FqV_{[K]}' \big)$ for an LCBC problem is defined as, 
\begin{align} \label{eq:def_communication_cost}
	\Delta^*\big({\Bbb F}_q, \FqV_{[K]}, \FqV_{[K]}' \big) = \inf_{(L, \Phi, \Psi_{[K]},\mathcal{S})\in\mathcal{A}\big({\Bbb F}_q, \FqV_{[K]}, \FqV_{[K]}' \big)} \Delta.
\end{align}
The capacity, $C^*$, of an LCBC problem is the reciprocal of its optimal broadcast cost,
\begin{align}
C^*\big({\Bbb F}_q, \FqV_{[K]}, \FqV_{[K]}' \big) = 1/\Delta^*\big({\Bbb F}_q, \FqV_{[K]}, \FqV_{[K]}' \big).
\end{align}
Note that although the side information and demands are linear functions of the data, the achievable schemes, i.e., the encoding and decoding operations are not restricted to be linear.

\subsection{Generic Capacity}
Define
\begin{align}
&\mathcal{L}_n\big(p, K, d, m_{[K]}, m'_{[K]}\big) \notag \\
&=\left\{\mbox{LCBC}\big({\Bbb F}_q,  \FqV_{[K]}, \FqV_{[K]}' \big)\left| \begin{array}{rl}
q&=~~p^n\\
\FqV_k&\in~~\mathbb{F}_q^{d\times m_k}\\
\FqV_k'&\in~~\mathbb{F}_q^{d\times m_k'}\\
\forall k& \in~~[K]
\end{array}\right.\right\},
\end{align}
or $\mathcal{L}_n$ in short, as the set of all LCBC instances with the  `dimensional' parameters $p$, $n$, $K$, $m_{[K]}$, $m_{[K]}'$. Let $\Lambda_n$ be a uniformly randomly chosen instance from $\mathcal{L}_n$, and $\Delta^*(\Lambda_n)$  be the optimal download cost of $\Lambda_n$. In order to define generic capacity, let us fix the parameters $(p,K,d,m_{[K]},m'_{[K]})$ and allow $n$ to approach infinity.

The generic optimal broadcast cost $\Delta_g$, if exists, is defined as the value that $\Delta^*(\Lambda_n)$ converges to in probability, i.e.,
\begin{align}
	\lim_{n\to \infty} {\sf Pr}\Big( \big|\Delta^*(\Lambda_n) - \Delta_g \big| < \varepsilon \Big) = 1, ~~~~~~ \forall \varepsilon >0.
\end{align}
The generic capacity is then defined as the reciprocal, i.e.,
\begin{align}
	C_g = 1/\Delta_g.
\end{align}
Since $\Delta_g, C_g$ may not always exist, we further define the following upper and lower extremal metrics, which always exist and help in the analysis of generic capacity.
We say that $\Delta$ is achievable asymptotically almost surely (a.a.s)  (cf. Definition 1.1.2(v) \cite{tao2012topics}) if,
\begin{align}
\lim_{n\rightarrow\infty}\mathsf{Pr}\Big(\Delta^*(\Lambda_n)\leq \Delta\Big)= 1,\label{eq:limaasach}
\end{align}
which is expressed  compactly as $\Delta^*(\Lambda_n)\aasleq\Delta$.
Define the smallest such $\Delta$ as $\Delta_u^*\left(p,K,d,m_{[K]},m'_{[K]}\right)$, or $\Delta_u^*$ in short, i.e., 
\begin{align}
\Delta_u^*&\triangleq \inf\left\{\Delta: \Delta^*(\Lambda_n)\aasleq \Delta\right\}.\label{eq:infdeltau}
\end{align}
Similarly, define $\Delta_l^*$ as,
\begin{align}
\Delta_l^*&\triangleq \sup\left\{\Delta: \Delta^*(\Lambda_n)\aasg \Delta\right\}.\label{eq:defdeltal}
\end{align}
Thus, $\Delta_u^*$ is the infimum of broadcast costs that are achievable a.a.s.  (tightest upper bound),  and $\Delta_l^*$ is the supremum of broadcast costs that are \emph{not} achievable a.a.s.  (tightest  lower bound). By definition, $\Delta_u^*\geq \Delta_l^*$.  If\footnote{Note that the definition does not automatically preclude strict inequality, e.g., as a thought experiment suppose one half of all instances have $\Delta^*=1$ and the other half have $\Delta^*=2$, then $\Delta^*_u = 2$ and $\Delta^*_l=1$.}    $\Delta_u^*=\Delta_l^*$, then they are equal to the generic optimal broadcast cost $\Delta_g$, i.e., 
\begin{align}
\Delta_g=\Delta_u^*=\Delta_l^*
\end{align}
and $C_g=1/\Delta_g$ exists. If $\Delta_u^*\neq \Delta_l^*$, then $\Delta_g, C_g$ do not exist.

\begin{remark}
It is worth noting that the definition of `generic capacity' connects to the notion of  generic subsets in the literature on generic case complexity \cite{Generic}. To briefly point out this connection, following along the lines of Definition 3.1 and Lemma 3.2 of \cite{Generic}, we may define a generic subset as follows: Let $I$ be a set of inputs with size\footnote{A size function for a set $I$ is a map $\sigma: I\rightarrow\mathbb{N}$, the nonnegative integers, such that the preimage of each integer is finite (Definition 2.4 of \cite{Generic}).} function $\sigma$. Define $I_r$, the sphere of radius $r$, by $I_r=\{w\mid w\in I, \sigma(w)=r\}$, the set of inputs of size $r$. A subset $R\subset I$ is said to have asymptotic density $\alpha$, written $\rho'(R)=\alpha$, if $\lim_{r\rightarrow\infty}|R\cap I_r|/|I_r|=\alpha$ where $|X|$ denotes the size of a set $X$. If $R$ has asymptotic density $1$, it is called generic; and if it has asymptotic density $0$, it is negligible. Now, for our problem, the set  $I=\cup_{n\in\mathbb{N}}\mathcal{L}_n$ is the set of all LCBC instances for fixed $p,K,m_{[K]},m'_{[K]}$. The size function $\sigma=n$, and  the sphere $I_n=\mathcal{L}_n$. Let ${R}_u(\Delta)=\{\mathcal{L}\in I\mid \Delta\geq \Delta^*(\mathcal{L})\}$ be the subset of LCBC instances for which the broadcast cost $\Delta$ is achievable. Then $\Delta_u^*$ is the infimum of the values of $\Delta$ for which ${R}_u(\Delta)$ is generic. Similarly,  $\Delta_l^*$ is the supremum of the values of $\Delta$ for which ${R}_l(\Delta)=\{\mathcal{L}\in I\mid \Delta < \Delta^*(\mathcal{L})\}$ is generic. In plain words, $\Delta_u^*$ is the infimum of broadcast costs that are generically achievable, while $\Delta_l^*$ is the supremum of broadcast costs that are generically not achievable. When they match,  we have the generic optimal broadcast cost, and as its reciprocal notion, the generic capacity.
\end{remark}

\section{Results: Generic Capacity}
In this work we mainly focus on the \emph{symmetric} LCBC,  where we have,
\begin{align} \label{eq:symmetric_assumption}
	&m_1 = m_2 = ... = m_K = m, \notag \\
	&m_1' = m_2' = ... =m_K' = m'.
\end{align}
Note that the generic capacity (if it exists) can only be a function of $(p,K, d, m,m')$, since these are the only parameters left. 

\subsection{$K=1,2,3$ Users}
While we are interested primarily in LCBC settings with large number of users (large $K$), it is instructive to start with the generic capacity characterizations for $K=1,2,3$ users. Recall that the LCBC problem is already fully solved for $K=2$ in \cite{Sun_Jafar_CBC} and $K=3$ in \cite{Yao_Jafar_3LCBC}. Therefore, the following theorem essentially follows from \cite{Sun_Jafar_CBC, Yao_Jafar_3LCBC}. The $K=1$ case is trivial and is included for the sake of completeness.
\begin{theorem}\label{thm:GLCBC23} The generic capacity $C_g=1/\Delta_g$  for the symmetric LCBC with $K=1,2,3$ users  is characterized as follows.
\begin{align} \label{eq:gb_K2}
&K=1 \mbox{ user:} \notag \\
	 &~~~~\Delta_g = \begin{cases}
		0, & d \leq m'; \\
		d-m', & m' \leq d \leq m+m'; \\
		m, & m+m' \leq d.
	\end{cases}\\
&K=2 \mbox{ users:}\notag \\
	&~~~~\Delta_g= \begin{cases}
		0, & d \leq m'; \\
		d-m', & m' \leq d \leq m+m'; \\
		m, & m+m' \leq d \leq m+2m'; \\
		d-2m', & m+2m' \leq d \leq 2(m+m');\\
		2m,&2(m+m')\leq d.
	\end{cases}\\
&K=3 \mbox{ users:}\notag \\
	&~~~~\Delta_g =
	\begin{cases}
		0,  & d\leq m'; \\
		d-m', & m'\leq d \leq m'+m; \\
		m, & m'+m \leq d \leq m+1.5m'; \\
		d-1.5m', & m+1.5m' \leq d \leq 1.5(m+m'); \\
		1.5m, & 1.5(m+m') \leq d \leq 1.5m+2m'; \\
		d-2m',  & 1.5m+2m' \leq d \leq 2(m+m'); \\
		2m, & 2(m'+m) \leq d \leq 2m+3m'; \\
		d-3m', & 2m+3m' \leq d \leq 3(m+m');\\
		3m,&3(m+m')\leq d.
		\end{cases}
\end{align}
\end{theorem}
The proof of Theorem \ref{thm:GLCBC23} is relegated to Appendix \ref{app:proof_generic_K23}. The task left for the proof is to correctly account for the generic cases (non-trivial for $K=3$), after which the capacity results from \cite{Sun_Jafar_CBC, Yao_Jafar_3LCBC} can be directly applied.

\subsection{Large $K$}
The main result of this work appears in the following theorem.
\begin{theorem}\label{thm:large_K}
For the symmetric LCBC with the number of users satisfying
\begin{align}
K\geq d/\mbox{gcd}(d,m+m'),\label{eq:large_K}
\end{align}
the generic capacity $C_g=1/\Delta_g$ is characterized as follows,
\begin{align} \label{eq:Delta_large_K}
\Delta_g &= \begin{cases}
		0, & d \leq m'; \\
		d-m', & m' \leq d \leq m+m'; \\
		dm/(m+m'), & m+m' < d \leq K(m+m').
	\end{cases}
\end{align}
\end{theorem}
The proof of converse for Theorem \ref{thm:large_K} appears in Section \ref{sec:proof_converse_large_K} while the achievability  is proved in Appendix \ref{sec:proof_large_K}.  The main technical challenge is to show the achievability of $dm/(m+m')+\varepsilon, \forall \varepsilon >0$ in the  non-trivial regime, $m+m'<d\leq K(m+m')$. Remarkably, in this regime we are able to show that $\forall \varepsilon>0$, a broadcast cost of $dm/(m+m')+\varepsilon$ is achievable a.a.s., regardless of the number of users $K$, based on an asymptotic interference alignment (IA) scheme.  Examples to illustrate the asymptotic IA construction are provided in Section \ref{sec:ex}. The condition \eqref{eq:large_K} on the number of users in Theorem \ref{thm:large_K} is needed  for our converse bound in Section \ref{sec:proof_converse_large_K} to match the achievability. 

Next, let us briefly address arbitrary number of users and asymmetric settings through the following corollary of Theorem \ref{thm:large_K}. 
\begin{corollary}\label{cor:allK}
For the (not necessarily symmetric) LCBC with arbitrary number of users,
\begin{align}
\Delta_l^*&\geq \max_{\mathcal{K}\subset [K]}\min\left\{d-\sum_{k\in\mathcal{K}}m_k', ~\sum_{k\in\mathcal{K}}m_k\right\},\label{eq:genlbound}\\
\Delta_u^*&\leq \min\left\{\frac{m_{\max}d}{m_{\max}+m'_{\min}}, (d-m'_{\min})^+, \sum_{k\in[K]}m_k\right\},\label{eq:genubound}
\end{align}
where $m_{\max}\triangleq\max_{k\in[K]}m_k$, and $m_{\min}'\triangleq\min_{k\in[K]}m_k'$.
\end{corollary}
\proof
\eqref{eq:genlbound} follows from a simple cooperation bound. Consider any LCBC instance ${\Lambda}_n\in\mathcal{L}_n$. For any subset $\mathcal{K}\subset[K]$ let the users in $\mathcal{K}$ fully cooperate as one user, and eliminate all other users $[K]\setminus\mathcal{K}$, to obtain a single user LCBC instance $\Lambda_{n,\mathcal{K}}$ with optimal broadcast cost $\Delta^*(\Lambda_{n,\mathcal{K}})\leq \Delta^*(\Lambda_n)$. The combined user has demand coefficient matrix $\FqV_{\mathcal{K}}\triangleq \begin{bmatrix}\FqV_k, k\in\mathcal{K}\end{bmatrix}$ with $d$ rows and $m_{\mathcal{K}}\triangleq \sum_{k\in\mathcal{K}}m_k$ columns, and the side-information coefficient matrix $\FqV'_{\mathcal{K}}\triangleq \begin{bmatrix}\FqV_k', k\in\mathcal{K}\end{bmatrix}$ with $d$ rows and $m_{\mathcal{K}}'\triangleq \sum_{k\in\mathcal{K}}m_k'$ columns. From Theorem \ref{thm:GLCBC23}, based on the generic capacity for the single user case, we immediately obtain \eqref{eq:genlbound}. For \eqref{eq:genubound} note that $\Delta^*(\Lambda_n)\leq\sum_{k\in[K]}m_k$ because serving the users separately is always an option. $\Delta^*(\Lambda_n)\leq (d-m_{\min}')^+$ also holds because broadcasting $(d-m_{\min}')$ generic linear combinations (over a sufficiently large field extension) of the data allows each user a total number of generic equations $(d-m_{\min}')^++m_k'\geq d$, which suffices for each user to recover all  $d$ data dimensions. Finally, it is always possible to mimic a symmetric setting by adding superfluous demands and discarding some of the side-information at each user until every user has $m'_{\min}$ generic linear combinations of side-information and $m_{\max}$ generic linear combinations of demand. Note that if  $d\leq m_{\max}+m'_{\min}$ then $\frac{m_{\max}d}{m_{\max}+m'_{\min}}>(d-m'_{\min})^+$, and if $d>K(m_{\max}+m_{\min}')$, then $\frac{m_{\max}d}{m_{\max}+m'_{\min}}>\sum_{k\in[K]}m_k$. The only remaining case is $m_{\max}+m'_{\min}< d\leq K(m_{\max}+m_{\min}')$, in which case the achievability of $\frac{m_{\max}d}{m_{\max}+m'_{\min}}$ is shown in the proof of Theorem \ref{thm:large_K}.$\hfill\square$

\subsection{One dimensional case: $m=m'=1$}\label{sec:onedim}
In the special case where the side-information and demands are one-dimensional,   the generic capacity is characterized for \emph{any} number of users, as follows.
\begin{theorem}\label{thm:one_dim}
For LCBC with $m=m'=1$, the generic capacity $C_g=1/\Delta_g$  is characterized as follows.
\begin{enumerate}
	\item For even $d$,
		\begin{align}
		\Delta_g &= \begin{cases}
				d/2, & 2 \leq d \leq 2K;\\
				K,&  d \geq 2K.
			\end{cases}
		\end{align}
	\item For odd $d$,
		\begin{align}
		\Delta_g &= \begin{cases}
				0, & d = 1; \\
				d/2, & 3 \leq d < 2K-1;\\
				K-1,& d = 2K-1; \\
				K,& d > 2K-1.
			\end{cases}
		\end{align}
\end{enumerate}
\end{theorem}
The result for even $d$ follows directly from Theorem \ref{thm:large_K} and Corollary \ref{cor:allK}. Specifically, note that for even $d$ we have $d/\mbox{gcd}(d,2) = d/2$, and thus \eqref{eq:Delta_large_K} finds the generic capacity for even $d$ when $K\geq d/2$. Meanwhile, letting $\mathcal{K} = [K]$ in \eqref{eq:genlbound}, together with \eqref{eq:genubound} proves the capacity for $K\leq d/2$.

For cases with odd $d$, \eqref{eq:genlbound} and \eqref{eq:genubound} provide the capacity for $d=1$ and $d>2K-1$ (equivalently, $d\geq 2K$). For the remaining $2$ regimes, \eqref{eq:genubound} shows achievability for $1\leq d<2K-1$. \eqref{eq:genlbound} provides the converse for $d=2K-1$ by specifying $\mathcal{K} = [1:(d-1)/2]$, since $\min\{ d-(d-1)/2,  (d-1)/2\} = (d-1)/2 = K-1$. To complete the proof of Theorem \ref{thm:one_dim}, it remains to show the converse for $3\leq d < 2K-1$, and the achievability for $d=2K-1$. These proofs are provided in Section \ref{sec:proof_thm_one_dim}.

\subsection{Observations}\label{sec:obs}
\begin{enumerate}
\item  \label{obs:obs1} The following observations follow directly by specializing Corollary \ref{cor:allK} to the symmetric LCBC with arbitrary number of users $K$.
\begin{align}
	&~~~~~~~~~~~~~d \leq m' \implies \Delta_g=0;  \\
	&m' \leq d \leq m+m' \implies \Delta_g=d-m';  \\
	&1 < \frac{d}{m+m'} \leq K  \implies \mbox{if $\Delta_g$ exists, then} \notag \\ 
	&\max  \left\{ d-m'\left\lceil \frac{d}{(m+m')}\right\rceil, m\left\lfloor \frac{d}{m+m'} \right\rfloor \right\} \notag \\ 
	&~~~~~~~~~~~~~~\leq \Delta_g\leq \frac{md}{(m+m')}; \label{eq:factor2}\\
	&~~~K(m+m')\leq d \implies \Delta_g=Km.
\end{align}
Regarding \eqref{eq:factor2}, where $\Delta_g$ is not fully established, note that in this regime, if $(m+m')$ divides $d$, then $\Delta_g=md/(m+m')$ is settled. On the other hand, if $(m+m')$ does not divide $d$, and if $\Delta_g$ exists, then we have its value within a multiplicative factor of $2$ because in this regime,
\begin{align}
 \frac{md/(m+m')}{m\lfloor d/(m+m') \rfloor} = \frac{d/(m+m')}{\lfloor d/(m+m') \rfloor} \leq 2.
\end{align}
Thus, the generic capacity of the symmetric LCBC is characterized within a factor of $2$ when it exists, for arbitrary number of users $K$.

\item Theorem \ref{thm:large_K} and Corollary \ref{cor:allK} lead to sharp generic capacity results for various asymmetric cases as well. For example, from Corollary \ref{cor:allK} we find that for any number of users, if $m_k'\leq d\leq m_k+m_k'$ for all $k\in[K]$, then $\Delta_g=d-\min_{k\in[K]}m_k'$. If $d\geq \sum_{k\in[K]}(m_k+m_k')$ then $\Delta_g=\sum_{k\in[K]}m_k$. In Theorem \ref{thm:large_K} the optimal broadcast cost $\Delta_g=dm/(m+m')$ for the non-trivial regime $(m+m')<d\leq K(m+m')$, remains unchanged if we include another $K'$ users, say Users $K+1, K+2, \cdots, K+K'$, with asymmetric demands and side-information such that $m_{k'}\geq m'$ and $m_{k}\leq m$ for all $k'\in[K+1:K+K']$. The original converse still holds because  additional users cannot help. The original achievability still holds because the additional users have more side-information and less demands than original users, so they can pretend to be like the original users by discarding some of their side-information and adding superfluous demands. This creates a symmetric setting with $K+K'$ users, but in this regime $\Delta_g=dm/(m+m')$ does not depend on the number of users.

\item \label{obs:exgain} The generic capacity of the symmetric LCBC $C_g(p,K,d,m,m')$ can be quite significantly higher than the best of the rates achievable through the baseline schemes of random coding $(1/(d-m'))$ and separate transmissions $(1/(Km))$. Gains up to the order of $K$ and $d$ in the number of users and data dimensions are observed. To make this precise, consider the  \emph{extremal} gain \cite{Chan_extremal} of generic capacity over the baseline schemes as a function of the number of users, $K$, 
\begin{align}
	\eta_K&\triangleq \sup_{p,d,m,m'}\frac{C_g(p,K,d,m,m')}{\max\left\{\frac{1}{(d-m')^+},\frac{1}{Km}\right\}} \notag\\
	&= \sup_{d,m,m'}\frac{\frac{m+m'}{md}}{\max\left\{\frac{1}{(d-m')^+},\frac{1}{Km}\right\}}\notag\\
	&=K,\label{eq:etaK}
\end{align}
and as a function of the data dimension $d$, 
\begin{align}
	\eta_d&\triangleq \sup_{p,K,m,m'}\frac{C_g(p,K,d,m,m')}{\max\left\{\frac{1}{(d-m')^+},\frac{1}{Km}\right\}} \notag\\
	&= \sup_{K,m,m'}\frac{\frac{m+m'}{md}}{\max\left\{\frac{1}{(d-m')^+},\frac{1}{Km}\right\}}\notag\\
	&\in[d/4, d/4+1].\label{eq:etad}
\end{align}
To see \eqref{eq:etaK} note that it is trivial that $\eta_K\leq K$, whereas with $d=Km+m'$ and $m=1, m'\rightarrow\infty$, we have $\eta_K\geq \lim_{m'\rightarrow\infty}K(1+m')/(K+m')=K$. For \eqref{eq:etad}, setting $m=1, m'=\lfloor d/2\rfloor$ and $K>\lceil d/2\rceil$, we have $\eta_d\geq (1+\lfloor d/2 \rfloor)(d-\lfloor d/2\rfloor)/d \geq d/4$, whereas we also have $\eta_d\leq (1+m')(1-m'/d)\leq (d+1)^2/(4d)\leq d/4+1$. The strong gains  are indicative of the crucial role of interference alignment in the capacity of linear computation networks, especially when side-information is abundant.

\item In all cases where the question of \emph{existence} of generic capacity is settled, the answer is in the affirmative. However, it remains unknown whether this is always true, e.g., for all $p,d, K, m,m'$, must we have $\Delta_u^*=\Delta_l^*$? We conjecture that this is indeed the case.
\item In all cases where the generic capacity of the LCBC is known, it does not depend on the characteristic, $p$, i.e., for a fixed $K$, the generic optimal broadcast cost can be expressed as $\Delta_g(d,m,m')$. The capacity of the \emph{general} LCBC should depend on the characteristic, because there exist  examples of network coding problems where such dependence has been demonstrated, and there exists an equivalence between network coding and index coding, which in turn is a special case of the LCBC. However, it remains  unknown whether the \emph{generic} capacity of LCBC could depend on the characteristic $p$. 

\item The functional form of $\Delta_g(d,m,m')$ is plotted in Figure \ref{fig:GLCBCK3} for $K=2$, $K=3$ and for large $K$. While  $\Delta_g(d,m,m')$ characterizations are only defined for non-negative integer values of $d,m,m'$, the  functional form is shown as a continuous plot for simplicity. There exist three slightly different forms of the plot for $K=3$, depending on the relationship between $m$ and $m'$. The $K=3$ plot shown in Figure \ref{fig:GLCBCK3} assumes $m<m'$. While the lengths of the steps for $K=3$ are determined by the relative sizes of $m, m', d$, the plot always takes the shape of a staircase function with alternating horizontal (slope = 0) and slanted (slope = 1) edges.  The slope of the outer envelope  is $m/(m'+m)$.

\begin{figure*}[t]
\center
	\begin{tikzpicture}

    	\begin{axis}[
    	axis line style = thick,
      	width=1\textwidth,
      	height=0.35\textwidth,
      	xmin = -0.5, xmax = 22,
		ymin = -1, ymax = 8,
		axis lines = center,
		xlabel = $d$, ylabel = {$\Delta_g(d,m,m')$},
		xtick = \empty, ytick = \empty,
		legend style={at={(0.85,0.75)},anchor=north west},
    	]
      \addplot[color = blue, densely dotted, opacity=1, ultra thick] coordinates {(6, 2) (18.5, 6.16)};
      \addplot[color=black, thick] coordinates {
      	(0,0)
        (4,0)
        (6,2)
        (8,2)
        (9,3)
        (11,3)
        (12,4)
        (16,4)
        (18,6)
        (18.5,6)
      };
      \addplot[color = red!70, line width=0.7mm, opacity=1,   dashed] coordinates {
      	(0,0)
        (4,0)
        (6,2)
        (10,2)
        (12,4)
        (18,4)
        (18.5, 4)
      };
      \addplot[color = black!30, line width=2mm, opacity=0.4] coordinates {
      	(0,0)
        (4,0)
        (6,2)
        (18,2)
        (18.5,2)
              };
              
        \addplot[color = gray, opacity=1, loosely dotted, thick] coordinates {(0, 2) (6, 2)};
    	\addplot[color = gray, opacity=1, loosely dotted, thick] coordinates {(0, 3) (9, 3)};
    	\addplot[color = gray, opacity=1, loosely dotted, thick] coordinates {(0, 4) (12, 4)};
    	\addplot[color = gray, opacity=1, loosely dotted, thick] coordinates {(0, 6) (18, 6)};

    	\addplot[color = gray,  opacity=1, loosely dotted, thick] coordinates {(6, 0) (6, 2)};
    	\addplot[color = gray, opacity=1, loosely dotted, thick] coordinates {(8, 0) (8, 2)};
    	\addplot[color = gray, opacity=1, loosely dotted, thick] coordinates {(9, 0) (9, 3)};
    	\addplot[color = gray, opacity=1, loosely dotted, thick] coordinates {(11, 0) (11, 3)};
    	\addplot[color = gray, opacity=1, loosely dotted, thick] coordinates {(12, 0) (12, 4)};
    	\addplot[color = gray, opacity=1, loosely dotted, thick] coordinates {(16, 0) (16, 4)};
	
    	\addplot[color = gray, opacity=1,  loosely dotted, thick] coordinates {(18, 0) (18, 6)};
    	\addplot[color = gray, opacity=1, loosely dotted, thick] coordinates {(10, 0) (10, 2)};

	\addplot[color = black, loosely dotted, opacity=1, thick] coordinates {(0, 0) (6,2)};
    	
       \addlegendentry{$K\geq d$}
       \addlegendentry{$K=3$}
       \addlegendentry{$K=2$}
       \addlegendentry{$K=1$}
       legend style={at={(0.03,0.5)},anchor=west}
    \end{axis}
    \node at (0.1,0.2) {\small $0$};
    \node at (-0.2, 1.4) {\small $m$};
    \node at (-0.2, 1.9) {\small $1.5m$};
    \node at (-0.2, 2.4) {\small $2m$};
    \node at (-0.2, 3.3) {\small $3m$};
    \node[rotate=30] at (3, -0.2) {\footnotesize $m'$};
    \node[rotate=30] at (4.3, -0.2) {\footnotesize $m+m'$};
    \node[rotate=30] at (5.2, -0.2) {\footnotesize $m+1.5m'$};
    \node[rotate=30] at (6, -0.2) {\footnotesize $1.5(m+m')$};
    \node[rotate=30] at (6.7, -0.2) {\footnotesize $m+2m'$};
    \node[rotate=30] at (7.4, -0.2) {\footnotesize $1.5m+2m'$};
    \node[rotate=30] at (8.2, -0.2) {\footnotesize $2(m+m')$};
    \node[rotate=30] at (11, -0.2) {\footnotesize $2m+3m'$};
    \node[rotate=30] at (12.5, -0.2) {\footnotesize $3(m+m')$};
    \node at (9.6,3) {\color{blue} $\frac{md}{m+m'}$};

  \end{tikzpicture}
   \caption{Functional form of $\Delta_g$. The  large $K$ setting listed as $K\geq d$ for brevity, also allows more general $K$ as in \eqref{eq:large_K}.}
   \label{fig:GLCBCK3}
\end{figure*}
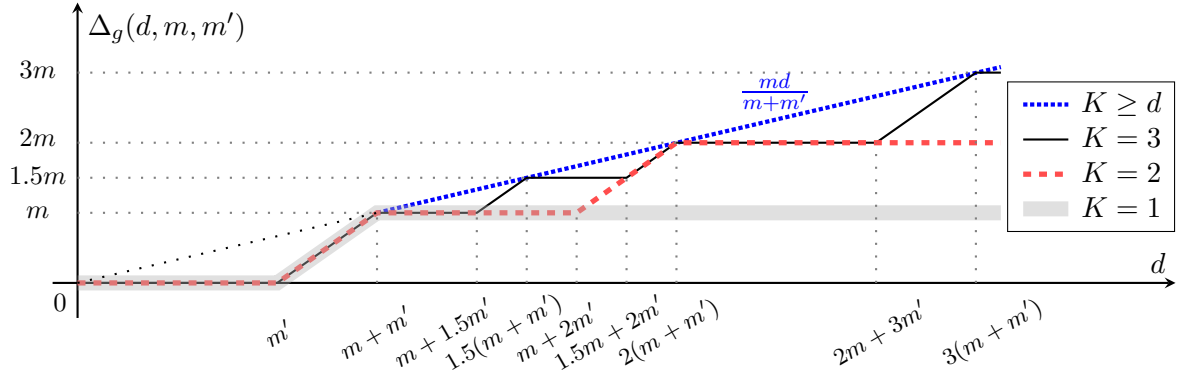

\item A remarkable \emph{scale-invariance} property is evident in $\Delta_g(d,m,m')$, in the sense that scaling $d,m,m'$ by the same constant results in  a  scaling of $\Delta_g(d,m,m')$ by the same constant as well. Specifically, 
\begin{align}
\Delta_g(\lambda d, \lambda m, \lambda m')=\lambda \Delta_g(d, m,m').
\end{align}
This is reminiscent of scale-invariance  noted in DoF studies of wireless networks \cite{Sun_Gou_Jafar}. 

\item The  initial (where $d\leq m+m'$) and final stages $(d\geq K(m+m'))$ represent somewhat trivial regimes that are the same for all $K$. In the remaining non-trivial regime, while $\Delta_g$ for $K=2$ and $K=3$ takes the shape of a slanted staircase function, for large number of users we obtain a smooth ramp function instead. A comparison of $K=2$ with $K=3$ suggests that the number of steps in the staircase increases with $K$, bringing the staircase closer to its upper linear envelope ($md/(m+m')$), until $K$ exceeds a threshold, beyond which the stairs disappear and $\Delta_g$ is equal to that linear envelope.

\item In the non-trivial regime $1<d/(m+m')<K$ for large $K$ (e.g., $K\geq d$) it is remarkable that $\Delta_g$ does not depend on $K$. In other words, once the number of users exceeds a threshold (e.g., $K\geq d$), additional users do not add to the generic broadcast cost of the LCBC. The achievable scheme in this parameter regime relies on linear asymptotic interference alignment (IA) \cite{Cadambe_Jafar_int} over a sub-field of $\mathbb{F}_q$, and while $\Delta_g$ remains unaffected by additional users, the cost of additional users may be reflected in the need for larger $n$ values to approach the same broadcast cost, as the number of users increases. As usual with asymptotic IA schemes, the achievable scheme is far from practical, and serves primarily to establish the fundamental limit of an asymptotic metric, in this case $\Delta_g$. What is possible with practical schemes, e.g., with limited $n$ and other coding-theoretic complexity constraints, remains an important open problem.

\item The generic capacity explored in this work is for LCBC instances over $\mathbb{F}_{p^n}$ where we allow large $n$ but $p$ is arbitrary. This formulation is appealing because the large $n$ limit allows $\mathbb{F}_{p^n}$ to be interpreted as high dimensional vector subspaces over  subfields of $\mathbb{F}_{p^n}$, e.g., $\mathbb{F}_p$. This facilitates  linear vector coding schemes,  allows dimensional analysis from vector space perspectives, and leads to new insights from linear interference alignment schemes, that may be broadly applicable. The alternative, where  $p$ is allowed to be large but $n$ is arbitrary (especially $n=1$) remains unexplored. By analogy with wireless DoF studies, the latter is somewhat  reminiscent of algebraic interference alignment schemes based on rational dimensions \cite{Motahari_Gharan_Khandani}, i.e., non-linear IA schemes.

\item Linear asymptotic IA  has been used previously for network coding problems, e.g., distributed storage exact repair \cite{Cadambe_Jafar_Maleki_Ramchandran_Suh}, and  $K$ user multiple unicast \cite{PBNA}, under the assumption of large `$q$'. Note that since $q=p^n$, a large-$q$ assumption is more general than a large-$n$ assumption, e.g., large-$q$ also allows $n=1$ with large $p$. So at first sight it may seem that our IA schemes that require large-$n$ are weaker than conventional asymptotic IA schemes that only require large-$q$. This interpretation however misses a crucial aspect of our construction, which is somewhat subtle but technically quite significant. Conventional (large-$q$) constructions of asymptotic IA schemes rely on a diagonal structure of underlying linear transformations (matrices), based on symbol extensions (batch processing), and most importantly require these diagonal matrices to have sufficient diversity, which is possible with \emph{time-varying} coefficients \cite{Cadambe_Jafar_int}. In fact, such constructions are also possible for LCBC if we allow time-varying demands and side-information, e.g., new coefficient matrices are drawn i.i.d. uniform for each $\ell\in[L]$. However, for the LCBC with fixed demands and side-information, i.e., fixed coefficient matrices $\FqV_k,\FqV_k'$, symbol extensions only give rise to diagonal matrices that are scaled versions of the identity matrix (consider large $p$ and $n=1$), i.e., they lack the diversity that is needed for linear asymptotic IA schemes. Our construction  works with fixed coefficient matrices, consistent with the original LCBC definition. In this regard, a key technical contribution of this work is to show that the large-$n$ assumption allows sufficient diversity for linear asymptotic IA. For this we modify the conventional asymptotic IA construction to include an additional overhead (see Remark \ref{rem:overhead} in Section \ref{sec:example1}), and then show that while this overhead has a negligible impact on $\Delta_g$, it gives us sufficient diversity a.a.s.

\end{enumerate}

\section{Examples}\label{sec:ex}
In this section let us present two examples to convey the main ideas of our asymptotic IA constructions, with somewhat simplified notation. The complete achievability proof for Theorem \ref{thm:large_K} appears later, in Appendix \ref{sec:proof_large_K}. The first example that is presented in Section \ref{sec:example1} is perhaps the smallest example where asymptotic alignment is needed. However, the proof in this limited case hides too many of the details that are needed in the general case,  so we provide a larger example in the following subsection which may be more useful in understanding the general proof.
\subsection{Example 1: $\left(p,K,d=4,m=1,m'=1\right)$} \label{sec:example1}
Let $L=1$. For $q=p^n$, we will interpret $\mathbb{F}_q$ as an $n$-dimensional vector space over $\mathbb{F}_p$, and design a linear scheme over $\mathbb{F}_p$. Accordingly, let us clarify the notation as follows.
\begin{enumerate}
	\item The elements of the data and coefficient matrices are chosen from ${\Bbb F}_{q} = {\Bbb F}_{p^n}$.
	\item The data $\FqX^T = [\Fqx_1,\Fqx_2,\Fqx_3,\Fqx_4]\in {\Bbb F}_q^{1\times 4}$, is  equivalently represented over  $\mathbb{F}_p$ as $\FpX^T = [\Fpx_1^T,\Fpx_2^T,$ $\Fpx_3^T,\Fpx_4^T]$ $\in\mathbb{F}_p^{1\times 4n}$, where $\Fpx_i\in\mathbb{F}_p^{n\times 1}$ is the $n\times 1$  vector representation of $\Fqx_i$ over $\mathbb{F}_p$.
	\item User $k$ has side information $\FqX^T\FqV_{k}'\in\mathbb{F}_q$ and wishes to compute $\FqX^T\FqV_{k}\in\mathbb{F}_q$, where the elements of $(\FqV_k')^T=[\Fqv_{k1}', \Fqv_{k2}', \Fqv_{k3}', \Fqv_{k4}']$,  $\FqV_k^{T}=[\Fqv_{k1}, \Fqv_{k2}, \Fqv_{k3}, \Fqv_{k4}]$, are drawn i.i.d. uniform in ${\Bbb F}_q$. Equivalently, over $\mathbb{F}_p$,  User $k$ has side information $\FpX^T\FpV_{k}'\in\mathbb{F}_p^{1\times n}$ and wishes to compute $\FpX^T\FpV_{k}\in\mathbb{F}_p^{1\times n}$, where $(\FpV_k')^T=\big[(\Fpv_{k1}')^{T}, (\Fpv_{k2}')^{T}, (\Fpv_{k3}')^{T}, (\Fpv_{k4}')^{T}\big]\in\mathbb{F}_p^{n\times 4n}$,  $(\FpV_k)^{T}=\big[(\Fpv_{k1})^T, (\Fpv_{k2})^T, (\Fpv_{k3})^T, (\Fpv_{k4})^T\big]\in\mathbb{F}_p^{n\times 4n}$ and $\FpV_{ki}', \FpV_{ki}$ are the $n\times n$ matrix representations in $\mathbb{F}_p$ of $\Fqv_{ki}'$ and $\Fqv_{ki}$, respectively.
\item Let $\Fqr$ be  uniformly randomly chosen in ${\Bbb F}_q$, and denote by $\Fpr\in\mathbb{F}_p^{n\times n}$ the matrix representation of $r$ in $\mathbb{F}_p$. 
\item Define the set of variables,
\begin{align}
\mathcal{V}&\triangleq \Big\{\Fqv_{ki}: k\in[K], i\in[4]\Big\}  \cup \Big\{ \Fqv_{ki}': k\in[K],i\in[4] \Big\}\cup\Big\{\Fqr\Big\},
\end{align}
and note that $|\mathcal{V}|=8K+1$.
\end{enumerate}
Our goal is to show that $\Delta_{u}^* \leq dm/(m+m')=2$.
For all $k\in[K]$ and for all $i\in[4]$, let us define $\Fqt_{ki}\in\mathbb{F}_q$ as,
\begin{align} \label{eq:def_t0}
	\Fqt_{ki} \triangleq \Fqv_{ki}'-\Fqv_{ki}\Fqr,
\end{align}
so that we have
\begin{align} \label{eq:def_vvt}
	\begin{bmatrix}
		\Fqv_{k1}' \\ \Fqv_{k2}' \\ \Fqv_{k3}' \\ \Fqv_{k4}'
	\end{bmatrix}
	=
	\begin{bmatrix}
		\Fqv_{k1} \\ \Fqv_{k2} \\ \Fqv_{k3} \\ \Fqv_{k4}
	\end{bmatrix}
	\Fqr
	+
	\begin{bmatrix}
		\Fqt_{k1} \\ \Fqt_{k2} \\ \Fqt_{k3} \\ \Fqt_{k4}
	\end{bmatrix}.
\end{align}
Let $\Fpt_{ki}\in\mathbb{F}_p^{n\times n}$ denote the  $n\times n$ matrix representation of $\Fqt_{ki}$ in $\mathbb{F}_p$, so that we have in $\mathbb{F}_p$,
\begin{align} \label{eq:def_VVT}
	\underbrace{\begin{bmatrix}
		\Fpv_{k1}' \\ \Fpv_{k2}' \\ \Fpv_{k3}' \\ \Fpv_{k4}'
	\end{bmatrix}}_{\FpV_k'\in\mathbb{F}_p^{4n\times n}}
	=
	\underbrace{\begin{bmatrix}
		\Fpv_{k1} \\ \Fpv_{k2} \\ \Fpv_{k3} \\ \Fpv_{k4}
	\end{bmatrix}}_{\FpV_k\in\mathbb{F}_p^{4n\times n}}
	\Fpr
	+
	\underbrace{\begin{bmatrix}
		\Fpt_{k1} \\ \Fpt_{k2} \\ \Fpt_{k3} \\ \Fpt_{k4}
	\end{bmatrix}}_{\FpT_k\in\mathbb{F}_p^{4n\times n}},
\end{align}
and $\FpT_k$ is defined as in \eqref{eq:def_VVT}.

Next let us construct a matrix, $\FpH\in\mathbb{F}_p^{n\times \eta}$, whose column-span over $\mathbb{F}_p$ is almost invariant under linear transformations $\FpV_{ki}$ and $\FpT_{ki}$ for all $k\in [K], i\in [4]$, i.e., $\langle \Fpv_{ki}\FpH \rangle_p \approx \langle \FpH \rangle_p$ and $\langle \FpT_{ki}\FpH \rangle_p \approx \langle \FpH \rangle_p$. In addition, we want $\eta/n\approx 1/2$, so that the columns of $\FpH$ span approximately half of the $n$-dimensional vector space. For this, we invoke the asymptotic IA scheme of \cite{Cadambe_Jafar_int}.  

For a natural number $N$ whose value (as a function of $n$) will be specified later, let us start first by constructing the vector $\FqH \in\mathbb{F}_q^{1\times \eta}$ as follows,
\begin{align} \label{eq:def_h_e1}
	\FqH &= \left[ \prod_{k=1}^K \prod_{i=1}^4 \Fqv_{ki}^{\alpha_{ki}}\Fqt_{ki}^{\beta_{ki}},~ \mbox{s.t.}~ 0 \leq \alpha_{ki}, \beta_{ki} \leq N-1 \right] \\
	&\triangleq (\Fqh_1,\Fqh_2, \ldots ,\Fqh_{\eta}),
\end{align}
and similarly define $\overline{\FqH}\in\mathbb{F}_q^{1\times \overline{\eta}}$ as follows,
\begin{align}
	\overline{\FqH} &= \left[ \prod_{k=1}^K \prod_{i=1}^4 \Fqv_{ki}^{\alpha_{ki}}\Fqt_{ki}^{\beta_{ki}},~ \mbox{s.t.}~ 0 \leq \alpha_{ki}, \beta_{ki} \leq N \right] \\
	&\triangleq (\overline{\Fqh}_1,\overline{\Fqh}_2,\ldots ,\overline{\Fqh}_{\overline{\eta}}).
\end{align}
Note that we have,
\begin{align}
	\eta = N^{8K}, ~~~\overline{\eta} = (N+1)^{8K}.
\end{align}
This construction ensures that the elements of $\Fqv_{ki}\FqH$ and $\Fqt_{ki}\FqH$ are contained among the elements of $\overline{\FqH}$ for all $k\in[K], i\in[4]$. Now let $\Fph_1, \Fph_2,...,\Fph_{\eta}\in\mathbb{F}_p^{n\times n}$ be the matrix representations in $\mathbb{F}_p$ of $\Fqh_1,\Fqh_2,...,\Fqh_{\eta}\in\mathbb{F}_q$, and $\overline{\Fph}_1,\overline{\Fph}_2,...,\overline{\Fph}_\eta\in\mathbb{F}_p^{n\times n}$ be the matrix representations in $\mathbb{F}_p$ of $\overline{\Fqh}_1,\overline{\Fqh}_2,...,\overline{\Fqh}_{\overline{\eta}}\in\mathbb{F}_q$. Define,
\begin{align}
	\FpH &= \big[\Fph_1{\bf 1},\Fph_2{\bf 1},\ldots,\Fph_{\eta}{\bf 1}\big]~~\in\mathbb{F}_p^{n\times\eta},
\end{align}
and
\begin{align}
	\overline{\FpH} &= \big[\overline{\Fph}_1{\bf 1},\overline{\Fph}_2{\bf 1},\ldots,\overline{\Fph}_{\overline{\eta}}{\bf 1}\big]~~\in\mathbb{F}_p^{n\times \overline\eta},
\end{align}
where ${\bf 1}$ denotes the $n\times 1$ vector of $1$'s.
By construction, the columns of $\Fpv_{ki}\FpH$ and $\Fpt_{ki}\FpH$ are subsets of the columns of $\overline{\FpH}$, which  implies that $\forall k\in[K],~ i\in[4]$,
\begin{align} \label{eq:H_subset_e1}
	\langle \Fpv_{ki} \FpH \rangle_p \subset \langle \overline{\FpH} \rangle_p,&&\langle \Fpt_{ki} \FpH \rangle_p \subset \langle \overline{\FpH} \rangle_p.
\end{align}
Consider the matrix $[\FpH, \Fpr\FpH]\in\mathbb{F}_p^{n\times 2\eta}$, and define the event $E_n$ as,
\begin{align} \label{eq:E_n_e1}
E_n&\triangleq \Big(\rk([\FpH, \Fpr\FpH])=2\eta\Big).
\end{align}
The next steps, \eqref{eq:gotozerobegin}-\eqref{eq:gotozero} show that   $E_n$ holds a.a.s., which will subsequently be essential to guarantee successful decoding by each user.

$[\FpH, \Fpr \FpH]\in\mathbb{F}_p^{n\times 2\eta}$ has full column rank if and only if for all ${\bf c}=[c_1,\ldots,c_\eta]^T, {\bf c}'=[c_1', \ldots, c_\eta']^T\in {\Bbb F}_p^{\eta \times 1}$ such that $\big[{\bf c}^T, {\bf c}^{'T}\big] \neq {\bf 0}^{1\times 2\eta}$,
{\small
	\begin{align}
		&{\bf 0}^{n\times 1}\notag\\
		&\neq [
		\FpH, \Fpr \FpH]\begin{bmatrix}{\bf c}\\{\bf c}'\end{bmatrix} \label{eq:gotozerobegin}\\ 
		& = \big[\Fph_1{\bf 1},\Fph_2{\bf 1},\ldots,\Fph_{\eta}{\bf 1}, \Fpr\Fph_1{\bf 1},\Fpr\Fph_2{\bf 1},\ldots,\Fpr\Fph_{\eta}{\bf 1} \big]\begin{bmatrix}{\bf c}\\{\bf c}'\end{bmatrix}\\
		&= \big(c_1\Fph_1{\bf 1}+c_2\Fph_2{\bf 1}+ \cdots +c_{\eta}\Fph_{\eta}{\bf 1}\big)\notag\\
		&~~~~+\big(c'_{1}\Fpr\Fph_1{\bf 1}+c'_2\Fpr\Fph_2{\bf 1} +\ldots+c'_{\eta}\Fpr\Fph_{\eta}{\bf 1}\big)\\
		&= \big(c_1\Fph_1+c_2\Fph_2+\ldots+c_{\eta}\Fph_{\eta}\notag\\
		&~~~~+c'_{1}\Fpr\Fph_1+c'_2\Fpr\Fph_2 +\ldots+c'_{\eta}\Fpr\Fph_{\eta}\big){\bf 1}\\
		&\triangleq \Fpf_{{\bf c},{\bf c}'}{\bf 1} 
	\end{align}
}where $\Fpf_{{\bf c},{\bf c}'}$ is an $n\times n$ matrix in $\mathbb{F}_p$, which has a scalar representation in $\mathbb{F}_q$ as,
	\begin{align}
		\Fqf_{{\bf c},{\bf c}'} &= \underbrace{c_1h_1+c_2h_2+...+c_{\eta}h_{\eta}}_{\eta}\notag \\
		&+\underbrace{c_1'rh_1+c_2'rh_2+...+c_{\eta}'rh_{\eta}}_{\eta} ~\in\mathbb{F}_q.
	\end{align}
	Note that since $\mathbb{F}_p$ is a sub-field of $\mathbb{F}_q$, the scalars $c_i, c_i'$ in $\mathbb{F}_p$, are also scalars  $c_i, c_i'$ in $\mathbb{F}_q$.
	Thus, $\Fpf_{{\bf c},{\bf c}'}{\bf 1}\in\mathbb{F}_p^{n\times 1}$ can be equivalently represented in $\mathbb{F}_q$ as the product of $\Fqf_{{\bf c},{\bf c}'}$ with the scalar representation in ${\Bbb F}_q$,  of ${\bf 1}$ (the all $1$ vector in $\mathbb{F}_p$). Since the ${\Bbb F}_q$ representation of ${\bf 1}$ is not $0$, we obtain that
	\begin{align}
		\Fpf_{{\bf c},{\bf c}'}{\bf 1}\not= {\bf 0}^{n\times 1} \Longleftrightarrow \Fqf_{{\bf c},{\bf c}'} \not= 0.
	\end{align}
	Therefore, $[\FpH, \Fpr \FpH]$ has full column rank if and only if,
	\begin{align}
		P \triangleq \prod_{[{\bf c}^T,{\bf c}^{'T}]\in {\Bbb F}_{p}^{1\times 2\eta}\backslash \{{\bf 0}\}}\Fqf_{{\bf c}, {\bf c}'} \not= 0.\label{eq:nonzero}
	\end{align}
To distinguish polynomials from polynomial functions, let us indicate polynomials with square parentheses around them. For example, $[\Fqf_{{\bf c},{\bf c}'}]\in\mathbb{F}_p[\mathcal{V}]$ is a polynomial in the indeterminate variables $\mathcal{V}$, with coefficients in $\mathbb{F}_p$, and $\Fqf_{{\bf c},{\bf c'}}(\mathcal{V}): \mathbb{F}_q^{|\mathcal{V}|}\rightarrow\mathbb{F}_q$ is a function that maps the variables $\mathcal{V}$, which take values in $\mathbb{F}_q$, to a scalar value in $\mathbb{F}_q$. Similarly, $[P]\in\mathbb{F}_p[\mathcal{V}]$ is a polynomial, whereas $P(\mathcal{V}):\mathbb{F}_q^{|\mathcal{V}|}\rightarrow\mathbb{F}_q$ is a function. The condition \eqref{eq:nonzero}, which is equivalent to the event $E_n$, says that a uniformly random evaluation of the function $P(\mathcal{V})$ produces a non-zero value. We will show that this is true a.a.s. in $n$.
	
	First let us show that $[\Fqf_{{\bf c},{\bf c}'}]\in\mathbb{F}_p[\mathcal{V}]$ is a non-zero polynomial for all $\big[{\bf c}^T,{\bf c}^{'T}\big] \in\mathbb{F}_p^{1\times 2\eta}\setminus{\bf 0}$. We consider two cases.
	\begin{enumerate}
		\item Case I: At least one of $c_1,c_2,...,c_{\eta}$ is not zero. Let us set $\Fqr = 0$, which implies $\Fqt_{ki}= \Fqv'_{ki}$ by \eqref{eq:def_vvt}. Meanwhile, $\Fqh_1,\Fqh_2,...,\Fqh_{\eta}$ are different monomials in the elements of $\Fqv_{ki}'$ and $\Fqv_{ki}$ due to \eqref{eq:def_h_e1}. Since different monomials are linearly independent, we have that $[\Fqf_{{\bf c},{\bf c}'}] = c_1\Fqh_1+c_2\Fqh_2+\ldots+c_{\eta}\Fqh_{\eta}$ is a non-zero polynomial.
		\item Case II: $c_1=c_2=...=c_{\eta}=0$ and thus at least one of $c_1',c_2',...,c_{\eta}' \not=0$. For this case, we have $[\Fqf_{{\bf c},{\bf c}'}] = \Fqr(c_1'\Fqh_1+c_2'\Fqh_2+...+c_{\eta}'\Fqh_{\eta})$, which is also a non-zero polynomial since it is a product of $\Fqr$ with a non-zero polynomial.
	\end{enumerate}
	Thus, $[\Fqf_{{\bf c},{\bf c}'}]\in\mathbb{F}_p[\mathcal{V}]$ is a non-zero polynomial. It has degree not more than $12K(N-1)+1$. Therefore, $[P]\in\mathbb{F}_p[\mathcal{V}]$ is a non-zero polynomial  with degree not more than $(p^{2\eta}-1)[12K(N-1)+1]$. By Schwartz-Zippel Lemma, when all the variables $\mathcal{V}$ are assigned i.i.d. uniformly chosen values  in ${\Bbb F}_q$,
	\begin{align}
		\mathsf{Pr}\Big(P(\mathcal{V}) \not= 0\Big) &\geq 1-\frac{(p^{2\eta}-1)[12K(N-1)+1]}{q}\\
		&=1-\frac{(p^{2\eta}-1)[12K(N-1)+1]}{p^n} \\
		&\geq 1-12K\frac{N}{p^{n-2\eta}} \\
		& \to 1 \label{eq:thusEn}
	\end{align}
	as $n\to \infty$ if $\lim_{n\to \infty}\frac{N}{p^{(n-2\eta)}} = 0$.

Now let us specify the value of $N$  as follows,
\begin{align} \label{eq:def_N}
	N = \left\lfloor \left( \frac{n-\sqrt{n}}{2} \right)^{1/(8K)} \right\rfloor,
\end{align}
from which it follows that,
\begin{align}
	N \leq \left( \frac{n-\sqrt{n}}{2} \right)^{1/(8K)} 
\end{align}
and
\begin{align}
	n-2\eta = n-2N^{8K} \geq \sqrt{n} .\label{eq:overhead}
\end{align}
Therefore,
\begin{align}
	\lim_{n\to \infty} \frac{N}{p^{n-2\eta}} \leq \lim_{n\to \infty}\frac{\left( \frac{n-\sqrt{n}}{2} \right)^{1/(8K)}}{p^{\sqrt{n}}} = 0\label{eq:gotozero}
\end{align}
and since $N\geq 0$, we have $\lim_{n\to \infty}\frac{N}{p^{n-2\eta}} = 0$. Thus, we have shown that $E_n$ holds a.a.s., i.e., $[\FpH, \Fpr\FpH]\in\mathbb{F}_p^{n\times 2\eta}$ has full column rank $2\eta$, a.a.s.

Now let ${\bf Z} = ({\bf I}^{n\times n}|[\FpH, \Fpr \FpH]) \in {\Bbb F}_p^{n\times (n-2\eta)}$, so that $[\FpH, \Fpr \FpH, {\bf Z}]\in\mathbb{F}_p^{n\times n}$ has full rank $n$.
Let the server broadcast ${\bf S}=({\bf S}_0, {\bf S}_1, {\bf S}_2,\cdots,{\bf S}_K)\in\mathbb{F}_p^{1\times(4\overline\eta+K(n-2\eta))}$, such that,
\begin{align}
	{\bf S}_0 = \FpX^T
	\begin{bmatrix}
		\overline{\FpH} & {\bf 0} & {\bf 0} & {\bf 0} \\
		{\bf 0} & \overline{\FpH} & {\bf 0} & {\bf 0} \\
		{\bf 0} & {\bf 0} & \overline{\FpH} & {\bf 0} \\
		{\bf 0} & {\bf 0} & {\bf 0} & \overline{\FpH}
	\end{bmatrix}~~\in\mathbb{F}_p^{1\times 4\overline{\eta}}
\end{align}
and for all $k\in[K]$,
\begin{align}
	{\bf S}_k =  \FpX^T \FpV_k{\bf Z}
	~~\in\mathbb{F}_p^{1\times(n-2\eta)}.
\end{align}

\begin{remark}\label{rem:overhead} From \eqref{eq:overhead} we note that $n\geq 2\eta+\sqrt{n}$. The $\sqrt{n}$ term represents an overhead that is not present in conventional asymptotic IA constructions. The overhead is evident in the separate transmissions, of the projections of the desired information along the columns of ${\bf Z}$, for each of the $K$ users, as in ${\bf S}_1, {\bf S}_2, \cdots, {\bf S}_K$. Digging deeper, this overhead is essential for our scheme to ensure  that $E_n$ holds a.a.s. Note that it is  because of this $\sqrt{n}$ overhead that we have $\lim_{n\to \infty}\frac{N}{p^{(n-2\eta)}} = 0$ in \eqref{eq:thusEn},\eqref{eq:gotozero}, because in the fraction $\frac{N}{p^{n-2\eta}}$, the numerator is sub-linear in $n$ (roughy $n^{1/8K}$), while the denominator is super-polynomial in $n$, roughly $(p^{\sqrt{n}}$). Fortunately, the extra broadcast cost of $K\sqrt{n}$ due to this overhead, is negligible compared to $n$ for large $n$, so it does not affect the asymptotic achievability. 
\end{remark}

The decoding process works as follows.
User $k$ is able to compute $\FpX^T\FpV_{k}'\FpH$ directly from its side information $\FpX^T\FpV_k'$. The user is able to compute $\FpX^T\FpV_k\FpH$ and $\FpX^T\FpT_k\FpH$ from ${\bf S}_0$, since,
\begin{align}
	\left\langle \FpV_k\FpH \right\rangle_p = \Bigg\langle \begin{bmatrix}\Fpv_{k1} \FpH\\ \Fpv_{k2} \FpH \\ \Fpv_{k3} \FpH\\ \Fpv_{k4} \FpH\end{bmatrix} \Bigg\rangle_p \subset 
		\Bigg\langle \begin{bmatrix}
		\overline{\FpH} & {\bf 0} & {\bf 0} & {\bf 0} \\
		{\bf 0} & \overline{\FpH} & {\bf 0} & {\bf 0} \\
		{\bf 0} & {\bf 0} & \overline{\FpH} & {\bf 0} \\
		{\bf 0} & {\bf 0} & {\bf 0} & \overline{\FpH}
	\end{bmatrix} \Bigg\rangle_p
\end{align}
and
\begin{align}
	\left\langle \FpT_k\FpH \right\rangle_p = \Bigg\langle \begin{bmatrix}\Fpt_{k1} \FpH\\ \Fpt_{k2}\FpH \\ \Fpt_{k3}\FpH \\ \Fpt_{k4}\FpH\end{bmatrix} \Bigg\rangle_p \subset 
		\Bigg\langle \begin{bmatrix}
		\overline{\FpH} & {\bf 0} & {\bf 0} & {\bf 0} \\
		{\bf 0} & \overline{\FpH} & {\bf 0} & {\bf 0} \\
		{\bf 0} & {\bf 0} & \overline{\FpH} & {\bf 0} \\
		{\bf 0} & {\bf 0} & {\bf 0} & \overline{\FpH}
	\end{bmatrix}\Bigg\rangle_p
\end{align}
due to \eqref{eq:H_subset_e1}. Thus, User $k$ is able to compute 
\begin{align}
\FpX^T\FpV_k\Fpr\FpH&=\FpX^T\FpV_k'\FpH-\FpX^T\FpT_k\FpH
\end{align}
according to \eqref{eq:def_VVT}. Together with ${\bf S}_k$, User $k$ thus obtains,
\begin{align}
[\FpX^T\FpV_k\FpH, ~\FpX^T\FpV_k\Fpr\FpH, ~ {\bf S}_k] &=\FpX^T\FpV_k[\FpH, \Fpr\FpH, {\bf Z}]
\end{align}
and since $[\FpH, \Fpr\FpH, {\bf Z}]\in\mathbb{F}_p^{n\times n}$ is invertible (has full rank) a.a.s., User $k$ is able to retrieve its desired computation, $\FpX^T\FpV_k\in\mathbb{F}_p^{1\times n}$ a.a.s.

For $q=p^n$, the cost of broadcasting each $p$-ary symbol is $1/n$ in $q$-ary units. Thus, the broadcast cost of this scheme is,
\begin{align}
	\Delta_n = \frac{4\overline{\eta}+K(n-2\eta)}{n}.\label{eq:deltan}
\end{align}
The next few steps \eqref{eq:limdeltanbegin}-\eqref{eq:limdeltanend} show that $\lim_{n\rightarrow\infty}\Delta_n= 2$.

By \eqref{eq:def_N}, we have that,
\begin{align}
	\eta = N^{8K} \leq \frac{n-\sqrt{n}}{2} \leq (N+1)^{8K} = \overline{\eta}\label{eq:limdeltanbegin}
\end{align}
which implies that
\begin{align}
	\lim_{n\to \infty}\frac{\eta}{n} =\lim_{n\to \infty} \frac{N^{8K}}{n} \leq \lim_{n\to \infty}\frac{n-\sqrt{n}}{2n} = \frac{1}{2}.
\end{align}
On the other hand,
\begin{align}
	\lim_{n\to \infty}\frac{\eta}{n} &=\lim_{n\to \infty} \frac{(N+1)^{8K}/((1+1/N)^{8K})}{n} \notag\\
	&\geq \lim_{N\to \infty} \frac{1}{(1+1/N)^{8K}}\lim_{n\to \infty}\frac{n-\sqrt{n}}{2n} \notag\\
	&= 1\times \frac{1}{2} = \frac{1}{2}.
\end{align}
Thus, we have that 
\begin{align}
	\lim_{n\to \infty}\frac{\eta}{n} = \frac{1}{2},\label{eq:limeta}
\end{align}
which also implies that
\begin{align}
	&\lim_{n\to \infty}\frac{\overline{\eta}}{n} = \lim_{n\to \infty}  \frac{\eta}{n} \times\lim_{N\to \infty}(1+1/N)^{8K} = \frac{1}{2} \times 1 = \frac{1}{2}.\label{eq:limetabar}
\end{align}
Combining \eqref{eq:deltan} with \eqref{eq:limeta} and \eqref{eq:limetabar} we have 
\begin{align}
\lim_{n\to \infty}\Delta_n = 4\times \frac{1}{2} + 0 =2\label{eq:limdeltanend}
\end{align}
since $K$ is independent of $n$. Thus, for any $\varepsilon >0$, $\exists n_0>0$ such that $\Delta_n \leq 2+\varepsilon$ for all $n\geq n_0$. Recall that the broadcast cost $\Delta_n$ is achievable if $E_n$ holds, i.e., $\Delta^*(\Lambda_n)\leq \Delta_n\leq 2+\epsilon$ if $n\geq n_0$ and $E_n$ holds. Now let us show that $2+\epsilon$ is achievable a.a.s., by evaluating the limit in \eqref{eq:limaasach} as follows,
\begin{align}
&\lim_{n\rightarrow\infty}\mathsf{Pr}\Big(\Delta^*(\Lambda_n) \leq 2+\epsilon\Big) \notag \\
&\geq \lim_{n\rightarrow\infty}\mathsf{Pr}\Bigg(\Big(\Delta^*(\Lambda_n)\leq 2+\epsilon\Big)\land E_n\Bigg)\\
&=\lim_{n\rightarrow\infty}\mathsf{Pr}(E_n)\mathsf{Pr}\Big(\Delta^*(\Lambda_n)\leq 2+\epsilon~\Big|~ E_n\Big)\\
&=1
\end{align}
which implies that $\lim_{n\rightarrow\infty}\mathsf{Pr}\Big(\Delta^*(\Lambda_n)\leq 2+\epsilon\Big)=1$.
Since this is true for all $\varepsilon>0$, according to \eqref{eq:infdeltau} we have $\Delta_u^*\leq \inf\{2+\epsilon\mid \epsilon>0\} = 2$. $\hfill \qed$

\subsection{Example 2: $\left(p,K,d=4,m=2,m'=1\right)$}
Let $L=1$. For $q=p^n$, we will interpret ${\Bbb F}_q$ as an $n$-dimensional vector space over ${\Bbb F}_p$, and design a linear scheme over ${\Bbb F}_p$. Accordingly, let us clarify the notation as follows.
\begin{enumerate}
	\item The elements of the data and coefficient matrices are chosen from ${\Bbb F}_{q} = {\Bbb F}_{p^n}$.
	\item The data $\FqX^T = [\Fqx_1,\Fqx_2,\Fqx_3,\Fqx_4]\in {\Bbb F}_q^{1\times 4}$, is  equivalently represented over  $\mathbb{F}_p$ as $\FpX^T = [\Fpx_1^T,\Fpx_2^T,$ $\Fpx_3^T,\Fpx_4^T]$ $\in\mathbb{F}_p^{1\times 4n}$, where $\Fpx_i\in\mathbb{F}_p^{n\times 1}$ is the $n\times 1$ vector representation of $\Fqx_i$ over $\mathbb{F}_p$.
	\item User $k$ has side information $\FqX^T\FqV_k' \in {\Bbb F}_q$ and wishes to compute $\FqX^T \FqV_k = \FqX^T\big[\FqV_k^{1},\FqV_k^{2}\big] \in {\Bbb F}_q^{1\times 2}$, where the elements of $(\FqV_k')^T = [\Fqv'_{k1},\Fqv'_{k2},\Fqv'_{k3},\Fqv'_{k4}]$, $(\FqV_k^{\mu})^T = [\Fqv_{k1}^{\mu},\Fqv_{k2}^{\mu},\Fqv_{k3}^{\mu},\Fqv_{k4}^{\mu}], {\mu}\in[2]$ are drawn i.i.d. uniform in ${\Bbb F}_q$. Equivalently, over ${\Bbb F}_p$, User $k$ has side information $\FpX^T \FpV'_k$ and wishes to compute $\FpX^T \FpV_k$, where $(\FpV_k')^T = \big[(\Fpv'_{k1})^T,(\Fpv'_{k2})^T,(\Fpv'_{k3})^T,(\Fpv'_{k4})^T \big] \in {\Bbb F}_p^{n\times 4n}$, $\FpV_k = [\FpV_k^1,\FpV_k^2]$ and $(\FpV_k^{\mu})^T=\big[ (\Fpv_{k1}^{\mu})^T,(\Fpv_{k2}^{\mu})^T,(\Fpv_{k3}^{\mu})^T,(\Fpv_{k4}^{\mu})^T \big] \in {\Bbb F}_p^{n\times 4n}$. $\Fpv_{ki}',\Fpv_{ki}^{\mu}, {\mu}\in[2]$ are the $n\times n$ matrix representations in ${\Bbb F}_p$ of $\Fqv_{ki}'$ and $\Fqv_{ki}^{\mu}$, respectively.
	\item Let $\Fqr, \Fqs_1, \Fqs_2$ be uniformly randomly chosen in ${\Bbb F}_q$.  Denote by $\Fpr$ the matrix representations of $\Fqr, \Fqs_1, \Fqs_2$ in ${\Bbb F}_p$, respectively.
	\item Define the set of variables,
		\begin{align}
			\mathcal{V}& \triangleq \Big\{ \Fqv_{ki}^{j}: k\in [K], i\in [4], j \in[2] \Big\}  \cup \Big\{ \Fqv_{ki}': k\in [K], i\in [4] \Big\} \cup \Big\{ \Fqr, \Fqs_1, \Fqs_2 \Big\},
		\end{align}
		and note that $|\mathcal{V}| = 12K+3$.
\end{enumerate}
Our goal is to show that $\Delta^*_u \leq dm/(m+m') = \frac{8}{3}$. For all $k\in [K], i\in [4], {\mu}\in[2]$, let us define $t_{ki}^{\mu}$ such that
\begin{align} \label{eq:def_vvt_e2}
	\begin{bmatrix}
		\Fqv_{k1}' \\ \Fqv_{k2}' \\ \Fqv_{k3}' \\ \Fqv_{k4}'
	\end{bmatrix}
	=
	\begin{bmatrix}
		\Fqv_{k1}^{\mu}  \\ \Fqv_{k2}^{\mu} \\ \Fqv_{k3}^{\mu} \\ \Fqv_{k4}^{\mu}
	\end{bmatrix}
	\Fqr
	+
	\begin{bmatrix}
		\Fqt_{k1}^{\mu}  \\ \Fqt_{k2}^{\mu}  \\ \Fqt_{k3}^{\mu}  \\ \Fqt_{k4}^{\mu} 
	\end{bmatrix}.
\end{align}

Let $\Fpt_{ki}^{\mu} \in {\Bbb F}_p^{n\times n}$ denote the $n\times n$ matrix representations of $\Fqt_{ki}^{\mu}$ in ${\Bbb F}_p$, so that we have in ${\Bbb F}_p$, 
\begin{align} \label{eq:def_VVT_e2}
	\underbrace{\begin{bmatrix}
		\Fpv_{k1}' \\ \Fpv_{k2}' \\ \Fpv_{k3}' \\ \Fpv_{k4}'
	\end{bmatrix}}_{\FpV_k' \in {\Bbb F}_p^{4n \times n}}
	=
	\underbrace{\begin{bmatrix}
		\Fpv_{k1}^{\mu}  \\ \Fpv_{k2}^{\mu} \\ \Fpv_{k3}^{\mu} \\ \Fpv_{k4}^{\mu}
	\end{bmatrix}}_{\FpV_k^{\mu} \in {\Bbb F}_p^{4n \times n}}
	\Fpr
	+
	\underbrace{\begin{bmatrix}
		\Fpt_{k1}^{\mu}  \\ \Fpt_{k2}^{\mu}  \\ \Fpt_{k3}^{\mu}  \\ \Fpt_{k4}^{\mu} 
	\end{bmatrix}}_{\FpT_k^{\mu} \in {\Bbb F}_p^{4n \times n}}
\end{align}
for all $k\in[K], {\mu}\in [2]$, and $\FpT_k^{\mu},{\mu}\in[2]$ are defined as in \eqref{eq:def_VVT_e2}.

Since $m=2>1$ in this example, unlike what we did in the previous $m=1$ example, this time we will need to create two $\FpH$ matrices, namely, $\FpH_1$ and $\FpH_2$, such that $\FpH_1$ is almost invariant under linear transformations $\Fpv_{ki}^j$ and $\Fpt_{ki}^1, \forall k\in[K], i\in[4],j\in[2]$, i.e., $\langle \Fpv_{ki}^j  \FpH_1 \rangle_p \approx \langle \FpH_1 \rangle_p$ and $\langle \Fpt_{ki}^1 \FpH_1 \rangle_p \approx \langle \FpH_1 \rangle_p$. $\FpH_2$ is almost invariant under linear transformations $\Fpv_{ki}^j$ and $\Fpt_{ki}^2, \forall k\in[K], i\in[4],j\in[2]$, i.e., $\langle \Fpv_{ki}^j  \FpH_2 \rangle_p \approx \langle \FpH \rangle_p$ and $\langle \Fpt_{ki}^2 \FpH_2 \rangle_p \approx \langle \FpH_2 \rangle_p$. In addition, we want $\eta/n\approx 1/3$ so that the columns of $\FpH_\mu,\mu \in[2]$ span approximately one third of the $n$-dimensional vector space. Moreover, $\FpH_1$ and $\FpH_2$ are required to be linearly independent a.a.s. For these, we invoke the asymptotic IA scheme of \cite{Cadambe_Jafar_int}, and design $\FqH_{\mu}, \overline{\FqH}_{\mu}, \mu \in [2]$ in the following way,
\begin{align} \label{eq:def_h_e2}
	&\FqH_{\mu}^{1\times \eta} = \Bigg[ \Fqs_{\mu} \prod_{k=1}^K \Bigg( \Big( \prod_{i=1}^4\prod_{j=1}^2 (\Fqv_{ki}^j )^{\alpha_{ki}^j} \Big) \Big( \prod_{i=1}^4(\Fqt_{ki}^{\mu})^{\beta_{ki}^\mu} \Big) \Bigg),  \mbox{s.t.}~ 0 \leq \alpha_{ki}^j, \beta_{ki}^\mu \leq N-1 \Bigg] \\
	&~~~~~~~~~\triangleq (\Fqh_\mu^1,\Fqh_\mu^2,...,\Fqh_\mu^{\eta}), ~~~~~~ \forall \mu\in[2], \\
	&\overline{\FqH}_{\mu}^{1\times \eta} = \Bigg[ \Fqs_{\mu} \prod_{k=1}^K \Bigg( \Big( \prod_{i=1}^4\prod_{j=1}^2 (\Fqv_{ki}^j  )^{\alpha_{ki}^j} \Big) \Big( \prod_{i=1}^4(\Fqt_{ki}^{\mu})^{\beta_{ki}^\mu} \Big) \Bigg), \mbox{s.t.}~ 0 \leq \alpha_{ki}^j, \beta_{ki}^\mu \leq N \Bigg] \\
	&~~~~~~~~~\triangleq (\overline{\Fqh}_\mu^1,\overline{\Fqh}_\mu^2,...,\overline{\Fqh}_\mu^{\overline{\eta}}), ~~~~~~ \forall \mu\in[2].
\end{align}
Note that we have,
\begin{align}
	\eta = N^{12K},~~ \overline{\eta} = {(N+1)}^{12K}.
\end{align}
This construction ensures that for $\mu\in[2]$, the elements of $\Fqv_{ki}^j \FqH_\mu$ and $\Fqt_{ki}^\mu \FqH_\mu$ are contained among the elements of $\overline{\FqH}_\mu$ for all $i\in[4], j\in[2]$. Now let $\Fph_1, \Fph_2, \hdots, \Fph_{\eta} \in {\Bbb F}_p^{n\times n}$ be the matrix representations in ${\Bbb F}_p$ of $\Fqh_1,\Fqh_2,\hdots,\Fqh_{\eta} \in {\Bbb F}_q$, and $\overline{\Fph}_1,\overline{\Fph}_2,\hdots, \overline{\Fph}_{\overline{\eta}} \in {\Bbb F}_p^{n\times n}$ be the matrix representations in ${\Bbb F}_p$ of $\overline{\Fqh}_1,\overline{\Fqh}_2,\hdots, \overline{\Fqh}_{\overline{\eta}} \in {\Bbb F}_q$. Define,
\begin{align}
	\FpH_\mu = \big[ \Fph_{\mu}^1{\bf 1}, \Fph_{\mu}^2{\bf 1}, \hdots, \Fph_{\mu}^{\eta}{\bf 1}  \big] ~~ \in {\Bbb F}_p^{n\times \eta}, ~~~~~~ {\mu} \in [2]
\end{align}
and
\begin{align}
	\overline{\FpH}_{\mu} = \big[ \overline{\Fph}_{\mu}^1{\bf 1}, \overline{\Fph}_{\mu}^2{\bf 1}, \hdots, \overline{\Fph}_{\mu}^{\overline{\eta}}{\bf 1}  \big]~~\in {\Bbb F}_p^{n\times \overline{\eta}}, ~~~~~~ {\mu} \in [2],
\end{align}
where ${\bf 1}$ denotes the $n\times 1$ vector of $1$'s. By construction, the columns of $\Fpv_{ki}^j \FpH_\mu$ and $\Fpt_{ki}^\mu\FpH_\mu$ are subsets of the columns of $\FpH_{\mu}$, which implies that $\forall \mu\in[2], k\in[K], i\in[4],j\in[2]$,
\begin{align}  \label{eq:H_subset_e2}
	\langle \Fpv_{ki}^j \FpH_\mu  \rangle_p \subset \langle \overline{\FpH}_\mu \rangle_p, && \langle \Fpt_{ki}^\mu \FpH_\mu \rangle_p \subset \langle \overline{\FpH}_\mu \rangle_p.
\end{align}
Consider the $m=2$ matrices, $[\FpH_1,\FpH_2,\Fpr\FpH_1]\in {\Bbb F}_p^{n\times 3\eta}$ and $[\FpH_1,\FpH_2,\Fpr\FpH_2] \in {\Bbb F}_p^{n\times 3\eta}$, and define the event $E_n$ as,
\begin{align}
	&E_n \triangleq \left( \rk([\FpH_1,\FpH_2,\Fpr\FpH_1])= 3\eta \right)  \wedge \left( \rk([\FpH_1,\FpH_2,\Fpr\FpH_2])= 3\eta \right).
\end{align}
The next steps \eqref{eq:gotozerobegin_e2}-\eqref{eq:gotozero_e2}, show that $E_n$ holds a.a.s., which will subsequently be essential to guarantee successfully decoding by each user. In fact, due to symmetry, it suffices to prove that $ \rk([\FpH_1,\FpH_2,\Fpr\FpH_1]) \aaseq 3\eta $.

$[\FpH_1,\FpH_2,\Fpr\FpH_1]\in {\Bbb F}_p^{n\times 3\eta}$ has full column rank if and only if for all ${\bf c}_1 = [c_1^1,c_1^2,\hdots, c_1^{\eta}]^T, {\bf c}_2= [c_2^1,c_2^2,\hdots, c_2^{\eta}]^T, {\bf c}'= [c'^1,c'^2,\hdots, c'^{\eta}]^T \in {\Bbb F}_p^{\eta \times 1}$ such that $\overline{{\bf c}}^T = [{\bf c}_1^T, {\bf c}_2^T, {\bf c}'^T]\not= {\bf 0}^{1\times 3\eta}$, 
\begin{align} \label{eq:gotozerobegin_e2}
	{\bf 0}^{n\times 1} &\not= [ \FpH_1,  \FpH_2, \Fpr\FpH_1]
	\overline{{\bf c}} \\
	& = \FpH_1{\bf c}_1+  \FpH_2{\bf c}_2 + \Fpr\FpH_1{\bf c}' \\
	& = \sum_{j=1}^{\eta}c_1^j  \Fph_1^{j}{\bf 1} + \sum_{j=1}^{\eta}c_2^j  \Fph_2^{j}{\bf 1} + \sum_{j=1}^{\eta}c'^j \Fpr\Fph_1^{j}{\bf 1} \\
	& = \left(\sum_{j=1}^{\eta}c_1^j  \Fph_1^{j} + \sum_{j=1}^{\eta}c_2^j  \Fph_2^{j} + \sum_{j=1}^{\eta}c'^j \Fpr\Fph_1^{j}\right) {\bf 1} \\
	& \triangleq \Fpf_{\overline{\bf c}} {\bf 1}
\end{align}
where $\Fpf_{\overline{\bf c}}$ is an $n\times n$ matrix in ${\Bbb F}_p$, which has a scalar representation in ${\Bbb F}_q$ as,
\begin{align}
	\Fqf_{\overline{\bf c}} = \sum_{j=1}^{\eta}c_1^j  \Fqh_1^{j} + \sum_{j=1}^{\eta}c_2^j \Fqh_2^{j} + \sum_{j=1}^{\eta}c'^j \Fqr \Fqh_1^{j} && \in {\Bbb F}_q.
\end{align}
Note that since ${\Bbb F}_p$ is a sub-field of ${\Bbb F}_q$, the scalars $c_1^j,c_2^j, c'^j$ in ${\Bbb F}_p$ are also scalars $c_1^j,c_2^j, c'^j$ in ${\Bbb F}_q$.

Thus, $\Fpf_{\overline{\bf c}} \in {\Bbb F}_p^{n\times 1}$ can be equivalently represented in ${\Bbb F}_q$ as the product of $\Fqf_{\overline{\bf c}}$ with the scalar representation in ${\Bbb F}_q$, of ${\bf 1}$ (the all $1$ vector in ${\Bbb F}_p$). Since the ${\Bbb F}_q$ representation of ${\bf 1}$ is not $0$, we obtain that
\begin{align}
	\Fpf_{\overline{\bf c}} {\bf 1} \not= {\bf 0}^{n\times 1} \Longleftrightarrow \Fqf_{\overline{\bf c}} \not= 0.
\end{align}
Therefore, $[ \FpH_1, \FpH_2,\Fpr\FpH_1]\in {\Bbb F}_p^{n\times 3\eta}$ has full column rank if and only if,
\begin{align}
	P \triangleq \prod_{\overline{\bf c} \in {\Bbb F}_p^{3\eta \times 1} \backslash \{{\bf 0} \}} \Fqf_{\overline{\bf c}} \not =0. \label{eq:nonzero_e2}
\end{align}
The condition of \eqref{eq:nonzero_e2}, which is equivalent to the event $E_n$, says that a uniformly random evaluation of the function $P(\mathcal{V})$ produces a non-zero value. We will show that this is true a.a.s. in $n$.
\begin{enumerate}
	\item Case I: ${\bf c}_1$ or ${\bf c}_2$ is not the zero vector. Let us set $\Fqr=0$, which implies $\Fqf_{\overline{\bf c}} = \sum_{j=1}^{\eta}c_1^j  \Fqh_1^{j}+ \sum_{j=1}^{\eta}c_2^j  \Fqh_2^{j}$, and that  $\Fqt_{ki}^\mu = \Fqv'_{ki}, \forall \mu\in[2], i\in[4]$ by \eqref{eq:def_vvt_e2}. Meanwhile, $\Fqh_{1}^1,\Fqh_{1}^2,\hdots, \Fqh_{1}^{\eta}$ are different monomials in the elements of $\Fqv_{ki}^j$, $\Fqv_{ki}'$ and $\Fqs_1$. Similarly, $\Fqh_{2}^1,\Fqh_{2}^2,\hdots, \Fqh_{2}^{\eta}$ are different monomials in the elements of $\Fqv_{ki}^j$, $\Fqv_{ki}'$ and $\Fqs_2$ due to \eqref{eq:def_h_e2}. Moreover, since any $\Fqh_1^j$ has the factor $\Fqs_1$ but does not have the factor $\Fqs_2$, and any $\Fqh_2^j$ has the factor $\Fqs_2$ but does not have the factor $\Fqs_1$, it follows that $ \Fqh_1^1, \Fqh_1^2,\hdots,  \Fqh_1^{\eta},  \Fqh_2^1, \Fqh_2^2,\hdots,  \Fqh_2^{\eta}$ are different monomials. Since different monomials are linearly independent, we have that $[\Fqf_{\overline{\bf c}}]$ is a non-zero polynomial.
	\item Case II: ${\bf c}_1 = {\bf c}_2 = {\bf 0}$ and thus ${\bf c}'\not={\bf 0}$. For this case, we have $[\Fqf_{\overline{\bf c}}] = \Fqr(c'^1\Fqh_1^1+c'^2\Fqh_1^2+\hdots+c'^\eta\Fqh_1^\eta)$, which is also a non-zero polynomial since it is a product of $\Fqr$ with a non-zero polynomial (since $\Fqh_1^1,\Fqh_1^2,\hdots,\Fqh_1^\eta$ are linearly independent).
\end{enumerate}
Thus, $[\Fqf_{\overline{\bf c}}]$ is a non-zero polynomial. Since $h_i^j$ has degree not more than $16(N-1)K+1$, $[\Fqf_{\overline{\bf c}}]$ has degree not more than $16NK+2$. Therefore, $[P] \in {\Bbb F}_p[\mathcal{V}]$ is a non-zero polynomial with degree not more than $(p^{3\eta}-1)(16NK+2)$. By Schwartz-Zippel Lemma, when all the variables $\mathcal{V}$ are assigned i.i.d. uniformly chosen values in ${\Bbb F}_q$, 
\begin{align}
	{\sf Pr}\big( P\not=0 \big) &\geq 1-\frac{(p^{3\eta}-1)[16(N-1)K+2]}{q} \\
	&= 1-\frac{(p^{3\eta}-1)[16(N-1)K+2]}{p^n} \\
	& \geq 1- 16K\frac{N}{p^{n-3\eta}} \\
	& \to 1
\end{align}
as $n\to \infty$ if $\lim_{n\to \infty}\frac{N}{p^{n-3\eta}} = 0$. 

Now let us specify the value of $N$ as follows,
\begin{align} \label{eq:def_N_e2}
	N = \left\lfloor \left( \frac{n-\sqrt{n}}{3} \right)^{1/(12K)} \right\rfloor,
\end{align}
from which it follows that,
\begin{align}
	N \leq \left( \frac{n-\sqrt{n}}{3} \right)^{1/(12K)} 
\end{align}
and
\begin{align}
	n-3\eta = n-3N^{12K} \geq \sqrt{n}.
\end{align}
Therefore,
\begin{align} \label{eq:gotozero_e2}
	\lim_{n\to \infty} \frac{N}{p^{n-3\eta}} \leq \lim_{n\to \infty}\frac{\left( \frac{n-\sqrt{n}}{3} \right)^{1/(12K)}}{p^{\sqrt{n}}} = 0,
\end{align}
and since $N\geq 0$, we have $\lim_{n\to \infty} \frac{N}{p^{n-3\eta}} = 0$. Thus, $\rk([ \FpH_1, \FpH_2,\Fpr\FpH_1]) \aaseq 3\eta$, and due to symmetry it can be proved that $\rk([ \FpH_1, \FpH_2,\Fpr\FpH_2]) \aaseq 3\eta$. Thus, we have shown that $E_n$ holds a.a.s. 
Now for $\mu\in[2]$, let $({\bf Z}_{\mu} = {\bf I}^{n\times n}|[\FpH_1, \FpH_2, \Fpr \FpH_\mu]) \in {\Bbb F}_p^{n\times (n-3\eta)}$, so that
\begin{align}
	\big[ \FpH_1,  \FpH_2, \Fpr\FpH_\mu, {\bf Z}_\mu \big]
\end{align}
has full column rank $n$.
Let the server broadcast ${\bf S} = \big( {\bf S}_0, {\bf S}^1_{[K]}, {\bf S}^2_{[K]} \big) \in {\Bbb F}_p^{1\times (8\overline{\eta}+2K(n-3\eta))}$, where
\begin{align}
	{\bf S}_0 &=  \FpX^T
	\begin{bmatrix}
		\overline{\FpH}_1 & \overline{\FpH}_2 & {\bf 0} & {\bf 0} & {\bf 0} & {\bf 0} & {\bf 0} & {\bf 0} \\
		{\bf 0} & {\bf 0} & \overline{\FpH}_1 & \overline{\FpH}_2 & {\bf 0} & {\bf 0} & {\bf 0} & {\bf 0} \\
		{\bf 0} & {\bf 0} & {\bf 0} & {\bf 0} & \overline{\FpH}_1 & \overline{\FpH}_2 & {\bf 0} & {\bf 0} \\
		{\bf 0} & {\bf 0} & {\bf 0} & {\bf 0} & {\bf 0} & {\bf 0} & \overline{\FpH}_1 & \overline{\FpH}_2
	\end{bmatrix}  ~~\in {\Bbb F}_p^{1 \times 8\overline{\eta}}
\end{align}
and for $k\in [K],\mu\in[2]$,
\begin{align}
	{\bf S}_k^{\mu} = \FpX^T \FpV_k^\mu {\bf Z}_\mu.
\end{align}
The decoding process works as follows. User $k$ is able to compute
\begin{align} \label{eq:decode_1_e2}
	\FpX^T \FpV_k' \FpH_\mu, ~~~~~~ \mu \in[2]
\end{align}
directly from its side information $\FpX^T \FpV_k'$. Meanwhile, it is able to compute
\begin{align} \label{eq:projection_1_e2}
	\FpX^T  \FpV_k^j   \FpH_\mu, ~~~~~~ j\in[2], ~~ \mu \in[2]
\end{align}
and
\begin{align} \label{eq:decode_2_e2}
	\FpX^T \FpT_k^{\mu} \FpH_\mu, ~~~~~~ \mu \in[2]
\end{align}
from ${\bf S}_0$, since for $j\in[2],\mu \in[2]$,
\begin{align}
	\left\langle  \FpV_k^{j}   \FpH_\mu \right\rangle_p &= \Bigg\langle  \begin{bmatrix}
		\Fpv_{k1}^{j}   \FpH_{\mu} \\ \Fpv_{k2}^{j}   \FpH_{\mu} \\ \Fpv_{k3}^{j}   \FpH_{\mu} \\ \Fpv_{k4}^{j}   \FpH_{\mu}
	\end{bmatrix} \Bigg\rangle_p \subset
	\Bigg\langle  \begin{bmatrix}
		\overline{\FpH}_{\mu} & {\bf 0} & {\bf 0} & {\bf 0} \\
		{\bf 0} & \overline{\FpH}_{\mu} & {\bf 0} & {\bf 0} \\
		{\bf 0} & {\bf 0} & \overline{\FpH}_{\mu} & {\bf 0} \\
		{\bf 0} & {\bf 0} & {\bf 0} & \overline{\FpH}_{\mu} \end{bmatrix}
		\Bigg\rangle_p,   
\end{align}
and for $\mu \in [2]$,
\begin{align}
	& \Bigg\langle  \begin{bmatrix}
		\Fpt_{k1}^\mu \FpH_{\mu} \\ \Fpt_{k2}^\mu \FpH_{\mu} \\ \Fpt_{k3}^\mu \FpH_{\mu} \\ \Fpt_{k4}^\mu \FpH_{\mu}
	\end{bmatrix} \Bigg\rangle_p 
	\subset
	\Bigg\langle  \begin{bmatrix}
		\overline{\FpH}_{\mu} & {\bf 0} & {\bf 0} & {\bf 0} \\
		{\bf 0} & \overline{\FpH}_{\mu} & {\bf 0} & {\bf 0} \\
		{\bf 0} & {\bf 0} & \overline{\FpH}_{\mu} & {\bf 0} \\
		{\bf 0} & {\bf 0} & {\bf 0} & \overline{\FpH}_{\mu} \end{bmatrix}
		\Bigg\rangle_p
\end{align}
due to \eqref{eq:H_subset_e2}.
Thus, User $k$ is able to compute
\begin{align} \label{eq:projection_2_e2}
	\FpX^T \FpV_k^\mu \Fpr \FpH_\mu = \FpX^T\FpV_k'\FpH_\mu - \FpX^T \FpT_k^\mu \FpH_\mu, ~~ \mu\in[2]
\end{align}
according to \eqref{eq:def_VVT_e2}.
Together with ${\bf S}_k^\mu$, User $k$ thus obtains,
\begin{align}
	\big[\FpX^T \FpV_k^j  \FpH_\mu, \FpX^T \FpV_k^\mu \Fpr \FpH_\mu,{\bf S}_k^\mu \big],~~ j\in[2], ~~ \mu \in[2]
\end{align}
which are 
\begin{align}
	\FpX^T \FpV_k^1 \big[ \FpH_1,  \FpH_2, \Fpr \FpH_1,{\bf Z}_1 \big], ~\FpX^T \FpV_k^2 \big[ \FpH_1,  \FpH_2, \Fpr \FpH_2,{\bf Z}_2 \big]
\end{align}
and since $\big[ \FpH_1,  \FpH_2, \Fpr \FpH_1,{\bf Z}_1 \big] \in {\Bbb F}_p^{n\times n}$ and $\big[  \FpH_1,  \FpH_2, \Fpr \FpH_2,{\bf Z}_2 \big]\in {\Bbb F}_p^{n\times n}$  are invertible (have full rank) a.a.s., User $k$ is able to retrieve its desired computation, $\FpX^T \FpV_k=\big[\FpX^T \FpV_k^1,\FpX^T \FpV_k^2\big] \in{\Bbb F}_p^{1\times 2n}$ a.a.s.

For $q=p^n$, the cost of broadcasting each $p$-ary symbol is $1/n$ in $q$-ary units. Thus, the broadcast cost of this scheme is,
\begin{align} \label{eq:deltan_e2}
	\Delta_n = \frac{8\overline{\eta}+2K\big(n-3\eta\big)}{n}.
\end{align}
The next few steps \eqref{eq:limdeltanbegin_e2}-\eqref{eq:limdeltanend_e2} show that $\lim_{n\rightarrow\infty}\Delta_n= \frac{8}{3}$.

By \eqref{eq:def_N_e2}, we have that
\begin{align} \label{eq:limdeltanbegin_e2}
	\eta = N^{12K} \leq \frac{n-\sqrt{n}}{3} \leq (N+1)^{12K} = \overline{\eta}
\end{align}
which implies that
\begin{align}
	\lim_{n\to \infty}\frac{\eta}{n} = \lim_{n\to \infty}\frac{N^{12K}}{n}  \leq \lim_{n\to \infty}\frac{n-\sqrt{n}}{n} \times \frac{1}{3} = \frac{1}{3}.
\end{align}
On the other hand,
\begin{align}
	\lim_{n\to \infty} \frac{\eta}{n} &=\lim_{n\to \infty} \frac{(N+1)^{12K}/(1+1/N)^{12K}}{n} \\
	&\geq \lim_{N\to \infty} \frac{1}{(1+1/N)^{12K}}\lim_{n\to \infty}\frac{n-\sqrt{n}}{n}\times \frac{1}{3} \\
	& = \frac{1}{3}.
\end{align}
Thus, we have that 
\begin{align} \label{eq:limeta_e2}
	\lim_{n\to \infty}\frac{\eta}{n} = \frac{1}{3}
\end{align}
which also implies that
\begin{align} \label{eq:limetabar_e2}
	\lim_{n\to \infty}\frac{\overline{\eta}}{n} = \lim_{n\to \infty}  \frac{\eta}{n} \times\lim_{N\to \infty}(1+1/N)^{12K} = \frac{1}{3}.
\end{align}
Combining \eqref{eq:deltan_e2} with \eqref{eq:limeta_e2} and \eqref{eq:limetabar_e2} we have
\begin{align} \label{eq:limdeltanend_e2}
	\lim_{n\to \infty}\Delta_n = 8\times \frac{1}{3} + 0 = \frac{8}{3}
\end{align}
since $K$ is independent of $n$. Thus, for any $\varepsilon >0$, $\exists n_0>0$ such that $\Delta_n \leq \frac{8}{3}+\varepsilon$ for all $n\geq n_0$. Recall that the broadcast cost $\Delta_n$ is achievable if $E_n$ holds, i.e., $\Delta^*(\Lambda_n)\leq \Delta_n\leq \frac{8}{3}+\epsilon$ if $n\geq n_0$ and $E_n$ holds. Now let us show that $\frac{8}{3}+\epsilon$ is achievable a.a.s., by evaluating the limit in \eqref{eq:limaasach} as follows,
\begin{align}
&\lim_{n\rightarrow\infty}\mathsf{Pr}\Big(\Delta^*(\Lambda_n)\leq \frac{8}{3}+\epsilon\Big)\notag\\
&\geq \lim_{n\rightarrow\infty}\mathsf{Pr}\Bigg(\Big(\Delta^*(\Lambda_n)\leq \frac{8}{3}+\epsilon\Big)\land E_n\Bigg)\\
&=\lim_{n\rightarrow\infty}\mathsf{Pr}(E_n)\mathsf{Pr}\Big(\Delta^*(\Lambda_n)\leq \frac{8}{3}+\epsilon~\Big|~ E_n\Big)\\
&=1
\end{align}
which implies that $\lim_{n\rightarrow\infty}\mathsf{Pr}\Big(\Delta^*(\Lambda_n)\leq \frac{8}{3}+\epsilon\Big)=1$.
Since this is true for all $\varepsilon>0$, according to \eqref{eq:infdeltau} we have $\Delta_u^*\leq \inf\{\frac{8}{3}+\epsilon\mid \epsilon>0\} = \frac{8}{3}$. $\hfill \qed$

\section{Proof of Converse: Theorem \ref{thm:large_K}}\label{sec:proof_converse_large_K}
Recall that $\FqU_k \triangleq [\FqV_k',\FqV_k],\forall k\in [K]$, and the data $\FqXL \in \mathbb{F}_q^{d\times L}$ for a scheme with batch size equal to $L$. Since a scheme must work for all data realizations, it must work if $\FqXL$ is uniformly distributed. The decoding constraint \eqref{eq:decoding_constraint} implies
\begin{align} \label{eq:conv_ex_decoding}
	H({\bf S}, \FqXL^T \FqV_k') = H({\bf S}, \FqXL^T \FqU_k), ~ \forall k\in [K].
\end{align}
The converse for $d \geq K(m+m')$ is obtained trivially by allowing all $K$ users to cooperate fully, see proof of \eqref{eq:genlbound} in Corollary \ref{cor:allK}. The converse for $d\leq m+m'$ is obtained as 
\begin{align}
	\Delta^*(\Lambda_n) & \geq H({\bf S})/L \label{eq:conv_r1_1} \\
	&\geq H({\bf S}\mid \FqXL^T\FqV_1')/L \label{eq:conv_r1_2} \\
	&= H({\bf S}, \FqXL^T\FqU_1\mid \FqXL^T\FqV_1')/L \label{eq:conv_r1_3} \\
	&\geq H(\FqXL^T\FqU_1~|~\FqXL^T\FqV_1')/L \label{eq:conv_r1_4} \\
	&\geq \big(H(\FqXL^T\FqU_1) - H(\FqXL^T\FqV_1')\big)/L \label{eq:conv_r1_5} \\
	&= \rk(\FqU_1) - \rk(\FqV_1')\aaseq (d-m')^+ \label{eq:conv_r1_6}
\end{align}
Step \eqref{eq:conv_r1_1} is due to Shannon's source coding theorem. Steps \eqref{eq:conv_r1_2}, \eqref{eq:conv_r1_4} and \eqref{eq:conv_r1_5} follow from basic information inequalities. Step \eqref{eq:conv_r1_3} is because User $1$ must decode $\FqXL^T\FqU_1$ from ${\bf S}$ and $\FqXL^T\FqV_1'$. Step \eqref{eq:conv_r1_6} applies the useful connection between entropy and ranks, that $H(\FqXL^T\FqU) = L\cdot \rk(\FqU)$ for a uniformly distributed $\FqXL$ and a deterministic matrix $\FqU$. This leaves us with the only non-trivial regime, $(m+m')<d<K(m+m')$, for which we will show that $\Delta^*(\Lambda_n)\aasgeq dm/(m+m')$ in the remainder of this section. 

Let us provide an intuitive outline before launching into the technical details. Recall that Theorem \ref{thm:large_K} considers $K \geq d/\gcd(d,m+m')$. If $m+m'$ divides $d$, it immediately follows that $K \geq d/\gcd(d,m+m') = d/(m+m') \implies d \geq K(m+m')$. Therefore, the non-trivial cases must be that $m+m'$ does not divide $d$.
What we want for the converse argument, intuitively, is to still have the first $d/(m+m')$ users cooperate fully. This is not directly possible because $d/(m+m')$ is not a natural number, but let us set that concern aside for a moment. The $d/(m+m')$ users together already have side information that is equivalent to $m'd/(m+m')$ dimensional projection of the data, which together with the broadcast symbol ${\bf S}$ allows them to recover $(m+m')d/(m+m') = d$ dimensions of the data. If so, then we would have that $H({\bf S})/L \geq d-m'd(m+m') = dm/(m+m')$ as the desired converse bound. Now, how do we overcome the obstacle that we cannot have a fractional number of users? Intuitively, this is achieved by invoking functional submodularity (Lemma 1 of \cite{Yao_Jafar_3LCBC}, \cite{Tao_FS,Kontoyiannis_Madiman}). The idea is that functional submodularity helps to identify and introduce additional entropic terms of certain (linear) functions of the side information and  demands. These functions are essentially the projection of the data into finer subspaces. If we regard the entropies of the subsets of the side information and demands as a set of regular building blocks, the additional entropies introduced by functional submodularity are similar to finer fragments. By rearranging and combining these regular building blocks and fragments in a more efficient way, we are able to derive a better converse bound. To make the details concrete, the readers may refer to the following proof sketch for the example with $m+m'=6,d=10$ and $K = d/\gcd(d,m+m') = 5$. For this example, we want to show that $\Delta^*(\Lambda_n) \aasgeq md/(m+m') = 5m/3$. 

In the following, steps labeled $(*)$ uses functional submodularity (Lemma 1 of \cite{Yao_Jafar_3LCBC}, \cite{Tao_FS,Kontoyiannis_Madiman}).
We proceed as follows.
\begin{align}
	&\underbrace{H({\bf S}, \FqXL^T \FqV_1') + H({\bf S}, \FqXL^T \FqV_2')}_{T_{12}} \notag\\
	&\stackrel{\eqref{eq:conv_ex_decoding}}{=} H({\bf S}, \FqXL^T \FqU_1) + H({\bf S}, \FqXL^T \FqU_2) \\
	&\stackrel{(*)}{\geq} H({\bf S}, \FqXL^T [\FqU_1,\FqU_2]) + H({\bf S}, \FqXL^T (\FqU_1 \cap \FqU_2)) \\
	&\geq H(\FqXL^T [\FqU_1,\FqU_2]) + H({\bf S}, \FqXL^T (\FqU_1 \cap \FqU_2)) \\
	&\aaseq 10L + H({\bf S}, \FqXL^T (\FqU_1 \cap \FqU_2)) \label{eq:conv_ex_aas}
\end{align}
The last step is because as $n\to\infty$ the rank of $[\FqU_1,\FqU_2]$ is equal to $10$ a.a.s. (The proof is omitted here but can be found in the proof for the general case). Then,
\begin{align}
	&\underbrace{T_{12} + H({\bf S}, \FqXL^T\FqV_3')}_{T_{123}} \notag\\
	&\stackrel{\eqref{eq:conv_ex_decoding}}{=} T_{12} + H({\bf S}, \FqXL^T\FqU_3)  \\
	&\stackrel{(*)}{\geq} 10L + H({\bf S}, \FqXL^T [\FqU_1 \cap \FqU_2, \FqU_3]) + H({\bf S})
\end{align}
It follows that,
\begin{align}
	&\underbrace{T_{123}+ H({\bf S}, \FqXL^T\FqV_4')}_{T_{1234}} \notag \\
	&\stackrel{\eqref{eq:conv_ex_decoding}}{=} T_{123}+ H({\bf S}, \FqXL^T\FqU_4) \\
	&\stackrel{(*)}{\geq} 10L + H({\bf S}) + H({\bf S}, \FqXL^T[\FqU_1\cap \FqU_2, \FqU_3, \FqU_4]) \notag\\
	&~~~~+ H({\bf S}, \FqXL^T[(\FqU_1\cap \FqU_2, \FqU_3) \cap  \FqU_4])\\
	&\geq 10L + H({\bf S}) + H(\FqXL^T[\FqU_1\cap \FqU_2, \FqU_3, \FqU_4]) \notag\\
	&~~~~+ H({\bf S}, \FqXL^T[(\FqU_1\cap \FqU_2, \FqU_3) \cap  \FqU_4])\\
	&\aaseq 20L+H({\bf S})+ H({\bf S}, \FqXL^T[(\FqU_1\cap \FqU_2, \FqU_3) \cap  \FqU_4])
\end{align}
The last step is because as $n\to\infty$ the rank of $[\FqU_1\cap \FqU_2,\FqU_3, \FqU_4]$ is equal to $10$ a.a.s. Then,
\begin{align}
	&\underbrace{T_{1234} + H({\bf S}, \FqXL^T\FqV_5')}_{T_{12345}} \notag \\
	&\stackrel{\eqref{eq:conv_ex_decoding}}{=} T_{1234} + H({\bf S}, \FqXL^T\FqU_5) \\
	&\geq 20L + H({\bf S}) + H({\bf S}, \FqXL^T[(\FqU_1\cap \FqU_2, \FqU_3) \cap  \FqU_4]) \notag\\
	&~~~~+ H({\bf S}, \FqXL^T\FqU_5) \\
	&\stackrel{(*)}{\geq} 20L+2H({\bf S}) + H({\bf S}, \FqXL^T[(\FqU_1\cap \FqU_2, \FqU_3) \cap  \FqU_4, \FqU_5]) \\
	&\geq 20L+2H({\bf S}) + H(\FqXL^T[(\FqU_1\cap \FqU_2, \FqU_3) \cap  \FqU_4, \FqU_5])  \\
	&\aaseq 30L + 2H({\bf S})
\end{align}
The last step is because as $n\to\infty$ the rank of $[(\FqU_1\cap \FqU_2, \FqU_3) \cap  \FqU_4, \FqU_5]$ is equal to $10$ a.a.s.

On the other hand,
\begin{align}
	T_{12345} &= H({\bf S}, \FqXL^T\FqV_1') + H({\bf S}, \FqXL^T\FqV_2') + \cdots + H({\bf S}, \FqXL^T\FqV_5') \\
	&\leq H({\bf S})+H(\FqXL^T\FqV_1') + H({\bf S})+H(\FqXL^T\FqV_2') +\cdots \notag\\
	& + H({\bf S})+H(\FqXL^T\FqV_5')\\
	&\leq 5H({\bf S}) + 5m'L
\end{align}
We thus obtain
\begin{align}
	&5H({\bf S}) + 5m'L \aasgeq 30L+2H({\bf S}) \notag  \\
	 &\implies \Delta^*(\Lambda_n) \geq H({\bf S})/L \aasgeq (30-5m')/3=5m/3
\end{align}
as desired.

The general proof starts as follows. Let us start with a useful lemma, whose proof is relegated to Appendix \ref{sec:proof:lem:full_rank_matrix}.
\begin{lemma}\label{lem:full_rank_matrix}
Consider any $M' \in {\Bbb F}_{p^n}^{d\times \mu'}, (\mu' \leq d)$ that has full column rank $\mu'$. 
Let $M\in {\Bbb F}_{p^n}^{d\times \mu}$. If the elements of $M$ are chosen i.i.d uniform in ${\Bbb F}_{p^n}$, then $[M',M]$ has rank $\min\{d, \mu'+\mu\}$ a.a.s.
\end{lemma}

\noindent Define $\overline{m}$, and the constants $K_0,K_1,\cdots, K_{\overline{m}}$ as follows,
\begin{align}
\overline{m} &\triangleq \frac{m+m'}{\mbox{gcd}(d,m+m')},\\
K_i &\triangleq \left\lceil \frac{id}{m+m'} \right\rceil, ~~ \forall i\in[0:\overline{m}],
\end{align}
so that,
\begin{align}
&\mbox{1)}~K_0=0, ~K_{\overline{m}} = \left\lceil \frac{\overline{m}d}{m+m'} \right\rceil = \frac{d}{\mbox{gcd}(d,m+m')} \label{eq:Koverlinem},\\
&\mbox{2)}~(K_i-1)(m+m')<id\leq K_i(m+m'), \forall i\in[\overline{m}],\label{eq:Kim}\\
&\mbox{3)}~ K_i - K_{i-1}>0, ~\forall i\in[\overline{m}].
\end{align}
Define the matrices $\Insc_0,\cdots,\Insc_{\overline{m}}$, $\Un_1,\cdots,\Un_{\overline{m}}$, and $\Sp_1,\cdots,\Sp_{\overline{m}}$, as follows,
\begin{align}
&\Insc_0 \triangleq [~] \\
	& \Sp_1 \triangleq [\Insc_0, \FqU_{K_0+1},\FqU_{K_0+2},\hdots, \FqU_{K_1-1}], \notag\\
	& \Insc_1 \triangleq \Sp_1 \cap \FqU_{K_1},  ~~\Un_1 \triangleq [\Sp_1, \FqU_{K_1}], \\
	& \Sp_2 \triangleq [\Insc_1, \FqU_{K_1+1}, \FqU_{K_1+2}, \hdots, \FqU_{K_2-1}],\notag \\
	& \Insc_2 \triangleq \Sp_2 \cap \FqU_{K_2}, ~~\Un_2 \triangleq [\Sp_2, \FqU_{K_2}] \\
	& ~~~~~~ \vdots \notag \\
	& \Sp_{i+1} \triangleq [\Insc_{i}, \FqU_{K_{i}+1}, \FqU_{K_{i}+2}, \hdots, \FqU_{K_{i+1}-1}], \notag\\
	& \Insc_{i+1} \triangleq \Sp_{i+1} \cap \FqU_{K_{i+1}},~~ \Un_{i+1} \triangleq [\Sp_{i+1}, \FqU_{K_{i+1}}] \\
	& ~~~~~~ \vdots \notag \\
	& \Sp_{\overline{m}} \triangleq [\Insc_{\overline{m}-1}, \FqU_{K_{\overline{m}-1}+1}, \FqU_{K_{\overline{m}-1}+2}, \hdots, \FqU_{K_{\overline{m}}-1}], \notag\\
	& \Insc_{\overline{m}} \triangleq \Sp_{\overline{m}} \cap \FqU_{K_{\overline{m}}}, ~~\Un_{\overline{m}} \triangleq [\Sp_{\overline{m}}, \FqU_{K_{\overline{m}}}]
\end{align}
so that for all $i\in[\overline{m}]$,
\begin{align}
\Sp_i&\in\mathbb{F}_q^{d\times (\mbox{\footnotesize\bf rk}(\Insc_{i-1})+(K_i-K_{i-1}-1)(m+m'))},\label{eq:dimT}\\
\Un_i&\in\mathbb{F}_q^{d\times (\mbox{\footnotesize\bf rk}(\Insc_{i-1})+(K_i-K_{i-1})(m+m'))}.\label{eq:dimJ}
\end{align}
Define the event $E_n$ as follows,
\begin{align}
	E_n & \triangleq \Big( \rk(\FqV_k') = m',~~ \forall k\in [K]\Big) \land  \Big( \rk(\Un_i) = d, ~~ \forall i\in[\overline{m}]\Big).
\end{align}
The next steps \eqref{eq:rk_T_1}-\eqref{eq:Ii+1} show that $E_n$ holds a.a.s. 

From Lemma \ref{lem:full_rank_matrix} (let $M'=[~]$, $M = \FqV_k'$), we have $\rk(\FqV_k') \stackrel{\mbox{\tiny a.a.s.}}{=} m'$, 
since $m'\leq m+m' \leq d$. 
Similarly by Lemma \ref{lem:full_rank_matrix},  (letting $M'=[~]$ and $M=\Sp_1$, $M=\FqU_{K_1}$, $M=\Un_1$, respectively), we have
\begin{align}
	& \rk(\Sp_1) \aaseq (K_1-1)(m+m'), \label{eq:rk_T_1} \\
	&  \rk(\FqU_{K_1}) \aaseq m+m', \label{eq:rk_U_K1}\\
	& \rk(\Un_1) \aaseq d, \label{eq:rk_J_1} 
\end{align}
where \eqref{eq:rk_T_1} and \eqref{eq:rk_J_1} are due to \eqref{eq:Kim}, and \eqref{eq:rk_U_K1} follows from $m+m'\leq d$.
Then since $\rk(\Insc_1) = \rk(\Sp_1)+\rk(\FqU_{K_1}) - \rk(\Un_1)$, we have that
\begin{align} \label{eq:rank_I1}
	\rk(\Insc_1) \aaseq K_1(m+m')-d.
\end{align}
Next, to set up an inductive argument, suppose for some $i$, $1\leq i< \overline{m}$,
\begin{align}
	\rk(\Insc_i) \aaseq K_i(m+m')-id. \label{eq:Ii}
\end{align}
Conditioned on $\Big( \rk(\Insc_i) = K_i(m+m')-id \Big)$, from Lemma \ref{lem:full_rank_matrix} and \eqref{eq:Kim},\eqref{eq:dimT},\eqref{eq:dimJ} we have
\begin{align}
&\rk(\Sp_{i+1})\aaseq (K_{i+1}-1)(m+m')-id, \label{eq:rk_T_ip1}\\
&\rk(\Un_{i+1})\aaseq d, \label{eq:rk_J_ip1}\\
&\rk(\FqU_{K_{i+1}}) \aaseq m+m',\label{eq:rk_U_ip1}\\
&\rk(\Insc_{i+1}) \aaseq K_{i+1}(m+m')-(i+1)d.\label{eq:Ii+1}
\end{align}
where in order to obtain \eqref{eq:Ii+1}, we used the property $\rk(\Insc_{i+1}) = \rk(\Sp_{i+1}) +\rk(\FqU_{K_{i+1}}) - \rk(\Un_{i+1})$, along with 
\eqref{eq:rk_T_ip1}, \eqref{eq:rk_J_ip1} and \eqref{eq:rk_U_ip1}. By induction,  we obtain $\rk(\Un_i) \aaseq d, ~~ \forall i\in[\overline{m}]$, which implies that $E_n$ holds a.a.s.

Figure \ref{fig:ill_conv} may be useful in understanding the construction above and the proof. 
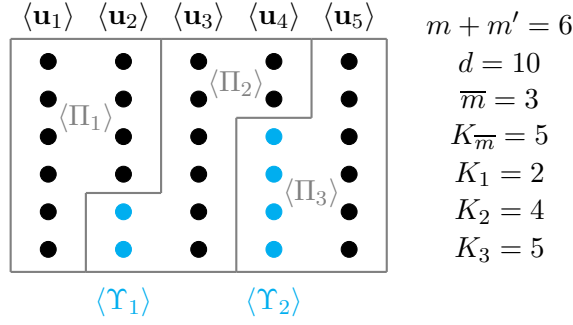
\begin{figure}[!h]
\center
\begin{tikzpicture}
	\filldraw (0,0) circle (3pt);
	\filldraw (0,0.5) circle (3pt);
	\filldraw (0,1) circle (3pt);
	\filldraw (0,1.5) circle (3pt);
	\filldraw (0,2) circle (3pt);
	\filldraw (0,2.5) circle (3pt);
	\filldraw[color=cyan]  (1,0) circle (3pt);
	\filldraw[color=cyan]  (1,0.5) circle (3pt);
	\filldraw (1,1) circle (3pt);
	\filldraw (1,1.5) circle (3pt);
	\filldraw (1,2) circle (3pt);
	\filldraw (1,2.5) circle (3pt);
	\filldraw (2,0) circle (3pt);
	\filldraw (2,0.5) circle (3pt);
	\filldraw (2,1) circle (3pt);
	\filldraw (2,1.5) circle (3pt);
	\filldraw (2,2) circle (3pt);
	\filldraw (2,2.5) circle (3pt);
	\filldraw[color=cyan]  (3,0) circle (3pt);
	\filldraw[color=cyan]  (3,0.5) circle (3pt);
	\filldraw[color=cyan]  (3,1) circle (3pt);
	\filldraw[color=cyan]  (3,1.5) circle (3pt);
	\filldraw (3,2) circle (3pt);
	\filldraw (3,2.5) circle (3pt);
	\filldraw (4,0) circle (3pt);
	\filldraw (4,0.5) circle (3pt);
	\filldraw (4,1) circle (3pt);
	\filldraw (4,1.5) circle (3pt);
	\filldraw (4,2) circle (3pt);
	\filldraw (4,2.5) circle (3pt);
	\draw [draw=gray, thick] (-0.5,-0.3)--(-0.5,2.8);
	\draw [draw=gray, thick] (-0.5,2.8)--(1.5,2.8);
	\draw [draw=gray, thick] (1.5,2.8)--(1.5,0.75);
	\draw [draw=gray, thick] (0.5,0.75)--(1.5,0.75);
	\draw [draw=gray, thick] (0.5,0.75)--(0.5,-0.3);
	\draw [draw=gray, thick] (-0.5,-0.3)--(0.5,-0.3);
	\draw [draw=gray, thick] (0.5,-0.3)--(2.5,-0.3);
	\draw [draw=gray, thick] (2.5,-0.3)--(2.5,1.75);
	\draw [draw=gray, thick] (2.5,1.75)--(3.5,1.75);
	\draw [draw=gray, thick] (3.5,1.75)--(3.5,2.8);
	\draw [draw=gray, thick] (3.5,2.8)--(1.5,2.8);
	\draw [draw=gray, thick] (2.5,-0.3)--(4.5,-0.3);
	\draw [draw=gray, thick] (4.5,-0.3)--(4.5,2.8);
	\draw [draw=gray, thick] (4.5,2.8)--(3.5,2.8);
	\node (I1) at (1,-0.7)  []   {\color{cyan} $\langle\Insc_1\rangle$};
	\node (I2) at (3,-0.7)  []   {\color{cyan} $\langle\Insc_2\rangle$};
	\node at (0,3.1)  []   { $\langle\FqU_1\rangle$};
	\node at (1,3.1)  []   { $\langle\FqU_2\rangle$};
	\node at (2,3.1)  []   { $\langle\FqU_3\rangle$};
	\node at (3,3.1)  []   { $\langle\FqU_4\rangle$};
	\node at (4,3.1)  []   { $\langle\FqU_5\rangle$};
	\node at (0.5,1.75)  []   {\color{gray} $\langle \Un_1 \rangle$};
	\node at (2.5,2.2)  []   {\color{gray} $\langle\Un_2\rangle$};
	\node at (3.5,0.75)  []   {\color{gray} $\langle\Un_3\rangle$};
	\node at (6,3) [] {$m+m' = 6$};
	\node at (6,2.5) [] {$d=10$};
	\node at (6,2) [] {$\overline{m}=3$};
	\node at (6,1.5) [] {$K_{\overline{m}}=5$};
	\node at (6,1)  []   { $K_1 = 2$};
	\node at (6,0.5)  []   { $K_2 = 4$};
	\node at (6,0)  []   { $K_3 = 5$};

\end{tikzpicture}	
\caption{\small Illustration of the converse proof for $m+m' = 6, d=10$. Each  dot represents one dimension. The number of dots represents the dimension for the corresponding space a.a.s. $\langle \FqU_k \rangle, k\in[5]$ has dimension $6$ a.a.s. $\langle\Insc_1\rangle$ is the intersection of $\langle\FqU_1\rangle$ and $\langle\FqU_2\rangle$, which has dimension $2$ a.a.s. $\langle\Insc_2\rangle$ is the intersection of $\langle[\Insc_1,\FqU_3]\rangle$ and $\langle\FqU_4\rangle$, which has dimension $4$ a.a.s. $\langle\Un_1\rangle$ is the union of $\langle\FqU_1\rangle$ and $\langle\FqU_2\rangle$, which has dimension $10$ a.a.s. $\langle\Un_2\rangle$ is the union of $\langle[\Insc_1, \FqU_3]\rangle$ and $\langle\FqU_4\rangle$, which has dimension $10$ a.a.s. $\langle \Un_3\rangle$ is the union of $\langle\Insc_2\rangle$ and $\langle \FqU_5\rangle$, which has dimension $10$ a.a.s.}\label{fig:ill_conv}
\end{figure}

For the next stage of the proof, we consider any given LCBC$\big({\Bbb F}_{p^n},  \FqV_{[K]}, \FqV_{[K]}' \big)$ where $E_n$ holds. Note that $\FqV_{[K]}, \FqV_{[K]}'$ are now arbitrary \emph{constant} matrices that satisfy $E_n$. Our goal now is to bound the optimal broadcast cost $\Delta^*\big({\Bbb F}_q,  \FqV_{[K]}, \FqV_{[K]}' \big)$. Let  $\FqXL$ be uniformly distributed to facilitate entropic accounting. Recall that there is no loss of generality in this assumption, because $\FqXL$ is independent of $E_n$ and any achievable scheme must work for all data realizations, so it must also work for all data distributions. Thus, we have,
{\small
\begin{align}
2\overline{m}  & H({\bf S}) + K_{\overline{m}} Lm'\\
	\stackrel{(a)}{=}&~2\overline{m}H({\bf S}) + \sum_{k=1}^{K_{\overline{m}}}H(\FqXL^T\FqV_k') \\
	=&\sum_{i=0}^{\overline{m}-1}\Big[\Big(H({\bf S})+H(\FqXL^T\FqV'_{K_i+1})+H(\FqXL^T\FqV'_{K_i+2})+ \hdots  +H(\FqXL^T\FqV'_{K_{i+1}-1})\Big)+\Big(H({\bf S})+H(\FqXL^T\FqV'_{K_{i+1}})\Big)\Big]\notag\\
	\geq &~\sum_{i=0}^{\overline{m}-1} \Big(H({\bf S}, \FqXL^T[\FqV'_{K_i+1},\FqV'_{K_i+2},\hdots, \FqV'_{K_{i+1}-1}]) + H({\bf S}, \FqXL^T\FqV'_{K_{i+1}}) \Big)\\
	\stackrel{(b)}{=}&~ \sum_{i=0}^{\overline{m}-1} \Big( H({\bf S},\FqXL^T[\FqU_{K_i+1}, \FqU_{K_i+2}, \hdots, \FqU_{K_{i+1}-1}]) +H({\bf S},\FqXL^T\FqU_{K_{i+1}})  \Big) \\
	= &~  H({\bf S},\FqXL^T[\FqU_{1}, \FqU_{2}, \hdots, \FqU_{K_{1}-1}]) +H({\bf S},\FqXL^T\FqU_{K_{1}}) \notag \\
	&+ \sum_{i=1}^{\overline{m}-1} \Big( H({\bf S},\FqXL^T[\FqU_{K_i+1}, \FqU_{K_i+2}, \hdots, \FqU_{K_{i+1}-1}]) + H({\bf S},\FqXL^T\FqU_{K_{i+1}})  \Big) \\
	=&~ H({\bf S}, \FqXL^T\Sp_1) + H({\bf S}, \FqXL^T\FqU_{K_1})  + \sum_{i=1}^{\overline{m}-1} \Big( H({\bf S},\FqXL^T[\FqU_{K_i+1}, \FqU_{K_i+2}, \hdots, \FqU_{K_{i+1}-1}]) +H({\bf S},\FqXL^T\FqU_{K_{i+1}})  \Big) \\
	 \stackrel{(c)}{\geq} &~ H(\FqXL^T\Un_1) + H({\bf S},\FqXL^T\Insc_1) + \sum_{i=1}^{\overline{m}-1} \Big( H({\bf S},\FqXL^T[\FqU_{K_i+1}, \FqU_{K_i+2}, \hdots, \FqU_{K_{i+1}-1}]) +H({\bf S},\FqXL^T\FqU_{K_{i+1}})  \Big) \\
	 = &~ H(\FqXL^T\Un_1) + H({\bf S},\FqXL^T\Insc_1) +  H({\bf S},\FqXL^T[\FqU_{K_1+1}, \FqU_{K_1+2}, \hdots, \FqU_{K_{2}-1}]) +H({\bf S},\FqXL^T\FqU_{K_{2}})  \notag \\
	 &+ \sum_{i=2}^{\overline{m}-1} \Big( H({\bf S},\FqXL^T[\FqU_{K_i+1}, \FqU_{K_i+2}, \hdots, \FqU_{K_{i+1}-1}]) +H({\bf S},\FqXL^T\FqU_{K_{i+1}})  \Big) \\
	 \stackrel{(c)}{\geq} &~  \Big( H(\FqXL^T\Un_1) + H({\bf S})\Big) + H({\bf S},\FqXL^T\Sp_2) + H({\bf S},\FqXL^T \FqU_{K_2}) \notag \\
	 &+ \sum_{i=2}^{\overline{m}-1} \Big( H({\bf S},\FqXL^T[\FqU_{K_i+1}, \FqU_{K_i+2}, \hdots, \FqU_{K_{i+1}-1}]) +H({\bf S},\FqXL^T\FqU_{K_{i+1}})  \Big) \\
	 \vdots~ & \notag \\
	 \geq &  ~ \Big( H(\FqXL^T\Un_1) + H({\bf S}) \Big) + \hdots + \Big( H(\FqXL^T\Un_{\overline{m}-1}) + H({\bf S}) \Big) + H({\bf S}, \FqXL^T\Sp_{\overline{m}}) + H({\bf S},\FqXL^T\FqU_{K_{\overline{m}}}) \\
	 \stackrel{(c)}{\geq} & ~ \Big( H(\FqXL^T\Un_1) + H({\bf S}) \Big) + \hdots   + \Big( H(\FqXL^T\Un_{\overline{m}}) + H({\bf S}) \Big) \\
	 \stackrel{(a)}{\geq} & ~ \overline{m}\big(L d+ H({\bf S}) \big)  \\
	 \implies & \overline{m}H({\bf S}) + K_{\overline{m}} Lm' \geq \overline{m}Ld 
\end{align}
}Steps labeled $(a)$ hold because $E_n$ holds. Steps labeled $(b)$ follow from the decodability constraint, i.e., $H({\bf S}, \FqXL^T\FqU_k) = H({\bf S}, \FqXL^T\FqV_k')$. Steps labeled $(c)$ use functional submodularity (Lemma 1 of \cite{Yao_Jafar_3LCBC}, \cite{Tao_FS,Kontoyiannis_Madiman}). 

Note that \eqref{eq:Koverlinem} implies that 
\begin{align}
	\frac{K_{\overline{m}}}{\overline{m}} = \frac{d}{m+m'}.
\end{align}
Thus, we obtain that 
\begin{align}
	\Delta \geq \frac{H({\bf S})}{L} \geq d - \frac{K_{\overline{m}}}{\overline{m}}m' = d- \frac{d}{m+m'}m' = \frac{md}{m+m'}.
\end{align}
This in turn implies that for any $\varepsilon>0$,
\begin{align}
&\lim_{n\rightarrow\infty}\mathsf{Pr}\left(\Delta^*(\Lambda_n)>\frac{md}{m+m'}-\epsilon\right)\\
&\geq \lim_{n\rightarrow\infty}\mathsf{Pr}\left(\left[\Delta^*(\Lambda_n)>\frac{md}{m+m'}-\epsilon\right]\land E_n\right)\\
&= \lim_{n\rightarrow\infty}\mathsf{Pr}(E_n)\mathsf{Pr}\Big(\Delta^*(\Lambda_n)>\frac{md}{m+m'}-\epsilon~\Big|~ E_n\Big)\\
&=1.
\end{align}
Thus $\Delta^*(\Lambda_n)\aasg md/(m+m')-\varepsilon$. Since this is true for all $\varepsilon>0$, according to \eqref{eq:defdeltal} we have $\Delta_l^*\geq md/(m+m')$.$\hfill\square$

\section{Proof of Theorem \ref{thm:one_dim}}\label{sec:proof_thm_one_dim}
In this section, let us show the converse for odd $d$ with $3\leq d < 2K-1$, and the achievability for $d = 2K-1$.
\subsection{Converse for odd $d$ with $3\leq d < 2K-1$}
The condition for this regime is equivalent to $2\leq \frac{d+1}{2}<K$. Since the generic capacity for this regime is only a function of $d$, and assuming a smaller $K$ will not hurt the converse, it suffices to show the converse for $K=\frac{d+1}{2}+1$, that is, $\Delta^*_l \geq d/2$.
We start with the following lemma.
\begin{lemma} \label{lem:3user}
	For LCBC(${\Bbb F}_q, \FqV_{[K]}, \FqV_{[K]}'$), the broadcast cost $\Delta$ satisfies,
	\begin{align} \label{eq:3user}
		\Delta \geq \frac{1}{2} &\Bigg( \rk(\FqV_1' \cap [\FqU_1\cap \FqU_{\mathcal{K}_1}, \FqU_1\cap \FqU_{\mathcal{K}_2}]) + \rk([\FqU_1, \FqU_{\mathcal{K}_1}])+ \rk([\FqU_1, \FqU_{\mathcal{K}_2}])  \notag\\
		&~~~~ - 2\rk(\FqV_1')-\rk(\FqV'_{\mathcal{K}_1})-\rk(\FqV'_{\mathcal{K}_2}) \Bigg),
	\end{align}
	where $\mathcal{K}_1$ and $\mathcal{K}_2$ are subsets of $[K]$.
\end{lemma}

\begin{proof}
	For simplicity, we will make use of the converse in\cite{Yao_Jafar_ISIT22}. Denote the original LCBC as $\Lambda$ and its optimal download cost as $\Delta^*(\Lambda)$. Now, consider another LCBC setting $L'$ with $3$ users, where User $1$ has side information $\FqX^T\FqV_1'$ and desires $\FqX^T\FqV_1$; User $2$ has side information $\FqX^T\FqV_{\mathcal{K}_1}'$ and desires $\FqX^T\FqV_{\mathcal{K}_1}$; User $3$ has side information $\FqX^T\FqV_{\mathcal{K}_2}'$ and desires $\FqX^T\FqV_{\mathcal{K}_2}$. Denote by $\Delta^*(\Lambda')$ the optimal download cost of $\Lambda'$. We have $\Delta^*(\Lambda) \geq \Delta^*(\Lambda')$ since for any scheme that works for $\Lambda$, we can construct another scheme that works for $\Lambda'$ with a same download cost by letting the users in $\mathcal{K}_1$ cooperate, and the users in $\mathcal{K}_2$ cooperate. Note that $\mathcal{K}_1$ and $\mathcal{K}_2$ can have a non-empty intersection. Also note that although the capacity result in \cite{Yao_Jafar_ISIT22} is only for $m_k=m_k'=1, \forall k\in[3]$, the converse holds for any LCBC with $3$ users. Now let us make use of the converse in \cite{Yao_Jafar_ISIT22}. Since $\Delta^*(\Lambda) \geq \Delta^*(\Lambda')$, by (7) of \cite{Yao_Jafar_ISIT22}, and by rearranging the terms, we have the desired bound.
\end{proof}

Next we show the converse for the generic capacity, i.e., $\Delta_l^* \geq d/2$ for $K=\frac{d+1}{2}+1$. Note that $\frac{d+1}{2}+1 = 3+\frac{d-3}{2}$. Let $\mathcal{K}_0 = [K-\frac{d-3}{2}+1:K]$, $\mathcal{K}_1 = \{2\} \cup \mathcal{K}_0$ and $\mathcal{K}_2 = \{3\}\cup \mathcal{K}_0$. Note that $\mathcal{K}_0 = \emptyset$ if $d=3$. Since the converse bound in Lemma \ref{lem:3user} is composed of ranks of certain matrices, we then finds these ranks in the a.a.s. sense.

First, since $2(1+|\mathcal{K}_1|) = 2(1+|\mathcal{K}_2|) = d+1>d$, by Lemma \ref{lem:full_rank_matrix},
\begin{align}
	\rk([\FqU_1, \FqU_{\mathcal{K}_1}]) \aaseq d, ~~ \rk([\FqU_1, \FqU_{\mathcal{K}_2}]) \aaseq d,
\end{align}
and
\begin{align}
	\rk(\FqV_1') \aaseq 1, ~~ \rk(\FqV'_{\mathcal{K}_1}) \aaseq \frac{d-1}{2}, ~~\rk(\FqV'_{\mathcal{K}_2}) \aaseq \frac{d-1}{2}.
\end{align}
This leaves us the only non-trivial term, $\rk(\FqV_1' \cap [\FqU_1\cap \FqU_{\mathcal{K}_1}, \FqU_1\cap \FqU_{\mathcal{K}_2}])$. To avoid complex notations, within this section let
\begin{align}
	&A = \FqU_1 \in {\Bbb F}_q^{d\times 2}, \\
	&B = \FqU_{\mathcal{K}_1} = [\FqU_{\mathcal{K}_0},\FqU_2] \in {\Bbb F}_q^{d\times (d-1)}, \\
	&C = \FqU_{\mathcal{K}_2} = [\FqU_{\mathcal{K}_0},\FqU_3] \in {\Bbb F}_q^{d\times (d-1)}.
\end{align}
Then let
\begin{align}
	&D = [A,{\bf 0}^{d\times (d-2)}] \big([A,B_{[1:d-2]}]\big) ^* B_{[d-1]} \in {\Bbb F}_q^{d\times 1}, \\
	&E = [A,{\bf 0}^{d\times (d-2)}]\big([A,C_{[1:d-2]}]\big)^* C_{[d-1]} \in {\Bbb F}_q^{d\times 1},
\end{align}
where $M^*$ denotes the adjugate matrix of a square matrix $M$ such that $MM^* = \det(M){\bf I}$. By construction, we have $\langle D \rangle \subset \langle A \rangle$ and $\langle E\rangle \subset \langle A \rangle$. Then note that
\begin{align}
	D&+[{\bf 0}^{d\times 2}, B_{[1:d-2]}]\big([A,B_{[1:d-2]}]\big) ^* B_{[d-1]} = \det([A,B_{[1:d-2]}]) B_{[d-1]}.
\end{align}
We obtain that $\langle D \rangle \subset \langle B \rangle$. Similarly, we have $\langle E \rangle \subset \langle C \rangle$. Therefore, $\langle D \rangle \subset \langle A\cap B \rangle$ and $\langle E \rangle \subset \langle A\cap C \rangle$. Next, let $Z\in {\Bbb F}_q^{d\times (d-2)}$. We claim that
\begin{align}
	P = \det([D,E,Z])
\end{align}
is a non-zero polynomial in the elements of ${\bf u}_{[K]},Z$. To see this, specify
\begin{align}
	&{\bf u}_1 = {\bf I}^{d\times d}_{[1:2]}, ~~ \FqU_{\mathcal{K}_0} = {\bf I}_{[3:d-1]},\\
	&\FqU_2 = [{\bf I}^{d\times d}_{[d]},{\bf I}^{d\times d}_{[1]}],\\
	&\FqU_3 = [{\bf I}^{d\times d}_{[d]},{\bf I}^{d\times d}_{[2]}], \\
	& Z = {\bf I}^{d\times d}_{[3:d]}.
\end{align}
We then have
\begin{align}
	[A,B_{[1:d-2]}] = [A,C_{[1:d-2]}] = {\bf I}^{d\times d},
\end{align}
and it follows that
\begin{align}
	D = {\bf I}_{[1]}, ~~ E = {\bf I}_{[2]} \implies \det([D,E,Z]) = \det({\bf I}^{d\times d}) = 1.
\end{align}
Therefore, $P$ is a non-zero polynomial, with degree not more than $2(d+1)+(d-2)$. By Schwartz-Zippel Lemma, the probability of $P$ evaluating to a non-zero value is not less than
\begin{align}
	1-\frac{2(d+1)+(d-2)}{p^n} 
\end{align}
which approaches $1$ as $n\to \infty$. Thus,
\begin{align}
	\rk([D,E]) \aasgeq 2 \implies \rk([\FqU_1\cap \FqU_{\mathcal{K}_1},\FqU_1\cap \FqU_{\mathcal{K}_2}]) \aasgeq 2.
\end{align}
Since $\langle [\FqU_1\cap \FqU_{\mathcal{K}_1},\FqU_1\cap \FqU_{\mathcal{K}_2}]\rangle \subset \langle \FqU_1 \rangle$, and $\rk(\FqU_1) \leq 2$, we have that
\begin{align}
	\langle [\FqU_1\cap \FqU_{\mathcal{K}_1},\FqU_1\cap \FqU_{\mathcal{K}_2}]\rangle \aaseq \langle \FqU_1 \rangle.
\end{align}
Since $\langle \FqV_1' \rangle \subset \langle \FqU_1 \rangle$, we obtain that
\begin{align}
	\rk(\FqV_1'\cap [\FqU_1\cap \FqU_{\mathcal{K}_1},\FqU_1\cap \FqU_{\mathcal{K}_2}]) \aaseq \rk(\FqV_1') \aaseq 1.
\end{align}
Now let us consider \eqref{eq:3user} in the a.a.s. sense. We have
\begin{align}
	\Delta_n \aasgeq \frac{1+2d-2-(d-1)}{2} = \frac{d}{2},
\end{align}
or, equivalently,
\begin{align}
	\Delta^*_l \geq \frac{d}{2},
\end{align}
which proves the desired converse. $\hfill \qed$

\subsection{Achievability for odd $d$ with $d = 2K-1$}
Consider the following $K$ matrices
\begin{align}
	M_k = [\FqV_{[K]}',\FqV_{[K]\backslash \{k\}}] \in {\Bbb F}_q^{d\times d}, ~~ k\in [K].
\end{align}
By Lemma \ref{lem:full_rank_matrix}, we have,
\begin{align} \label{eq:full_rank_K}
	\rk(M_k) \aaseq d, ~~ \forall k\in [K].
\end{align}
For any LCBC instance that satisfies $\rk(M_k) = d, ~~ \forall k\in [K]$, we are able to find non-zero ${\Bbb F}_q$ elements $\alpha_1,\alpha_2,...,\alpha_{K-1}$ and  ${\Bbb F}_q$ elements $\alpha_1',\alpha_2',...,\alpha_{K}'$ such that
\begin{align}
	\FqV_{K} = \sum_{k=1}^{K-1} \alpha_k \FqV_k + \sum_{k=1}^K \alpha_k' \FqV_k'.
\end{align}
To see this, first note that since $\rk(M_K) = d$, $\FqV_{K}$ can be represented by a linear combination of the $2K-1$ vectors $\FqV_1,\FqV_2,...,\FqV_{K-1},\FqV_1',\FqV_2',...,\FqV_K'$. Now let us show that the coefficients $\alpha_1,\alpha_2,...,\alpha_{K-1}$ are non-zero. We prove by contradiction. Suppose $\alpha_i=0$. We then have
\begin{align}
	\sum_{k\in[1:K-1]\backslash\{i\}} \alpha_k \FqV_k + \sum_{k=1}^K \alpha_k' \FqV_k' - \FqV_K = 0,
\end{align}
which implies that $M_i$ does not have full column rank $d$. This contradiction proves that $\alpha_k\not=0, \forall k\in[1:K-1]$. Let the batch size $L=1$. The server broadcasts ${\bf S} = {\bf S}_{[K-1]}$, where
\begin{align}
	{\bf S}_{k} = \FqX^T (\alpha_k \FqV_k + \alpha_k' \FqV_k') \in {\Bbb F}_q, ~~ k\in[K-1].
\end{align}
User $k, k\in[K-1]$ can get its desired computation by $\FqX^T\FqV_k =  (1/{\alpha_k})({\bf S}_{k} - \alpha_k'\FqX^T\FqV_k')$. User $K$ can get its desired computation by $\FqX^T\FqV_K = \sum_{k=1}^{K-1}{\bf S}_k + \alpha_K' \FqX^T\FqV_K'$. The broadcast cost is then $\Delta = K-1$. By \eqref{eq:full_rank_K}, we conclude that
\begin{align}
	\Delta^*_u \leq K-1.
\end{align}
which is the desired upper bound. $\hfill \qed$

\section{Conclusion}
The take home message of this work is optimistic. While a general capacity characterization of the LCBC for large number of users remains out of reach because it includes recognized hard problems such as index coding, a generic capacity characterization is shown to be tractable. As such, the LCBC setting that generalizes index coding, combined with the generic capacity formulation that focuses on \emph{almost all} instances of the LCBC, presents a promising path forward for future progress.  This is  analogous to DoF studies of wireless networks where much progress has come about by focusing on generic settings (‘almost all’ channel realizations rather than ‘all’ channel realizations) while the DoF of arbitrary instances still remain largely open. 

The latter limitation is  worth emphasizing. While a generic capacity characterization reveals the capacity of most LCBC settings, it is notable that the LCBC settings that have received the most attention thus far, say index coding and coded caching for example,  have specialized demand and side-information structures that are not generic. Thus, open questions in index coding and caching remain open and as important as ever for future work. The study of generic capacity is not intended to supersede the studies of caching, index coding or other specialized applications, but to complement those efforts with an understanding of what is missed in the study of specializations --- the generic case. Understanding the capacity limits for structureless, i.e., generic side-information and demands is especially important because the scope of possible linear computation scenarios that may arise in future  applications is far too broad to be understood through studies of specialized structures alone. For example, arbitrary linear filters may be applied by different users on large datasets held by a central server, with side-information arising from previously retrieved outputs of other filtering operations on the same datasets. Depending on the application, there may be little or no freedom to optimize the structure of the demand and side-information. Also, a theory cannot be built out of special cases while ignoring the generic case. So if a cohesive information theoretical understanding of communication networks used for computation tasks is to ever emerge, the generic case has to be at its foundation.

Promising directions for future work include the exploration of generic capacity for asymmetric settings, analysis of the LCBC download cost vs complexity tradeoff, and generic capacity in the large $q$ sense (especially for $n=1$).  Extensions of finite field results to degrees of freedom (DoF) results over real/complex numbers, and studies of the tradeoffs between precision and communication cost in the GDoF sense (as in \cite{Price_Precision}) are promising as well.
Last but not the least, while the capacity results  in this work establish the information theoretic fundamental limits, asymptotic IA schemes are far from practical. Therefore, the extent to which  the fundamental limit can be approached with \emph{practical} coding schemes, is a most interesting open question where future coding-theoretic analysis can shed light.

\appendix

\section{Theorem \ref{thm:large_K}: Proof of Achievability} \label{sec:proof_large_K}
Let us recall the compact notation $\FqU_k \triangleq [\FqV_k',\FqV_k],\forall k\in [K]$. For $d\geq K(m+m')$, the broadcast cost $Km$ is trivially achievable, simply by broadcasting each user's demand separately, i.e., ${\bf S}=\FqX^T[\FqV_1, \hdots, \FqV_K]$. The achievability for the remaining regimes is shown next.
\subsection{Achievability for $d\leq m+m'$} \label{subsec:achi_trivial}
Define the event 
\begin{align}
E_n\triangleq\Big(\rk(\FqV_k') = \min\{m',d\}\Big)\land\Big(\rk(\FqU_k) = d\Big).
\end{align}
In Lemma \ref{lem:full_rank_matrix}, letting $M'=[~]$, $M=\FqV_k'$ and $M=\FqU_k$, respectively, we obtain that $E_n$ holds a.a.s. The following argument is true if $E_n$ holds.
\begin{enumerate}
	\item If $d\leq m'$, each of the $K$ users is able to compute $\FqX$, since $\rk(\FqV_k') = d$. This implies that $\Delta^* = 0$.
	\item Using  field extensions (cf. Appendix B of \cite{Yao_Jafar_3LCBC}), let us consider the equivalent LCBC with field size $q^z = p^{nz}$. If $m'< d\leq m+m'$, for each $k\in [K]$, let $\FqU^{c}\in {\Bbb F}_{q^z}^{d\times (d-m')}$. We claim that $P_k = \det([\FqU_k, \FqU^c])$ is a non-zero polynomial in the elements of $\FqU_k, \FqU^c$. To see this, for each $k$, we can choose $\FqU^{c}= ({\bf I}^{d\times d}|\FqU_k) \in \mathbb{F}_{q^z}^{d\times (d-m')}$ such that $[\FqU_k, \FqU^c]$ spans $\langle {\bf I}^{d\times d} \rangle$. It follows that $P=\prod_k P_k$ is a non-zero polynomial in the elements of $\FqU_{[K]},\FqU^c$. By Schwartz-Zippel Lemma, if the elements of $\FqU^c$ are chosen uniformly in $\mathbb{F}_{q^z}$, the probability of $P$ evaluating to zero is not more than $\frac{\mbox{degree of }P}{q^z}\leq\frac{K(d-m')}{q^z}$. Thus, by choosing $z>\log_q(K(d-m'))$, we ensure that there exists such $\FqU^c$ that satisfies $\det([\FqU_k, \FqU^c])\not=0$ for all $k\in[K]$. Broadcasting ${\bf S} = \FqXL^T\FqU^c$, we have $\Delta = d-m'$, and each User $k$ is able to compute $\FqXL$ with ${\bf S}$ and its side information $\FqXL^T\FqV_k'$. 
\end{enumerate}
Thus we have the desired achievability, $\Delta^*(\Lambda_n) \aasleq \max\{0, d-m'\}$.

\subsection{Achievability for $(m+m')<d< K(m+m')$} \label{app:achi_non_trivial}
Let $L=1$. For $q=p^n$, we will interpret ${\Bbb F}_q$ as an $n$-dimensional vector space over ${\Bbb F}_p$, and design a linear scheme over ${\Bbb F}_p$. Accordingly, let us clarify the notation as follows.
\begin{enumerate}
	\item The elements of the data and coefficient matrices are chosen from ${\Bbb F}_{q} = {\Bbb F}_{p^n}$.
	\item The data $\FqX^T = [\Fqx_1,\Fqx_2,\hdots,\Fqx_d]\in {\Bbb F}_q^{1\times d}$, is  equivalently represented over  $\mathbb{F}_p$ as $\FpX^T = [\Fpx_1^T,\Fpx_2^T,$ $\hdots,\Fpx_d^T]$ $\in\mathbb{F}_p^{1\times dn}$, where $\Fpx_i\in\mathbb{F}_p^{n\times 1}$ is the $n\times 1$ vector representation of $\Fqx_i$ over $\mathbb{F}_p$.
	\item User $k$ has side information $\FqX^T\FqV_k' \in {\Bbb F}_q^{1\times m'}$ and wishes to compute $\FqX^T \FqV_k\in {\Bbb F}_q^{1\times m}$, where
		\begin{align}
			\FqV_k &= 
			\begin{bmatrix}
				\FqV_k^1 & \FqV_k^2 & \hdots & \FqV_k^m
			\end{bmatrix} = \begin{bmatrix}
				\Fqv_{k1}^1 & \Fqv_{k1}^2 & \hdots & \Fqv_{k1}^{m} \\ 
				\Fqv_{k2}^1 & \Fqv_{k2}^2 & \hdots & \Fqv_{k2}^{m} \\ 
				\vdots & \vdots & \ddots & \vdots			\\
				\Fqv_{kd}^1 & \Fqv_{kd}^2 & \hdots & \Fqv_{kd}^{m}
			\end{bmatrix}~~ \in {\Bbb F}_q^{d\times m}
		\end{align}
		and
		\begin{align}
			\FqV_k' &= 
			\begin{bmatrix}
				\FqV_k'^1 & \FqV_k'^2 & \hdots & \FqV_k'^{m'}
			\end{bmatrix} = \begin{bmatrix}
				\Fqv_{k1}'^1 & \Fqv_{k1}'^2 & \hdots & \Fqv_{k1}'^{m'} \\ 
				\Fqv_{k2}'^1 & \Fqv_{k2}'^2 & \hdots & \Fqv_{k2}'^{m'} \\ 
				\vdots & \vdots & \ddots & \vdots			\\
				\Fqv_{kd}'^1 & \Fqv_{kd}'^2 & \hdots & \Fqv_{kd}'^{m'}
			\end{bmatrix}~~ \in {\Bbb F}_q^{d\times m'}
		\end{align}
	\item Let $\Fqr_{1},\Fqr_2,\hdots,\Fqr_{m'}, \Fqs_{1},\Fqs_{2},\hdots, \Fqs_{m}$ be chosen i.i.d uniformly in ${\Bbb F}_q$.
	\item Define the set of variables,
		\begin{align}
			\mathcal{V} &\triangleq \Big\{ \Fqv_{ki}^{j}: k\in [K], i\in [d], {j}\in[m] \Big\} \cup \Big\{ \Fqv_{ki}'^{j'}: k\in [K], i\in [d], j'\in[m'] \Big\} \notag \\
			&~~~~\cup \Big\{ \Fqr_{1},\Fqr_2,\hdots,\Fqr_{m'}, \Fqs_{1},\Fqs_{2},\hdots, \Fqs_{m} \Big\},
		\end{align}
		and note that $|\mathcal{V}| = (m+m')(dK+1)$.
	\item We will also introduce the corresponding $n\times n$ matrix representations in ${\Bbb F}_p$ for several ${\Bbb F}_q$ variables (some of the ${\Bbb F}_q$ variables will be introduced later). The following table specifies them.
		\begin{center}
		\begin{tabular}{c|c|c}
		\toprule 
		${\Bbb F}_q$ variable & \begin{tabular}{c} ${\Bbb F}_p^{n\times n}$ matrix\\ representation \end{tabular} & Comment  \\ \hline
		$\Fqv_{ki}^j$ & $\Fpv_{ki}^j$ & {\small \begin{tabular}{c} $k\in[K]$, \\ $i\in [d]$, \\ $j\in [m]$ \end{tabular}} \\ \hline
		$\Fqv_{ki}'^{j'}$ & $\Fpv_{ki}'^{j'}$ & {\small \begin{tabular}{c} $k\in[K]$,\\$i\in [d]$, \\ $j'\in [m']$ \end{tabular}} \\ \hline
		$\Fqr_{j'}$ & $\Fpr_{j'}$ & $j'\in [m']$ \\ \hline
		$\Fqt_{k i}^{\mu j'}$ & $\Fpt_{k i}^{\mu j'}$ & {\small \begin{tabular}{c} $k\in[K]$, \\$i\in[d]$, \\ $\mu\in[m]$,\\ $j'\in[m']$ \end{tabular}} \\ \hline
		$\Fqh_{\mu}^{j}$ & $\Fph_{\mu}^{j}$ & {\small \begin{tabular}{c} $j\in[\eta]$, \\$\mu \in [m]$ \end{tabular}} \\ \hline
		$\overline{\Fqh}_{\mu}^{j}$ & $\overline{\Fph}_{\mu}^{j}$ & {\small \begin{tabular}{c} $j\in[\overline{\eta}]$, \\$\mu \in [m]$ \end{tabular}} \\  \bottomrule
		\end{tabular}
		\end{center}
\end{enumerate}
Our goal is to show that $\Delta^*_u \leq dm/(m+m')$. Note that if $d=0$ or $m=0$, $\Delta=0$ is trivially achieved for all cases. If $m'=0$, then $\Delta\leq d$ is trivially achieved for all cases by broadcasting ${\bf S} = {\bf X}$.  Thus, in the following we consider the cases when $d>0,m>0,m'>0.$
First, for $k\in[K], \mu \in[m], j'\in[m']$, let us define 
\begin{align} \label{eq:def_vpvt}
	\begin{bmatrix}
		\Fqv_{k1}'^{j'} \\ \Fqv_{k2}'^{j'} \\ \vdots \\ \Fqv_{kd}'^{j'}
	\end{bmatrix}
	=
	\begin{bmatrix}
		\Fqv_{k1}^{\mu} \\ \Fqv_{k2}^{\mu} \\ \vdots \\ \Fqv_{kd}^{\mu}
	\end{bmatrix}\Fqr_{j'}
	+
	\begin{bmatrix}
		\Fqt_{k1}^{\mu j'} \\ \Fqt_{k2}^{\mu j'} \\ \vdots \\ \Fqt_{kd}^{\mu j'}
	\end{bmatrix}.
\end{align}
We have
\begin{align} \label{eq:def_VpVT}
	\underbrace{\begin{bmatrix}
		\Fpv_{k 1}'^{j'} \\ \Fpv_{k 2}'^{j'} \\ \vdots \\ \Fpv_{k d}'^{j'}
	\end{bmatrix}}_{\FpV_k'^{j'}\in {\Bbb F}_p^{dn\times n}}
	=
	\underbrace{\begin{bmatrix}
		\Fpv_{k 1}^{\mu} \\ \Fpv_{k 2}^{\mu} \\ \vdots \\ \Fpv_{k d}^{\mu}
	\end{bmatrix}}_{\FpV_k^\mu \in {\Bbb F}_p^{dn\times n}}
	\Fpr_{j'}
	+
	\underbrace{\begin{bmatrix}
		\Fpt_{k 1}^{\mu j'} \\ \Fpt_{k 2}^{\mu j'} \\ \vdots \\ \Fpt_{k d}^{\mu j'}
	\end{bmatrix}}_{\FpT_k^{\mu j'} \in {\Bbb F}_p^{dn\times n} }
\end{align}
by \eqref{eq:def_vpvt}, and $\FpT_k^{\mu,j'}, \mu\in[m], j'\in[m']$ are defined as in \eqref{eq:def_VpVT}.

Next, construct $\FqH_{\mu}\in {\Bbb F}_q^{\eta \times 1}, \overline{\FqH}_{\mu}\in {\Bbb F}_q^{\overline{\eta} \times 1}, \mu\in [m]$ as,
\begin{align} \label{eq:def_h}
	&\FqH_{\mu}^{1\times \eta} = \Bigg[\Fqs_{\mu} \prod_{k=1}^K \Bigg( \Big( \prod_{i=1}^d\prod_{j=1}^{m} (\Fqv_{ki}^j )^{\alpha_{ki}^j} \Big) \Big( \prod_{i=1}^d\prod_{j'=1}^{m'}(\Fqt_{ki}^{\mu j'})^{\beta_{ki}^{\mu j'}} \Big) \Bigg), \mbox{s.t.}~ 0 \leq \alpha_{ki}^j, \beta_{ki}^{\mu j'} \leq N-1 \Bigg] \\
	&~~~~~~~~~\triangleq (\Fqh_\mu^1,\Fqh_\mu^2,...,\Fqh_\mu^{\eta}), ~~~~~~ \forall \mu\in[m]. \\
	&\overline{\FqH}_{\mu}^{1\times \eta} = \Bigg[ \Fqs_{\mu}  \prod_{k=1}^K \Bigg( \Big( \prod_{i=1}^d\prod_{j=1}^{m} (\Fqv_{ki}^j )^{\alpha_{ki}^j} \Big) \Big( \prod_{i=1}^d\prod_{j'=1}^{m'}(\Fqt_{ki}^{\mu j'})^{\beta_{ki}^{\mu j'}} \Big) \Bigg),  \mbox{s.t.}~ 0 \leq \alpha_{ki}^j, \beta_{ki}^{\mu j'} \leq N \Bigg] \\
	&~~~~~~~~~\triangleq (\overline{\Fqh}_\mu^1,\overline{\Fqh}_\mu^2,...,\overline{\Fqh}_\mu^{\overline{\eta}}), ~~~~~~ \forall \mu\in[m].
\end{align}
Note that we have,
\begin{align}
	\eta = N^{Kd(m+m')}, ~~~\overline{\eta} = (N+1)^{Kd(m+m')}.
\end{align}
This construction ensures that for $\mu\in[m]$, the elements of $\Fqv_{ki}^j  \FqH_\mu$ and $\Fqt_{ki}^{\mu,j'} \FqH_\mu$ are contained among the elements of $\overline{\FqH}_\mu$ for all $i\in[d], j\in[m], j'\in[m']$. Define, 
\begin{align}
	\FpH_\mu = \big[ \Fph_{\mu}^1{\bf 1}, \Fph_{\mu}^2{\bf 1}, \hdots, \Fph_{\mu}^{\eta}{\bf 1}  \big] ~~ \in {\Bbb F}_p^{n\times \eta}, ~~~~~~ {\mu} \in [m]
\end{align}
and
\begin{align}
	\overline{\FpH}_{\mu} = \big[ \overline{\Fph}_{\mu}^1{\bf 1}, \overline{\Fph}_{\mu}^2{\bf 1}, \hdots, \overline{\Fph}_{\mu}^{\overline{\eta}}{\bf 1}  \big]~~\in {\Bbb F}_p^{n\times \overline{\eta}}, ~~~~~~ {\mu} \in [m],
\end{align}
where ${\bf 1}$ denotes the $n\times 1$ vector of $1$'s. By construction, the columns of $\Fpv_{ki}^j  \FpH$ and $\Fpt_{ki}^\mu\FpH$ are subsets of the columns of $\FpH_{\mu}$, which implies that $\forall \mu\in[m], k\in[K], i\in[d],j\in[m],j'\in[m']$,
\begin{align}  \label{eq:H_subset_general}
	\langle \Fpv_{ki}^j  \FpH_\mu  \rangle_p \subset \langle \overline{\FpH}_\mu \rangle_p, && \langle \Fpt_{ki}^\mu \FpH_\mu \rangle_p \subset \langle \overline{\FpH}_\mu \rangle_p.
\end{align}

Consider the $m$ matrices,
{\small
\begin{align}
	&\big[ \FpH_1, \FpH_2,\hdots, \FpH_m,\Fpr_1 \FpH_1,\Fpr_2 \FpH_1,\hdots, \Fpr_{m'} \FpH_1\big]\in {\Bbb F}_p^{n\times (m+m')\eta}\\
	&\big[  \FpH_1, \FpH_2,\hdots, \FpH_m,\Fpr_1 \FpH_2,\Fpr_2 \FpH_2,\hdots, \Fpr_{m'} \FpH_2\big]\in {\Bbb F}_p^{n\times (m+m')\eta}\\
	&~~~~~~\vdots \notag \\
	&\big[  \FpH_1, \FpH_2,\hdots, \FpH_m,\Fpr_1 \FpH_m,\Fpr_2 \FpH_m,\hdots, \Fpr_{m'} \FpH_m\big]\in {\Bbb F}_p^{n\times (m+m')\eta}.
\end{align}
}Define event $E_n$ as,
{\small
\begin{align} \label{eq:E_n_general}
	&E_n = \bigwedge_{\mu=1}^m  \Big( \rk\big([ \FpH_1, \FpH_2,\hdots, \FpH_m,\Fpr_1 \FpH_\mu,\Fpr_2 \FpH_\mu,\hdots, \Fpr_{m'} \FpH_\mu \big]\big)  = (m+m')\eta \Big).
\end{align}
}The following lemma establishes a sufficient condition when $E_n$ holds a.a.s., which will subsequently be essential to guarantee successfully decoding by each user.
\begin{lemma} \label{lem:full_rank_general}
	If $\lim_{n\to \infty}\frac{N}{p^{[n-(m+m')\eta]}} = 0$, then $E_n$ holds a.a.s.
\end{lemma}
\begin{proof}
	See Appendix \ref{app:proof_full_rank}.
\end{proof}
Let us specify the value of $N$ as follows,
\begin{align} \label{eq:def_N_general}
	N = \left\lfloor \left( \frac{n-\sqrt{n}}{m+m'} \right)^{\frac{1}{Kd(m+m')}} \right\rfloor,
\end{align}
from which it follows that
\begin{align}
	N \leq \left( \frac{n-\sqrt{n}}{m+m'} \right)^{\frac{1}{Kd(m+m')}} 
\end{align}
and
\begin{align}
	n-(m+m')\eta \geq \sqrt{n}.
\end{align}
Therefore,
\begin{align}
	\lim_{n\to \infty} \frac{N}{p^{n-(m+m')\eta}} &\leq \lim_{n\to \infty}\frac{\left( \frac{n-\sqrt{n}}{m+m'} \right)^{\frac{1}{Kd(m+m')}}}{p^{\sqrt{n}}} \notag\\
	&= \lim_{n\to \infty}\frac{\mathcal{O}(n^{\alpha})}{p^{\sqrt{n}}} = 0
\end{align}
where $\alpha$ is independent of $n$. Since $N\geq 0$, we have $\lim_{n\to \infty}\frac{N}{p^{n-(m+m')\eta}} = 0$. Applying Lemma \ref{lem:full_rank_general}, we have that $E_n$ holds a.a.s.
Now for $\mu\in[1:m]$, let ${\bf Z}_{\mu} = ({\bf I}^{n\times n}|[\FpH_1,\hdots,\FpH_\mu, \Fpr_1 \FpH,\hdots, \Fpr_{m'} \FpH]) \in {\Bbb F}_p^{n\times (n-(m+m')\eta)}$, so that
\begin{align}
	\big[ \FpH_1,  \FpH_2,\hdots,  \FpH_{m},\Fpr_1\FpH_{\mu},\Fpr_2\FpH_{\mu}, \hdots, \Fpr_{m'}\FpH_{\mu}, {\bf Z}_{\mu} \big]
\end{align}
has full rank $n$. For compact notation, let 
\begin{align}
	\overline{\mathcal{H}} \triangleq [\overline{\FpH}_1,\overline{\FpH}_2,\hdots,\overline{\FpH}_m]^{n\times m\overline{\eta}}.
\end{align} 
Let the server broadcast ${\bf S} = \big({\bf S}_0,{\bf S}^1_{[K]},\hdots, {\bf S}^m_{[K]}\big) \in {\Bbb F}_p^{1\times [md\overline{\eta}+Km(n-(m+m')\eta)]}$, where
\begin{align}
	{\bf S}_0 = \FpX^T
	\begin{bmatrix}
		\overline{\mathcal{H}} & {\bf 0} & \hdots & {\bf 0} \\
		{\bf 0} & \overline{\mathcal{H}} & \hdots & {\bf 0} \\
		\vdots & \vdots & \ddots & \vdots \\
		{\bf 0} & {\bf 0} & \hdots & \overline{\mathcal{H}}
	\end{bmatrix} ~~ \in {\Bbb F}_p^{nd\times md\overline{\eta}}
\end{align}
and for $k\in [K], \mu\in[m]$,
\begin{align}
	{\bf S}_k^{\mu} = \FpX^T \FpV_k^\mu {\bf Z}_\mu.
\end{align}

The decoding process works as follows. User $k$ is able to compute
\begin{align} \label{eq:decode_1}
	\FpX^T \FpV_k'^{j'} \FpH_\mu, ~~~~~~ j'\in[m'], ~~ \mu \in[m] 
\end{align}
directly from its side information. Meanwhile, it is able to compute 
\begin{align} \label{eq:projection_1_general}
	\FpX^T  \FpV_k^j  \FpH_\mu, ~~~~~~ j\in[m], ~~ \mu \in[m]\end{align}
and
\begin{align} \label{eq:decode_2}
	\FpX^T \FpT_k^{\mu j'} \FpH_\mu, ~~~~~~ j'\in[m'],~~\mu \in[m]
\end{align}
from ${\bf S}_0$, since for $j\in [m], \mu \in [m]$,
{\small
\begin{align}
	\left\langle  \FpV_k^{j} \FpH_\mu \right\rangle_p = \Bigg\langle  \begin{bmatrix}
		\Fpv_{k1}^{j}   \FpH_{\mu} \\ \Fpv_{k2}^{j}   \FpH_{\mu} \\ \vdots \\ \Fpv_{kd}^{j}   \FpH_{\mu}
	\end{bmatrix} \Bigg\rangle_p 
	\subset
	\Bigg\langle  \begin{bmatrix}
		\overline{\FpH}_{\mu} & {\bf 0} & {\bf 0} & {\bf 0} \\
		{\bf 0} & \overline{\FpH}_{\mu} & {\bf 0} & {\bf 0} \\
		\vdots & \vdots & \ddots & \vdots \\
		{\bf 0} & {\bf 0} & {\bf 0} & \overline{\FpH}_{\mu} \end{bmatrix}
		\Bigg\rangle_p,
\end{align}
}and for $j'\in [m'], \mu\in [m]$,
\begin{align}
	& \Bigg\langle  \begin{bmatrix}
		\Fpt_{k1}^{\mu j'} \FpH_{\mu} \\ \Fpt_{k2}^{\mu j'} \FpH_{\mu} \\ \vdots \\ \Fpt_{kd}^{\mu j'} \FpH_{\mu}
	\end{bmatrix} \Bigg\rangle_p 
	\subset
	\Bigg\langle  \begin{bmatrix}
		\overline{\FpH}_{\mu} & {\bf 0} & {\bf 0} & {\bf 0} \\
		{\bf 0} & \overline{\FpH}_{\mu} & {\bf 0} & {\bf 0} \\
		\vdots & \vdots & \ddots & \vdots \\
		{\bf 0} & {\bf 0} & {\bf 0} & \overline{\FpH}_{\mu} \end{bmatrix}
		\Bigg\rangle_p
\end{align}
due to \eqref{eq:H_subset_general}.
Thus, User $k$ is then able to compute
\begin{align} \label{eq:projection_2_general}
	\FpX^T \FpV_k^\mu \Fpr_{j'} \FpH_\mu = \FpX^T\FpV_k'^{j'}\FpH_\mu - \FpX^T \FpT_k^{\mu j'} \FpH_\mu
\end{align}
for all $j'\in[m'],\mu\in[m]$ according to \eqref{eq:def_VpVT}.
Together with ${\bf S}_k^\mu$, User $k$ is able to compute
\begin{align}
	 ~\FpX^T  \FpV_k^j  \FpH_\mu, ~~\FpX^T \FpV_k^\mu \Fpr_{j'} \FpH_\mu ~ \mbox{and}~ {\bf S}_k^\mu,
\end{align}
for all $j\in[m], j'\in[m'], \mu\in[m]$, which are
{\small
\begin{align}
	& \FpX^T \FpV_k^{1} \big[ \FpH_1,  \FpH_2,\hdots, \FpH_m, \Fpr_1\FpH_1,\Fpr_2\FpH_1 ,\hdots, \Fpr_{m'}\FpH_1, {\bf Z}_1 \big], \\
	& \FpX^T \FpV_k^{2} \big[ \FpH_1,  \FpH_2,\hdots,  \FpH_m, \Fpr_1\FpH_2,\Fpr_2\FpH_2 ,\hdots, \Fpr_{m'}\FpH_2, {\bf Z}_2\big],\\
	& ~~~~~~ \vdots \notag \\
	& \FpX^T \FpV_k^{m} \big[ \FpH_1, \FpH_2,\hdots, \FpH_m, \Fpr_1\FpH_m,\Fpr_2\FpH_m ,\hdots, \Fpr_{m'}\FpH_m, {\bf Z}_m \big].
\end{align}
}Since $\big[ \FpH_1, \FpH_2,\hdots, \FpH_m, \Fpr_1\FpH_\mu,\Fpr_2\FpH_\mu ,\hdots, \Fpr_{m'}\FpH_\mu, {\bf Z}_\mu \big]$ is invertible (has full rank) a.a.s. for $\mu\in[m]$, User $k$ is able to compute its desired computation, $\FpX^T\FpV_k = \FpX^T\big[\FpV_k^1,\FpV_k^2,\hdots,\FpV_k^m\big]$ a.a.s.

For $q=p^n$, the cost of broadcasting each $p$-ary symbol is $1/n$ in $q$-ary units. Thus, the broadcast cost of this scheme is,
\begin{align} \label{eq:deltan_general}
	\Delta_n = \frac{md\overline{\eta}+Km\big(n-(m+m')\eta\big)}{n}.
\end{align}

By \eqref{eq:def_N_general}, we have that
\begin{align}
	\eta = N^{Kd(m+m')} \leq \frac{n-\sqrt{n}}{m+m'} \leq (N+1)^{Kd(m+m')} = \overline{\eta}
\end{align}
which implies that
\begin{align}
	\lim_{n\to \infty}\frac{\eta}{n} &= \lim_{n\to \infty}\frac{N^{Kd(m+m')}}{n} \notag\\
	&\leq \lim_{n\to \infty}\frac{n-\sqrt{n}}{n} \times \frac{1}{m+m'} \notag\\
	&= \frac{1}{m+m'}.
\end{align}
On the other hand,
\begin{align}
	\lim_{n\to \infty}\frac{\eta}{n} &=\lim_{n\to \infty} \frac{(N+1)^{Kd(m+m')}/(1+1/N)^{Kd(m+m')}}{n} \\
	&\geq \lim_{N\to \infty} \frac{1}{(1+1/N)^{Kd(m+m')}}\times\lim_{n\to \infty}\frac{n-\sqrt{n}}{n}\notag\\
	&~~~~\times \frac{1}{m+m'} \\
	& = \frac{1}{m+m'}.
\end{align}
Thus, we have that 
\begin{align} \label{eq:limeta_general}
	\lim_{n\to \infty}\frac{\eta}{n} = \frac{1}{m+m'}
\end{align}
which also implies that
\begin{align} \label{eq:limetabar_general}
	\lim_{n\to \infty}\frac{\overline{\eta}}{n} = \lim_{n\to \infty}  \frac{\eta}{n} \times\lim_{N\to \infty}(1+1/N)^{Kd(m+m')} = \frac{1}{m+m'}.
\end{align}
Combining \eqref{eq:deltan_general} with \eqref{eq:limeta_general} and \eqref{eq:limetabar_general} we have
\begin{align}
	\lim_{n\to \infty}\Delta_n = md\times \frac{1}{m+m'} + 0 = \frac{md}{m+m'}
\end{align}
since $K,m,m',d$ are independent of $n$. Thus, for any $\varepsilon >0$, $\exists n_0>0$ such that $\Delta_n \leq \frac{md}{m+m'}+\varepsilon$ for all $n\geq n_0$. Recall that the broadcast cost $\Delta_n$ is achievable if $E_n$ holds, i.e., $\Delta^*(\Lambda_n)\leq \Delta_n\leq \frac{md}{m+m'}+\epsilon$ if $n\geq n_0$ and $E_n$ holds. Now let us show that $\frac{md}{m+m'}+\epsilon$ is achievable a.a.s., by evaluating the limit in \eqref{eq:limaasach} as follows,
\begin{align}
&\lim_{n\rightarrow\infty}\mathsf{Pr}\Big(\Delta^*(\Lambda_n)\leq \frac{md}{m+m'}+\epsilon\Big)\notag\\
&\geq \lim_{n\rightarrow\infty}\mathsf{Pr}\Bigg(\Big(\Delta^*(\Lambda_n)\leq \frac{md}{m+m'}+\epsilon\Big)\land E_n\Bigg)\\
&=\lim_{n\rightarrow\infty}\mathsf{Pr}(E_n)\mathsf{Pr}\Big(\Delta^*(\Lambda_n)\leq \frac{md}{m+m'}+\epsilon~\Big|~ E_n\Big)\\
&=1
\end{align}
which implies that $\lim_{n\rightarrow\infty}\mathsf{Pr}\Big(\Delta^*(\Lambda_n)\leq \frac{md}{m+m'}+\epsilon\Big)=1$.
Since this is true for all $\varepsilon>0$, according to \eqref{eq:infdeltau} we have $\Delta_u^*\leq \inf\{\frac{md}{m+m'}+\epsilon\mid \epsilon>0\} = \frac{md}{m+m'}$. $\hfill \qed$

\section{Proof of Lemma \ref{lem:full_rank_general}} \label{app:proof_full_rank}
By Lemma 1.1.3(v) \cite{tao2012topics}, it suffices to prove that $\forall \mu \in [m]$,
\begin{align}
	&\rk\big( \big[ \FpH_1, \FpH_2,\hdots, \FpH_m, \Fpr_1\FpH_{\mu}, \Fpr_2\FpH_{\mu},\hdots, \Fpr_{m'}\FpH_{\mu} \big] \big) \notag\\
	&\aaseq (m+m')\eta.
\end{align}
Due to symmetry, without loss of generality, we will show the proof for $\mu=1$, which is
\begin{align}
	&\rk\big( \big[ \FpH_1, \FpH_2,\hdots,  \FpH_m, \Fpr_1\FpH_{1}, \Fpr_2\FpH_{1},\hdots, \Fpr_{m'}\FpH_{1} \big] \big) \notag\\
	&\aaseq (m+m')\eta.
\end{align}
Note that $\big[ \FpH_1, \FpH_2,\hdots, \FpH_m, \Fpr_1\FpH_{1}, \Fpr_2\FpH_{1},\hdots, \Fpr_{m'}\FpH_{1} \big]$ has full column rank $(m+m')\eta$ if and only if for all ${\bf c}_{i} = [c_i^1,c_i^2,\hdots, c_i^{\eta}]^T\in {\Bbb F}_p^{\eta \times 1},i\in[m]$, and ${\bf c}_i' = [c_i'^1,c_i'^2,\hdots, c_i'^{\eta}]^T\in {\Bbb F}_p^{\eta \times 1},i\in[m']$ such that $\overline{{\bf c}}^T = \big[{\bf c}_1^T, {\bf c}_2^T,\hdots, {\bf c}_m^T,{\bf c}'^T_1,{\bf c}'^T_2,\hdots, {\bf c}'^T_{m'}\big] \not= {\bf 0}^{1\times (m+m')\eta}$, 
\begin{align} \label{eq:gotozerobegin_general}
	{\bf 0}^{n\times 1} &\not= [ \FpH_1,  \FpH_2, \hdots,  \FpH_m, \Fpr_1\FpH_1,\Fpr_2\FpH_1,\hdots, \Fpr_{m'}\FpH_1]
	\overline{{\bf c}}\\
	& = \FpH_1{\bf c}_1+  \FpH_2{\bf c}_2 +\hdots+  \FpH_m{\bf c}_m  + \Fpr_1\FpH_1{\bf c}'_1 + \Fpr_2\FpH_1{\bf c}'_2 + \hdots + \Fpr_{m'}\FpH_1{\bf c}'_{m'} \\
	& = \sum_{j=1}^{\eta}c_1^j \Fph_1^{j}{\bf 1} +  \hdots + \sum_{j=1}^{\eta}c_m^j  \Fph_m^{j}{\bf 1}  +   \sum_{j=1}^{\eta}c'^j \Fpr_1\Fph_1^{j}{\bf 1} +\hdots + \sum_{j=1}^{\eta}c'^j \Fpr_{m'}\Fph_1^{j}{\bf 1} \\
	& = \Bigg( \sum_{j=1}^{\eta}c_1^j  \Fph_1^{j} +  \hdots + \sum_{j=1}^{\eta}c_m^j  \Fph_m^{j}  +   \sum_{j=1}^{\eta}c'^j \Fpr_1\Fph_1^{j} +\hdots + \sum_{j=1}^{\eta}c'^j \Fpr_{m'}\Fph_1^{j}\Bigg) {\bf 1} \\
	& \triangleq \Fpf_{\overline{\bf c}} {\bf 1}
\end{align}
where $\Fpf_{\overline{\bf c}}$ is an $n\times n$ matrix in ${\Bbb F}_p$, which has a scalar representation in ${\Bbb F}_q$ as,
\begin{align} \label{eq:f_c_general}
	&\Fqf_{\overline{\bf c}} = \sum_{j=1}^{\eta}c_1^j \Fqh_1^{j} +  \hdots + \sum_{j=1}^{\eta}c_m^j  \Fqh_m^{j} +   \sum_{j=1}^{\eta}c_1'^j \Fqr_1\Fqh_1^{j} +\hdots + \sum_{j=1}^{\eta}c_{m'}'^j \Fqr_{m'}\Fqh_1^{j} ~~ \in {\Bbb F}_q.
\end{align}
Note that since ${\Bbb F}_p$ is a sub-field of ${\Bbb F}_q$, the elements of $\overline{\bf c}$ in ${\Bbb F}_p$ are also in ${\Bbb F}_q$.
Thus, $\Fpf_{\overline{\bf c}} \in {\Bbb F}_p^{n\times 1}$ can be equivalently represented in ${\Bbb F}_q$ as the product of $\Fqf_{\overline{\bf c}}$ with the scalar representation in ${\Bbb F}_q$, of ${\bf 1}$ (the all $1$ vector in ${\Bbb F}_p$). Since the ${\Bbb F}_q$ representation of ${\bf 1}$ is not $0$, we obtain that
\begin{align}
	\Fpf_{\overline{\bf c}} {\bf 1} \not= {\bf 0}^{n\times 1} \Longleftrightarrow \Fqf_{\overline{\bf c}} \not= 0.
\end{align}
Therefore, $[ \FpH_1,  \FpH_2, \hdots,  \FpH_m, \Fpr_1\FpH_1,\Fpr_2\FpH_1,\hdots, \Fpr_{m'}\FpH_1]\in {\Bbb F}_p^{n\times (m+m')\eta}$ has full column rank if and only if,
\begin{align} \label{eq:nonzero_geneal}
	P \triangleq \prod_{\overline{\bf c} \in {\Bbb F}_p^{(m+m')\eta \times 1} \backslash \{{\bf 0} \}} \Fqf_{\overline{\bf c}} \not =0. 
\end{align}
The condition of \eqref{eq:nonzero_geneal}, which is equivalent to the event $E_n$, says that a uniformly random evaluation of the function $P(\mathcal{V})$ produces a non-zero value. We will show that this is true a.a.s. in $n$.
\begin{enumerate}
	\item Case I: At least one of $\{{\bf c}_1,{\bf c}_2,\hdots, {\bf c}_m\}$ is not ${\bf 0}^{\eta \times 1}$, then set $r_1=r_2=\hdots=r_{m'} =0$,  which implies $\Fqf_{\overline{\bf c}} = \sum_{j=1}^{\eta}c_1^j  \Fqh_1^{j}+\hdots+ \sum_{j=1}^{\eta}c_m^j  \Fqh_m^{j}$, and that  $\Fqt_{ki}^{\mu,j'} = \Fqv'^{j'}_{ki}, \forall \mu\in[m], i\in[d],j'\in[m']$ by \eqref{eq:def_vpvt}. Meanwhile, $\Fqh_{\mu}^1,\Fqh_{\mu}^2,\hdots, \Fqh_{\mu}^{\eta}$ are different monomials in the elements of $\Fqv_{ki}^j$, $\Fqv_{ki}'^{j'}$ and $\Fqs_\mu$. Moreover, since any $\Fqh_{\mu}^j$ has the factor $\Fqs_{\mu}$ but does not have the factor $\Fqs_{\mu'}$ if $\mu'\not=\mu$, it follows that $ \Fqh_1^1, \Fqh_1^2,\hdots,  \Fqh_1^{\eta},  \Fqh_2^1, \Fqh_2^2,\hdots,  \Fqh_2^{\eta},\hdots,  \Fqh_m^1, \Fqh_m^2,\hdots,  \Fqh_m^{\eta}$ ($m\eta$ in total) are different monomials. Since different monomials are linearly independent, we have that $[\Fqf_{\overline{\bf c}}]$ is a non-zero polynomial.
	\item Case II: ${\bf c}_1 = {\bf c}_2 = \hdots = {\bf c}_m = {\bf 0}$ and thus at least one of $\{ {\bf c}'_1,{\bf c}'_2,\hdots, {\bf c}'_{m'}\}$ is not ${\bf 0}^{n\times 1}$. For this case, we have $\Fqf_{\overline{\bf c}} = \sum_{j=1}^{\eta}c_1'^j \Fqr_1\Fqh_1^{j}+\hdots+ \sum_{j=1}^{\eta}c_{m'}'^j \Fqr_{m'}\Fqh_1^{j}$. From the discussion in the previous case, we know that $\Fqh_1^1,\Fqh_1^2,\hdots, \Fqh_1^\eta$ are non-zero polynomials, and none of them has a factor in $\{\Fqr_1,\Fqr_2,\hdots, \Fqr_{m'}\}$ (because otherwise letting $\Fqr_1=\Fqr_2=\hdots=\Fqr_{m'}=0$ would evaluate that polynomial to $0$, which is a contradiction). Thus, we obtain that $\Fqr_1\Fqh_1^1,\Fqr_1\Fqh_1^2,\hdots,\Fqr_1\Fqh_1^\eta, \Fqr_2\Fqh_1^1,\Fqr_2\Fqh_1^2,\hdots,\Fqr_2\Fqh_1^\eta,$ $\hdots, \Fqr_{m'}\Fqh_1^1,\Fqr_{m'}\Fqh_1^2,\hdots,\Fqr_{m'}\Fqh_1^\eta$ are linearly independent. It then follows that $[\Fqf_{\overline{\bf c}}]$ is a non-zero polynomial.
\end{enumerate}
Thus, $[\Fqf_{\overline{\bf c}}]$ is a non-zero polynomial. Since $h_i^j$ has degree not more than $(N-1)Kd(2m'+m)+1$, $[\Fqf_{\overline{\bf c}}]$ has degree not more than $(N-1)Kd(2m'+m)+2$. Therefore, $[P] \in {\Bbb F}_p[\mathcal{V}]$ is a non-zero polynomial with degree not more than $(p^{(m+m')\eta}-1)[(N-1)Kd(2m'+m)+2]$. By Schwartz-Zippel Lemma, when all the variables $\mathcal{V}$ are assigned i.i.d. uniformly chosen values in ${\Bbb F}_q$, 
\begin{align}
	&{\sf Pr}\big( P\not=0 \big) \notag\\
	&\geq 1-\frac{(p^{(m+m')\eta}-1)[(N-1)Kd(2m'+m)+2]}{q} \\
	&= 1-\frac{(p^{(m+m')\eta}-1)[(N-1)Kd(2m'+m)+2]}{p^n} \\
	& \geq 1- Kd(2m'+m)\frac{N}{p^{n-(m+m')\eta}} \\
	& \to 1
\end{align}
as $n\to \infty$ if $\lim_{n\to \infty}\frac{N}{p^{n-(m+m')\eta}} = 0$.

\section{Proof of Theorem \ref{thm:GLCBC23}} \label{app:proof_generic_K23}
For compact notation, let us define,
\begin{align}
\gamma_1 &\triangleq \big(\min\{3(m+m')-d, m+m',d\}\big)^+ \label{eq:def_gamma_1} \\
\gamma_2 &\triangleq \big( \min\{2(m+m')-d,m+m',d\} \big)^+\label{eq:def_gamma_2} \\
\gamma_3 &\triangleq \big( \min\{3(m+m')-2d,m+m',d\} \big)^+\label{eq:def_gamma_3}
\end{align}
and define $E_n\triangleq{\bf C1}\land{\bf C2}\land\cdots\land{\bf C6}$ as the event that the following conditions hold. We will show that $E_n$ holds a.a.s. The values of $\Delta_g$ for $K=1,2,3$ then follow by evaluating the capacity expression from  \cite{Yao_Jafar_3LCBC} by applying conditions ${\bf C1}$ to ${\bf C6}$ for the symmetric LCBC with $K\leq 3$.

\begin{enumerate}[wide, labelindent=0em ,labelwidth=!, leftmargin =3em, style = sameline, label={\bf C\arabic*}]
	\item $\rk(\FqV_k') = \min\{m',d\}, \hspace{0.3cm}  \forall k\in[1:K];  $ \label{con:Vkp}
	\item $\rk(\FqU_k) = \min\{m+m',d\}, \hspace{0.3cm}  \forall k\in[1:K];$ \label{con:Uk}
	\item $\rk([\FqV_i',\FqU_{ij}]) = \min\{ m'+ \gamma_2,m+m',d\}, \hspace{0.3cm} \forall i\not=j, i,j\in [1:K], ~K\geq 2; $ \label{con:VipUij}
	\item $\rk([\FqV_k' , \FqU_{123}]) = \min\{m'+\gamma_3, m+m',d\}, \hspace{0.3cm} \forall k\in[1:3],~~ K= 3; $\label{con:VkpU123}
	\item $\rk([\FqV_i', \FqU_{ij}, \FqU_{ik}]) = \min\{ m'+2\gamma_2,m+m',d\}, \hspace{0.1cm}$ for distinct $ i,j,k\in [1:3],~~K=3;$ \label{con:VipUijUik}
	\item $\rk([\FqV_i',\FqU_{i(j,k)}]) = \min\{m'+\gamma_1,m+m',d\}, \hspace{0.1cm}$ for distinct $ i,j,k\in [1:3],~~K=3,$ \label{con:VipUi(jk)}
\end{enumerate}
where
\begin{align}
	& \FqU_{ij} \triangleq \FqU_{i} \cap \FqU_{j}, ~~~~~~~~\forall  i,j\in [1:K], i\not=j,\\
	& \FqU_{123} \triangleq \FqU_{1} \cap \FqU_{2} \cap \FqU_{3},~~~~~~ \mbox{if $K=3$}.
\end{align}
By Lemma 1.1.3(v) \cite{tao2012topics}, we then show that $E_n$ holds a.a.s. by showing that each of the conditions ${\bf C1}$ to ${\bf C6}$ holds a.a.s.
\subsection{Conditions $\ref{con:Vkp}, \ref{con:Uk}, \ref{con:VipUij}$ and $\ref{con:VipUi(jk)}$}
In Lemma \ref{lem:full_rank_matrix}, let $M'=[~]$, $M=\FqV_k', \FqU_k$, respectively. We obtain that Conditions \ref{con:Vkp} and \ref{con:Uk} hold a.a.s.  Then let $M = [\FqU_j,\FqV_i']$, $M = [\FqU_i,\FqU_j]$, respectively. We obtain that
\begin{align}
	&\rk([\FqU_j,\FqV_i']) \aaseq \min\{m+2m',d\}, \\
	&\rk([\FqU_i,\FqU_j]) \aaseq \min\{2(m+m'),d\}.
\end{align}
Since
\begin{align}
	& \rk([\FqU_i, \FqU_j])-\rk([\FqU_j, \FqV_i']) \notag \\
	&= \rk(\FqU_i) + \rk(\FqU_j) - \rk(\FqU_{ij}) - [\rk(\FqU_j)+\rk(\FqV_i')-\rk(\FqU_j \cap \FqV_i')] \\
	&= \rk(\FqU_i) + \rk(\FqU_j) - \rk(\FqU_{ij})  - [\rk(\FqU_j)+\rk(\FqV_i')-\rk(\FqU_j\cap \FqU_i \cap \FqV_i')] \\
	&= \rk(\FqU_i) - [\rk(\FqU_{ij}) + \rk(\FqV_i')-\rk(\FqU_{ij} \cap \FqV_i')] \\
	& = \rk(\FqU_i) - \rk([\FqV_i',\FqU_{ij}]),
\end{align}
it follows that
{\small
\begin{align}
	&\rk([\FqV_i',\FqU_{ij}]) \notag\\
	& \aaseq \min\{m+m',d\}-\min\{2(m+m'),d\}+\min\{m+2m',d\}\\
	& = \min\{m'+\gamma_2,m+m',d\}
\end{align}
}which proves the result for \ref{con:VipUij}.
Next, let $M = [\FqU_j,\FqU_k,\FqV_i']$, $M = [\FqU_i, \FqU_j,\FqU_k]$, respectively. We obtain that
\begin{align}
	&\rk([\FqU_j,\FqU_k,\FqV_i']) \aaseq \min\{2m+3m',d\},\\
	&\rk([\FqU_i,\FqU_j,\FqU_k]) \aaseq \min\{3(m+m'),d\}.
\end{align}
Since
\begin{align}
	& \rk([\FqU_i, \FqU_j,\FqU_k])-\rk([\FqU_j, \FqU_k, \FqV_i']) \notag \\
	&= \rk([\FqU_j,\FqU_k]) + \rk(\FqU_i) - \rk(\FqU_{i(j,k)}) - [\rk([\FqU_j,\FqU_k])+\rk(\FqV_i')-\rk([\FqU_j,\FqU_k] \cap \FqV_i')] \\
	&= \rk([\FqU_j,\FqU_k]) + \rk(\FqU_i) - \rk(\FqU_{i(j,k)}) -  [\rk([\FqU_j,\FqU_k])+\rk(\FqV_i') -\rk([\FqU_j,\FqU_k] \cap \FqU_i \cap  \FqV_i')] \\
	&= \rk(\FqU_i) - [\rk(\FqU_{i(j,k)}) + \rk(\FqV_i')-\rk(\FqU_{i(j,k)} \cap \FqV_i')] \\
	& = \rk(\FqU_i) - \rk([\FqV_i',\FqU_{i(j,k)}]),
\end{align}
it follows that
\begin{align}
	&\rk([\FqV_i',\FqU_{i(j,k)}]) \notag\\
	& \aaseq \min\{m+m',d\}-\min\{3(m+m'),d\}\notag\\
	&~~~~~~~~~~~~~~~~~~~~~~~~+\min\{2m+3m',d\}\\
	& = \min\{m'+\gamma_1,m+m',d\}
\end{align}
which proves the result for \ref{con:VipUi(jk)}.

\subsection{Condition $\ref{con:VkpU123}$}
To see that \ref{con:VkpU123} holds a.a.s., we need the following lemma.
\begin{lemma} \label{lem:full rank 2}
Let $A\in \mathbb{F}_{p^n}^{d\times \mu'}$, $B,C \in \mathbb{F}_{p^n}^{d\times \mu}$, such that $\mu'\leq \mu$. Denote
\begin{align}
	M = \begin{bmatrix}
		A & B & {\bf 0} \\
		A & {\bf 0} & C
	\end{bmatrix}.
\end{align}
If the elements of $A,B,C$ are chosen i.i.d uniform, then $M$ has full rank $\min\{2d, \mu'+2\mu\}$ a.a.s.
\end{lemma}
\begin{proof}
Consider the following cases.
\begin{enumerate}
	\item If $2d \geq \mu'+2\mu$, which implies $\mu'\leq \mu \leq d$ and thus $\mu'+\mu\leq 2d$. Let $P_1 = \det([M,Z])$, where $Z\in {\Bbb F}_q^{2d\times (2d-(\mu'+2\mu))}$, be a polynomial in the elements of $A,B,C,Z$. To verify that it is not the zero polynomial, consider the following realizations of $A,B,C,Z$ for which $P_1\neq 0$. Let $A = {\bf I}^{d\times d}_{[1:\mu']}$, $B = C = {\bf I}^{d\times d}_{[\mu'+1:\mu]}$. Let $v_a\in {\Bbb F}_q^{\mu' \times 1}$, $v_b,v_c\in {\Bbb F}_q^{\mu \times 1}$. Then $M[v_a,v_b,v_c]^T=\mathbf{0}\implies ([A,B][v_a,v_b]^T=\mathbf{0}) \land ([A,C][v_a,v_c]^T=\mathbf{0}) \implies v_a=v_b=v_c = \mathbf{0}$. Therefore, $M$ has independent columns. Let $Z=({\bf I}^{d\times d}|M) \in {\Bbb F}_q^{2d \times [2d-(\mu'+2\mu)]}$, so that $[M,Z]$ has full rank, which yields a non-zero evaluation for $P_1$. Now, since $P_1$ is not the zero-polynomial, if the elements of $A, B, C, Z$ are chosen i.i.d uniform, then by Schwartz-Zippel Lemma, we obtain that as $n \rightarrow\infty$, the evaluation of $P_1$ is almost surely non-zero, which implies that $\rk(M) \aaseq \mu'+2\mu$.
	\item If $2\mu\leq 2d < \mu'+2\mu$, then we have $2\mu-d>d-\mu'\geq d-\mu \geq 0$. Let $P_2 = \det([M^T,Z^T]^T)$, where $Z\in {\Bbb F}_q^{(\mu'+2\mu-2d)\times (\mu'+2\mu)}$. Let $I_1 = {\bf I}^{d\times d}_{[1:2\mu-d]}$, $I_2 = {\bf I}^{d\times d}_{[2\mu-d+1:\mu]}$, $I_3 = {\bf I}^{d\times d}_{[\mu+1:d]}$, $B = [I_1,I_2], C = [I_1,I_3], A_0 = [I_2,I_3]$. Then we have $\langle A_0 \cap B \cap C \rangle  = \{\mathbf{0}\}$, which implies that the following matrix has full rank.
		\begin{align}
			M_0^{2d\times 2d} = \begin{bmatrix}
				A_0 & B & {\bf 0}\\
				A_0 & {\bf 0} & C
			\end{bmatrix}.
		\end{align}
		To see this, let $v_a\in {\Bbb F}_q^{(2d-2\mu) \times 1}$, $v_b,v_c\in {\Bbb F}_q^{\mu \times 1}$. Then $M_0v = M_0[v_a,v_b,v_c]^T = 0\implies A_0v_a=-Bv_b = -Cv_c \in \langle A_0 \cap B \cap C \rangle$. Since $A_0,B,C$ have only trivial intersection, the only solution for $v$ is $0$. Letting $A = [{\bf 0}^{d\times (\mu'+2\mu-2d)}, A_0]$, we obtain that $M$ has $2d$ linearly independent rows. Let $Z^T = ({\bf I}^{d\times d}|M^T)\in {\Bbb F}_q^{(\mu'+2\mu-2d)\times (\mu'+2\mu)}$, which is constituted by $(\mu'+2\mu-2d)$ rows of ${\bf I}^{(\mu'+2\mu)\times (\mu'+2\mu)}$, so that $\big[M^T,Z^T\big]^T$ has full rank. Therefore, $P_2$ is not the zero polynomial. By Schwartz-Zippel Lemma, we obtain that for i.i.d. uniform $A,B,C,Z$, as $n\rightarrow\infty$, $P_2$ will evaluate to a non-zero value almost surely, which implies that $\rk(M) \aaseq 2d$.
				\item If $d<\mu$, then by Lemma \ref{lem:full_rank_matrix}, we have that $\rk(B)=\rk(C) = d \implies \rk(M)\geq 2d$ holds asymptotically almost surely. Since $M$ has $2d$ rows, we conclude that $\rk(M) \aaseq 2d$.
\end{enumerate}
\end{proof}

In Lemma \ref{lem:full rank 2} let
\begin{align}
	M = \begin{bmatrix}
		\FqU_1 & \FqU_2 & {\bf 0} \\
		\FqU_1 & {\bf 0} & \FqU_3
	\end{bmatrix}.
\end{align}
We obtain that $\rk(M) \aaseq \min\{3(m+m'),2d\}$. It then follows from  \cite{tian2002dimension} that,
\begin{align}
	&\rk(\FqU_{123}) = \rk(\FqU_1 \cap \FqU_2 \cap \FqU_3) \\
	&=\rk(\FqU_1) + \rk(\FqU_2) + \rk(\FqU_3) - \rk(M) \\
	&\aaseq 3\min\{(m+m'),d\} - \min\{3(m+m'),2d\}\\
	&= \gamma_3,
\end{align}
Then applying Lemma \ref{lem:full rank 2} to
\begin{align}
	M = \begin{bmatrix}
		\FqV_1' & \FqU_2 & {\bf 0} \\
		\FqV_1' & {\bf 0} & \FqU_3
	\end{bmatrix},
\end{align}
we obtain that $\rk(M') \aaseq \min\{2m+3m',2d\}$. By \cite{tian2002dimension},
\begin{align}
	&\rk(\FqV_1' \cap \FqU_{123}) = \rk(\FqV_1' \cap \FqU_2 \cap \FqU_3) \\
	&=\rk(\FqV_1') + \rk(\FqU_2) + \rk(\FqU_3) - \rk(M')
\end{align}
Therefore,
\begin{align}
	&\rk([\FqV_1' , \FqU_{123}])\\
	&= \rk(\FqV_1') + \rk(\FqU_{123}) - \rk(\FqV_1' \cap \FqU_{123})\\
	&= \rk(\FqU_{123}) - \rk(\FqU_2) - \rk(\FqU_3) + \rk(M')\\
	&\aaseq \gamma_3 - 2\min\{m+m',d\} + \min\{2m+3m',2d\}\\
	& = \min\{m'+\gamma_3,m+m',d\}
\end{align}

\subsection{Condition $\ref{con:VipUijUik}$}

Finally let us prove for \ref{con:VipUijUik}. In Lemma \ref{lem:full_rank_matrix}, let $M = [\FqU_i,\FqU_j]$. We have
\begin{align} \label{eq:rkUiUj}
	\rk([\FqU_i,\FqU_j]) \aaseq \min\{2(m+m'),d\},
\end{align}
By the result for \ref{con:Uk} and (\ref{eq:rkUiUj}), we have that for distinct $i,j\in[1:3]$,
\begin{align} \label{eq:rkUij}
	&\rk(\FqU_{ij}) = \rk(\FqU_{i})+ \rk(\FqU_{j}) - \rk([\FqU_{i}, \FqU_{j}]) \\
	& \aaseq 2\min\{m+m',d\}-\min\{2(m+m'),d\}\\
	& = \big( \min\{2(m+m')-d,m+m',d\} \big)^+ \\
	& = \gamma_2
\end{align}
To prove that \ref{con:VipUijUik} holds asymptotically almost surely, due to symmetry, it suffices to prove for $i=1,j=2,k=3$. Let us consider the following cases. 

If $\gamma_2=0$, then we  almost surely have $\rk([\FqV_{1}' , \FqU_{12} , \FqU_{13}]) = \rk(\FqV_{1}')=\min\{m',d\} = \min\{m'+2\gamma_2,m+m',d\}$, as desired.
Otherwise, let us consider two sub-cases.

\begin{enumerate}
	\item $\gamma_2 = \min\{m+m',d\} > 0$. By \ref{con:Uk},
\begin{align}
	\rk([\FqV_1',\FqU_{12},\FqU_{13}]) \leq \rk(\FqU_1) \aaseq \min\{m+m',d\}
\end{align}
On the other hand, by \ref{con:VipUij}, 
\begin{align}
	\rk([\FqV_1',\FqU_{12},\FqU_{13}]) &\geq \rk([\FqV_1',\FqU_{12}]) \notag\\
	&\aaseq \min\{m+\gamma_2,m+m',d\}\\
	& = \min\{m+m',d\}
\end{align}
This implies that 
\begin{align}
	\rk([\FqV_1',\FqU_{12},\FqU_{13}]) &\aaseq \min\{m+m',d\} \notag\\
	&= \min\{m'+2\gamma_2,m+m',d\} 
\end{align}
as desired.
\item $\gamma_2 = 2(m+m')-d > 0$. This implies that $m+m'\leq d < 2(m+m')$.
Denote $A = \FqU_1$, $B = \FqU_2$, $C = \FqU_3$. Then $\FqV_1' = A_{[1:m']}$. Denote $a=d-(m+m')$, and let
\begin{align}
	D& = [A,{\bf 0}^{d\times a}][A,B_{[1:a]}]^* B_{[a+1:m+m']}, \\
	E& = [A,{\bf 0}^{d\times a}][A,C_{[1:a]}]^* C_{[a+1:m+m']}.
\end{align}
Recall that $M^*$ denotes the adjoint matrix of the square matrix $M$.
We claim that $\langle D \rangle \subset \langle \FqU_{12} \rangle$ and $\langle E \rangle \subset \langle \FqU_{13} \rangle$. It is obvious that $\langle D \rangle \subset \langle\FqU_{1}\rangle$ and $\langle E \rangle \subset \langle\FqU_{1}\rangle$. To see that $\langle D \rangle \subset \langle\FqU_{2}\rangle$ and $\langle E \rangle \subset \langle\FqU_{3}\rangle$, note that
\begin{align}
	&D+[{\bf 0}^{d\times (m+m')},B_{[1:a]}][A,B_{[1:a]}]^* B_{[a+1:m+m']}\\
	&=[A,B_{[1:a]}][A,B_{[1:a]}]^* B_{[a+1:m+m']}\\
	&=\det([A,B_{[1:a]}]) B_{[a+1:m+m']},
\end{align}
where we used the fact that for any square matrix $M$, the product of $M$ and its adjoint $M^*$ is equal to the product of the determinant of $M$ and the identity matrix. Thus, every column of $D$ is a linear combination of the columns of $B=\FqU_2$, which implies that $\langle D\rangle \subset \langle\FqU_{2}\rangle$. Similarly, 
\begin{align}
	&E+[{\bf 0}^{d\times (m+m')},C_{[1:a]}][A,C_{[1:a]}]^* C_{[a+1:m+m']}\\
	&=\det([A,C_{[1:a]}]) C_{[a+1:m+m']},
\end{align}
which implies that $\langle E \rangle \subset \langle\FqU_{3}\rangle$. 

Let us now show that $\rk([A_{[1:m']},D,E]) \geq \min\{m'+2\gamma_2, m+m'\}$ holds asymptotically almost surely. Denote $b = \min\{m,\gamma_2\}$ and $c = \min\{m,2\gamma_2\}$. Let $Z \in {\Bbb F}_q^{d\times (d-m'-c)}$.
The following determinant is a polynomial in the elements of $(A,B,C,Z)$. 
\begin{align}
	P = \det ([A_{[1:m']},D_{[1:b]},E_{[1:(c-b)]}, Z]).
\end{align}
To prove that this is  not the zero polynomial, let
\begin{align}
	I_1& = {\bf I}^{d\times d}_{[1:m']}, &&I_2 = {\bf I}^{d\times d}_{[m'+1:m'+b]}, \notag \\
	I_3& = {\bf I}^{d\times d}_{[m'+b+1:m'+c]}, &&I_4 = {\bf I}^{d\times d}_{[m'+c+1:m'+m]}, \notag\\
	I_5& = {\bf I}^{d\times d}_{[m+m'+1:d]},
\end{align}
and then consider the following evaluation,
\begin{align}
&A = [I_1,I_2,I_3,I_4] \implies A' = I_1, \\
&B = [I_5, I_2, {\bf 0}^{d\times (\gamma_2-b)}],\\
&C = [I_5, I_3, {\bf 0}^{d\times (\gamma_2-c+b)}],\\
&Z = [I_4,I_5].
\end{align}
Note that $[A,B_{[1:a]}]=[A,C_{[1:a]}] = [I_1,I_2,I_3,I_4,I_5] ={\bf I}^{d\times d}$. Thus,
\begin{align}
	D&=[A,\mathbf{0}]B_{[a+1:m+m']} \\
	&= [I_1,I_2,I_3,I_4,\mathbf{0}^{d\times a}][I_2,{\bf 0}^{d\times (\gamma_2-b)}]\\
	&= [I_2,{\bf 0}^{d\times (\gamma_2-b)}],
\end{align}
which implies that $D_{[1:b]} = I_2.$ Similarly, 
\begin{align}
	E&= [A,\mathbf{0}]C_{[a+1:m+m']} \\
	&= [I_1,I_2,I_3,I_4,\mathbf{0}][I_3, {\bf 0}^{d\times (\gamma_2-c+b)}]\\
	&= [I_3,{\bf 0}^{d\times (\gamma_2-c+b)}],
\end{align}
which implies that $E_{[1:(c-b)]} = I_3$.

Therefore $P=\det([I_1,I_2,I_3,I_4,I_5])=\det({\bf I}^{d\times d})=1\not=0$. Since $P$ is not the zero polynomial, by Schwartz-Zippel Lemma, we obtain that $\rk([A',D,E]) \geq m'+b+(c-b) = m'+c = \min\{m'+2\gamma_2,m+m'\}$  holds asymptotically almost surely. This implies that $\rk([\FqV_1' , \FqU_{12} , \FqU_{13}]) \aasgeq \min\{m'+2\gamma_2,m+m'\}$. Since $\rk(\FqV_1') \aaseq m'$, $\rk(\FqU_{12}) = \rk(\FqU_{13}) \aaseq \gamma_2$, from the result for \ref{con:Vkp} and (\ref{eq:rkUij}) and note that $\langle [\FqV_1' , \FqU_{12} , \FqU_{13}] \rangle \subset \langle \FqU_1 \rangle$, we conclude that $\rk([\FqV_1' , \FqU_{12}, \FqU_{13}]) \aaseq \min\{m'+2\gamma_2,m+m'\}=  \min\{m'+2\gamma_2,m+m',d\}$.
\end{enumerate}
Therefore, ${\bf Ci}, {\bf i}\in [6]$ holds asymptotically almost surely. We conclude that $E_n$ holds asymptotically almost surely. The proof is completed by evaluating (15) of \cite{Yao_Jafar_3LCBC} with conditions ${\bf C1}$ to ${\bf C6}$ to get $\Delta^*_g$ for the symmetric GLCBC with $K\leq 3$.

\section{Proof of Lemma \ref{lem:full_rank_matrix}}\label{sec:proof:lem:full_rank_matrix}
	If $d \geq \mu'+\mu$, let $Z\in {\Bbb F}_{p^n}^{d\times (d-\mu'-\mu)}$. Then $P = \det([M',M,Z])$ is a non-zero polynomial in the elements of $M$ and $Z$. To see this, let $[M,Z] = {\bf I}^{d\times d}|M'$, which will then yield that $\det ([M',M,Z]) \not=0$ since $\langle [M', ({\bf I}^{d\times d}|M')] \rangle = \langle {\bf I}^{d\times d} \rangle$. By Schwartz-Zippel Lemma, if the elements of $M$ and $Z$ are chosen i.i.d uniform, $\mbox{Pr}(P\not=0) \geq 1-\frac{\mbox{degree of }P}{p^n}\geq 1-\frac{d}{p^n}\to 1$ as $n\rightarrow\infty$, which implies that the probability that $[M',M]$ has full rank $\mu'+\mu$ goes to $1$ as $n\rightarrow\infty$. If $d<\mu'+\mu$, denote by $M_1$ the first $d-\mu'$ columns of $M$. It suffices to show that $\rk([M',M_1]) \aaseq d$. Note that $P = \det([M',M_1])$ is a non-zero polynomial in the elements of $M_1$, and thus $M$. To see this, let $M_1 = {\bf I}^{d\times d}|M'$, which will then similarly yield that $\det ([M',M_1]) \not=0$. By Schwartz-Zippel Lemma, if the elements of $M_1$ are chosen i.i.d uniform, $\mbox{Pr}(P\not=0) \geq 1-\frac{\mbox{degree of }P}{p^n}\geq 1-\frac{d}{p^n}\to 1$ as $n\rightarrow\infty$, which implies that $\rk([M',M_1]) \aaseq d$ as desired.$\hfill\square$

\bibliographystyle{IEEEtran}
\bibliography{../../KLCBC.bib}

\end{document}